%% file: Multiplicity_new.tex
\documentclass[reqno,12pt,letterpaper]{amsart}
\usepackage{amsmath,amssymb,amsthm,graphicx,mathrsfs,url}
\usepackage[usenames,dvipsnames]{color}
\usepackage[colorlinks=true,linkcolor=Red,citecolor=Green]{hyperref}
\usepackage{amsxtra}
\usepackage{placeins}
\usepackage{wasysym} 
\usepackage{graphicx}
\usepackage{subcaption}
\def\arXiv#1{\href{http://arxiv.org/abs/#1}{arXiv:#1}}
\usepackage{array}
\newcolumntype{P}[1]{>{\centering\arraybackslash}m{#1}}

\def\blue#1{\textcolor{blue}{#1}}
\def\red#1{\textcolor{red}{#1}}

\def\?[#1]{\textbf{[#1]}\marginpar{\Large{\textbf{??}}}}

\def\smallsection#1{\smallskip\noindent\textbf{#1}.}
\let\epsilon=\varepsilon 

\newcommand{\RR}{{\mathbb R}}
\newcommand{\NN}{{\mathbb N}}
\newcommand{\CC}{{\mathbb C}}

\newcommand{\ZZ}{{\mathbb Z}}

\newtheorem{theo}{Theorem}
\newtheorem{prop}{Proposition}[section]

\newtheorem{lemm}[prop]{Lemma}

\numberwithin{equation}{section}

\DeclareMathOperator{\Spec}{Spec}

\DeclareMathOperator{\Hom}{Hom}

\let\Im=\Imag

\let\Re=\Real

\DeclareMathOperator{\supp}{supp}

\DeclareMathOperator{\tr}{tr}

\usepackage{scalerel}

\newcommand\reallywidehat[1]{\arraycolsep=0pt\relax%
\begin{array}{c}
\stretchto{
  \scaleto{
    \scalerel*[\widthof{\ensuremath{#1}}]{\kern-.5pt\wedge\kern-.5pt}
    {\rule[-\textheight/2]{1ex}{\textheight}} 
  }{\textheight} %
}{0.5ex}\\           
#1\\                 
\rule{-1ex}{0ex}
\end{array}
}

\begin{document}

\title{
Degenerate flat bands in twisted bilayer graphene} 

\author{Simon Becker}
\email{simon.becker@math.ethz.ch}
\address{ETH Zurich, 
Institute for Mathematical Research, 8092 Zurich, CH.}

\author{Tristan Humbert}
\email{tristan.humbert@ens.psl.eu}
\address{Department of Mathematics, University of California,
Berkeley, CA 94720, USA.}

\author{Maciej Zworski}
\email{zworski@math.berkeley.edu}
\address{Department of Mathematics, University of California,
Berkeley, CA 94720, USA.}

\begin{abstract}
We prove that in the chiral limit of the Bistritzer--MacDonald Hamiltonian, there exist magic angles at which the Hamiltonian exhibits flat bands of multiplicity four instead of two. We analyse the structure of Bloch functions associated with the bands of arbitrary multiplicity, compute the corresponding Chern number to be $ -1 $, and show that there exist infinitely many degenerate magic angles for a generic choice of tunnelling potential, including the Bistritzer--MacDonald potential. 
Moreover, we demonstrate for generic tunnelling potentials  flat bands have only twofold or fourfold multiplicities. 
\end{abstract}

\maketitle 

\section{Introduction and statement of results}
\label{s:intr}

Twisted bilayer graphene is a material consisting of two stacked graphene layers which are twisted with respect to each other by an angle $\theta.$ It has been predicted theoretically \cite{BM11}  that at a certain angle, the bands at zero energy become flat and strongly correlated electron effects dominate. This has then been experimentally confirmed that at this {\em magic angle}, the material exhibits phenomena such as superconductivity and a quantum Hall effect without external magnetic fields \cite{Cao,Ser19,Yan}.
Theoretically \cite{magic} more magic angles have been expected. 
Contrary to common beliefs, we demonstrate that flat bands of higher multiplicity are common in this model of bilayer graphene. Higher multiplicity bands have recently also been theoretically observed in models of twisted trilayer graphene \cite{PT23,hel_tri}. We verify numerically that the presence of higher degenerate (almost flat) bands close to zero energy is also valid for the full (not just 
 chiral) model, see Figure~\ref{fig:beyond_chiral}.
 
The model we consider is based on the Bistritzer-MacDonald Hamiltonian \cite{BM11,CGG,Wa22} 
and its chiral limit of Tarnopolsky--Kruchkov--Vishwanath \cite{magic}: 
\begin{equation}
\label{eq:Hamiltonian}
  H(\alpha)=\begin{pmatrix} 0 &   D (\alpha)^*\\   D(\alpha) & 0 \end{pmatrix} \text{ with }   D(\alpha) = \begin{pmatrix} 2D_{\bar z }& \alpha U_+ (z ) \\ \alpha U_-(z) & 2D_{\bar z } \end{pmatrix} 
\end{equation}
where we use complex coordinates $z \in \CC$  and the parameter $\alpha$ is proportional to the inverse relative twisting angle. Here, we write $z=x+iy$ for real $x,y$ then $\partial_{z}:=1/2(\partial_x -i\partial_y),$ where $\partial_x$ denotes differentiation with respect to $x$, and $D_{z} =-i\partial_z.$ Similarly, $\partial_{\bar z}:=\frac12 (\partial_x +i\partial_y)$ and $D_{\bar z} =-i\partial_{\bar z}.$ Clearly, $H(\alpha):H^1(\CC;\CC^4) \subset L^2(\CC;\CC^4) \to  L^2(\CC;\CC^4)$ is self-adjoint and $D(\alpha):H^1(\CC;\CC^2) \subset L^2(\CC;\CC^2) \to L^2(\CC;\CC^2).$
With $\omega = e^{2\pi i /3}$ and  $\mathbf a = {4\pi i} (a_1 \omega + a_2 \bar \omega)$, $ a_j 
\in \mathbb Z $, we assume that 
\begin{equation}
\label{eq:symmetries}
U_\pm (z  + \mathbf a ) = \omega^{\mp (a_1+a_2 ) } U (z ), \quad U_\pm  (\omega z ) = \omega U_\pm (z ).\end{equation}
In the physics literature, the following choice is made \cite{BM11,magic}
\begin{equation}
\label{eq:defU}
U_+ ( z )= U ( z ) , \ \ \ U_- ( z ) = U ( - z), \ \ \ \overline{U( \bar z ) } = U ( z ) ,
\end{equation}
see \eqref{eq:potential} for concrete examples.

Floquet theory for the Hamiltonian \eqref{eq:Hamiltonian} is based on 
moir\'e translations: 
\begin{equation}
\label{eq:defLag}  
 {\mathscr L}_{\mathbf a  } u(z)  :=  
\begin{pmatrix} \omega^{ - ( a_1 + a_2 ) }  & 0  \\
0 & \omega^{a_1 + a_2} 
\end{pmatrix}  u ( z + \mathbf a  )   ,   \ \ \ \mathbf a = {4\pi i} (a_1 \omega + a_2 \bar \omega) .
  \end{equation}
The action is extended diagonally to $ \mathbb C^4 = \mathbb C^2 \times \mathbb C^2 $-valued functions
and we 
$ \mathscr L_{\mathbf a}  H ( \alpha ) = H ( \alpha ) \mathscr L_{\mathbf a }  $. 

The Floquet spectrum is given by 
\begin{equation}
\begin{gathered} 
\label{eq:FL_ev}  H ( \alpha ) u = E u , \ \ 
u \in H^1_{k} 
\ \ \ H^s_k := L^2_k \cap  H^s_{\rm{loc} }, \\
L^2_{k}  := 
\{ u = L^2_{\rm{loc}} ( \mathbb C ; \mathbb C^4 ) : \mathscr L_\mathbf a u = 
e^{i\langle k,\mathbf a\rangle} u \}  ,\quad \langle z, w \rangle:=\Re(z \bar w), \text{ and } k\in \CC.
\end{gathered}
 \end{equation}
The spectrum is discrete and symmetric with respect to the origin and we index it as follows
(with $ \mathbb Z^* :=  \mathbb Z \setminus \{ 0 \} $)
\begin{equation}
\label{eq:eigs} 
\begin{gathered} \{ E_{ j } ( \alpha, k ) \}_{ j \in  \mathbb Z^* } ,  \ \ \  E_{j } ( \alpha, k ) 
= - E_{-j} ( \alpha , k ) , \\ 0 \leq E_1 ( \alpha, k ) \leq E_2 ( \alpha, k ) \leq \cdots , \ \ \  E_1 ( \alpha, K) =  E_1 ( \alpha, - K ) = 0 , 
\end{gathered} \end{equation}
see \cite[\S 2]{bhz2} for more details. The points $ K, -K  $, $ K = i $, 
are called the {\em Dirac points} and are typically denoted by $ K $ and $ K' $ in the physics
literature.

Using the conjugation, 
\[ H_k ( \alpha )  := e^{ i \langle z, k \rangle } H( \alpha ) e^{ - i \langle z, k \rangle },\] we can equivalently study operators $H_k(\alpha)$ on $L^2_0.$

\medskip

\noindent
{\bf Definition} (Magic angles and their multiplicities). {\em A value of $ \alpha $ in \eqref{eq:Hamiltonian} is called magical if 
$ H ( \alpha ) $ has a flat band at zero
\[  E_1 ( \alpha, k ) \equiv 0 , \ \ k \in \mathbb C . \]
The set of magic $ \alpha$'s is denoted by $ \mathcal A $ or $ \mathcal A (U ) $
if we specify the dependence on the potential. 
The multiplicity of a magic $ \alpha $ is defined as 
\begin{equation} 
\label{eq:defm} m ( \alpha ) = m_U ( \alpha ) = \min \{ j > 0 : \max_k  E_{j+1} ( \alpha, k ) > 0  \}. 
\end{equation}
Magic angles are (up to physical constants) reciprocals of $ \alpha \in \mathcal A $.}

To formulate our first result on the multiplicity of flat bands, we need the following definition of multiplicity of zeros for $ \mathbf u \in \ker ( D ( \alpha) + K )  $:
\begin{equation}
\label{eq:defM}    M_{\mathbf u }  ( z_0 ) := \max \{ m :  [ \partial_z^{m-1} \mathbf u ] (z_0 ) = 0 \} , 
\end{equation}
with the convention that $ M_{\mathbf u } ( z_0 ) = 0 $ if $ \mathbf u ( z_0 ) \neq 0 $. As in 
\cite[Lemma 3.2]{bhz2} we easily see that this is equivalent to $ \mathbf u ( z ) = ( z - z_0 )^m \mathbf u_0 ( z ) $, where $ \mathbf u_0 $ is smooth near $ z_0 $. 

\begin{theo}[Zeros and multiplicities]
\label{p:zeros}
For the Hamiltonian \eqref{eq:Hamiltonian} with potentials satisfying \eqref{eq:symmetries}
and \eqref{eq:defU}, 
let $ \mathbf u ( \alpha ) := \mathbf u_K (\alpha ) \in \ker_{ H_0^1 } ( D ( \alpha ) + K ) $ be a family of protected states (see \cite[Theorem 1]{survey} and references given there). Then 
for $ \alpha \in \mathbb C $
\begin{equation}
\label{eq:Mua}  
 \sum_{ z \in \mathbb C/\Lambda } M_{\mathbf u ( \alpha ) } ( z ) = m ( \alpha ) .  \end{equation}
For  $ \alpha \in \mathcal A  $ and $ k \in \mathbb C/\Lambda^*  $, 
\[   \dim \ker_{H^1_0} H_k(\alpha)= 2 \dim \ker_{H^1_0}(D(\alpha)+k) =2 m(\alpha). \]
In particular, the zero-energy flat bands of the Hamiltonian are always spectrally gapped from the rest of the spectrum.
Moreover, 
\eqref{eq:Mua} holds for 
any $\mathbf{u}(\alpha) \in \ker_{H^1_0}(D(\alpha)+k) \setminus \{0\}$, $k \in \mathbb{C} / \Lambda^*$.
\end{theo}
To formulate the next result we define  
\[  L^2_{0, p } := \{ u \in L^2_0 : u ( \omega z ) = \bar \omega^p u ( z ) \} ,  \ \ \  p \in \mathbb Z_3 . \]
Using these spaces we have the following rigidity result for simple and double-degenerate magic angles: 

  \begin{theo}[Rigidity]
\label{theo:rigidity}
Under the assumptions of Theorem \ref{p:zeros} and with 
the definition of multiplicity \eqref{eq:defm}, 
\begin{equation}
\label{eq:imply} \begin{split}
& m ( \alpha ) = 1 \ \Longrightarrow \ \dim \ker_{ L^2_{0,2} }  D ( \alpha ) = 1 ,\\
& m (\alpha ) = 2 \ \Longrightarrow \ \dim \ker_{L^2_{0,0} } D ( \alpha ) = 
\dim \ker_{ L^2_{0,1} }  D ( \alpha ) = 1. 
\end{split} \end{equation}
Moreover, for all $  \alpha \in \mathcal A $, 
\begin{equation}
\label{eq:Nikita}
 m( \alpha ) \not \equiv 0 \! \! \! \mod 3 . 
\end{equation}
\end{theo}
The first implication in \eqref{eq:imply} is included in  \cite[Theorem 2]{bhz2}. 
The multiplicity statement \eqref{eq:Nikita} which is a consequence of the proof of \eqref{eq:imply} 
was added because of a recent paper of Iugov--Nekrasov~\cite{IN25}
where it was obtained using  different methods.

To prove the existence of magic $ \alpha$'s of higher multiplicities, we perform trace computations
first used to show that $ \mathcal A $ is non-empty \cite{beta} and then that $ | \mathcal A | = \infty $
\cite{bhz1}. The traces here refer to $ \tr T_k^{2 p } $ where $ T_k $ is a Birman--Schwinger
operator with spectrum given by $ \{ 1/\alpha: \alpha \in \mathcal A \} $ - see \S \ref{s:propH}, 
\cite[Theorem 3]{beta}, \cite[Theorem 1]{bhz1}. 

Theorem \ref{theo:rigidity} shows that to show the existence of degenerate $ \alpha$'s 
we need to show that $ \tr ( (T_0 |_{ L^2_{0, j} })^{2p } ) \neq 0 $, $ j = 0 , 1 $ (as explained in 
\S \ref{s:trace} we are allowed to take $ k = 0 $).

Because of the symmetry of the spectrum \eqref{eq:eigs} simple $ \alpha$'s correspond
 flat bands of multiplicity $ 2 $ and double $ \alpha $'s, to flat bands of multiplicity 4. In Figures \ref{fig:bands1} and \ref{fig:bands2}, we see that the band structure for complex and real double magic angles behaves similarly close to the magic angle. The two bands are closest at the $\Gamma$ point and are stacked on top of each other.

Examples of $ U $'s satisfying \eqref{eq:symmetries} and \eqref{eq:defU} are given by  
\begin{equation}
\label{eq:potential} U_{1}(z ) = \sum_{\ell=0}^2 \omega^\ell e^{\frac12 ( {z \bar \omega^\ell - \bar z  \omega^\ell}) } \ \text{ and } \ U_2(z ) =\tfrac{1}{\sqrt{2}}\Big( U_1(z ) - \sum_{\ell=0}^2 \omega^\ell e^{-(z \bar \omega^\ell - \bar z  \omega^\ell)}\Big).
\end{equation}
Numerical experiments suggest that these two potentials exhibit flat bands of different multiplicities:
\begin{equation}
\label{eq:U2A}
m_{ U_j } ( \alpha ) = j  ,  \ \ \alpha \in \mathcal A_{U_j}  \cap \mathbb R  , \ \ j = 1, 2 , 
\end{equation}
see Figure~\ref{fig:degeneracy}. We show (see Theorem \ref{mult} below) that the potential $ U_1 
$ (the Bistritzer--MacDonald potential) has infinitely many (complex) degenerate magic $ \alpha$'s.
While in the case of $U_1$ all magic angles on the real axis appear to be simple, the two-fold degenerate magic angles, with non-zero imaginary parts, become real when a suitable magnetic field is added \cite{Le22}.

 \begin{figure}
\includegraphics[width=7.5cm]{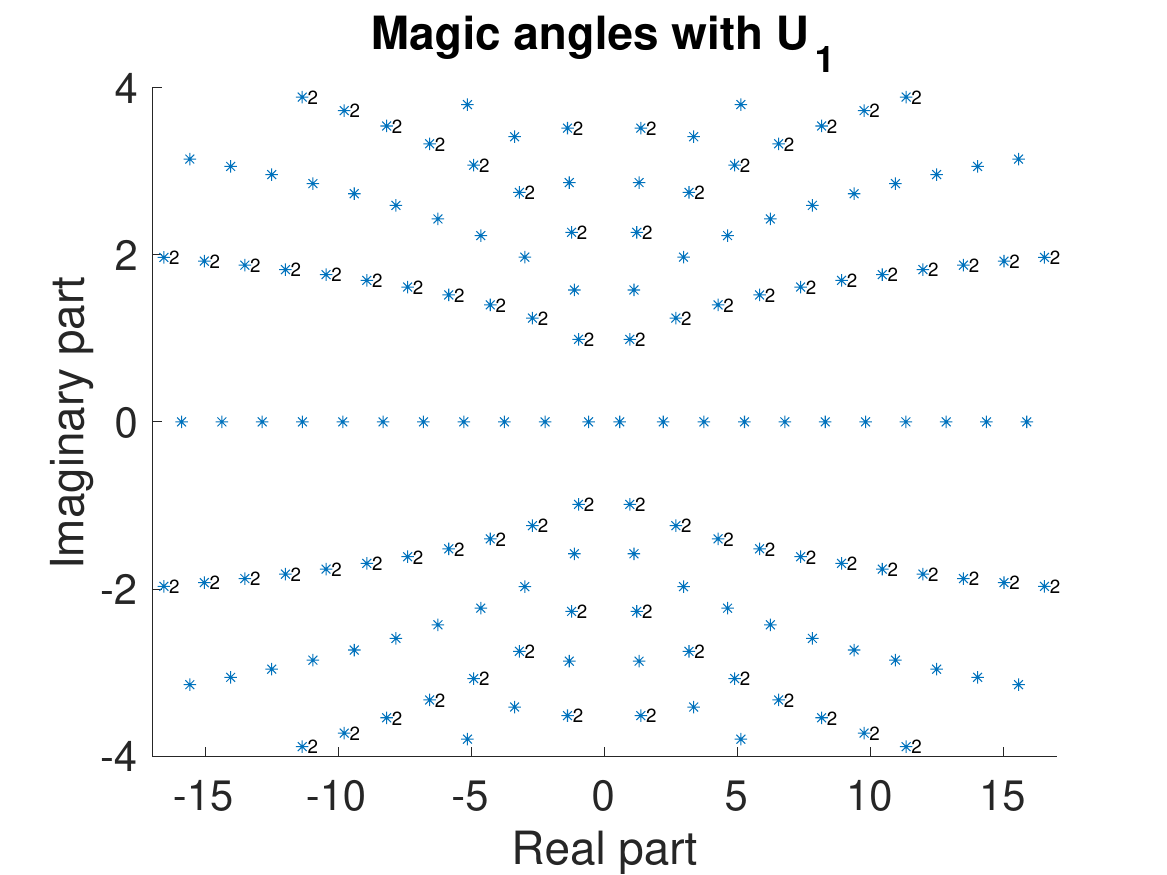}
\includegraphics[width=7.5cm]{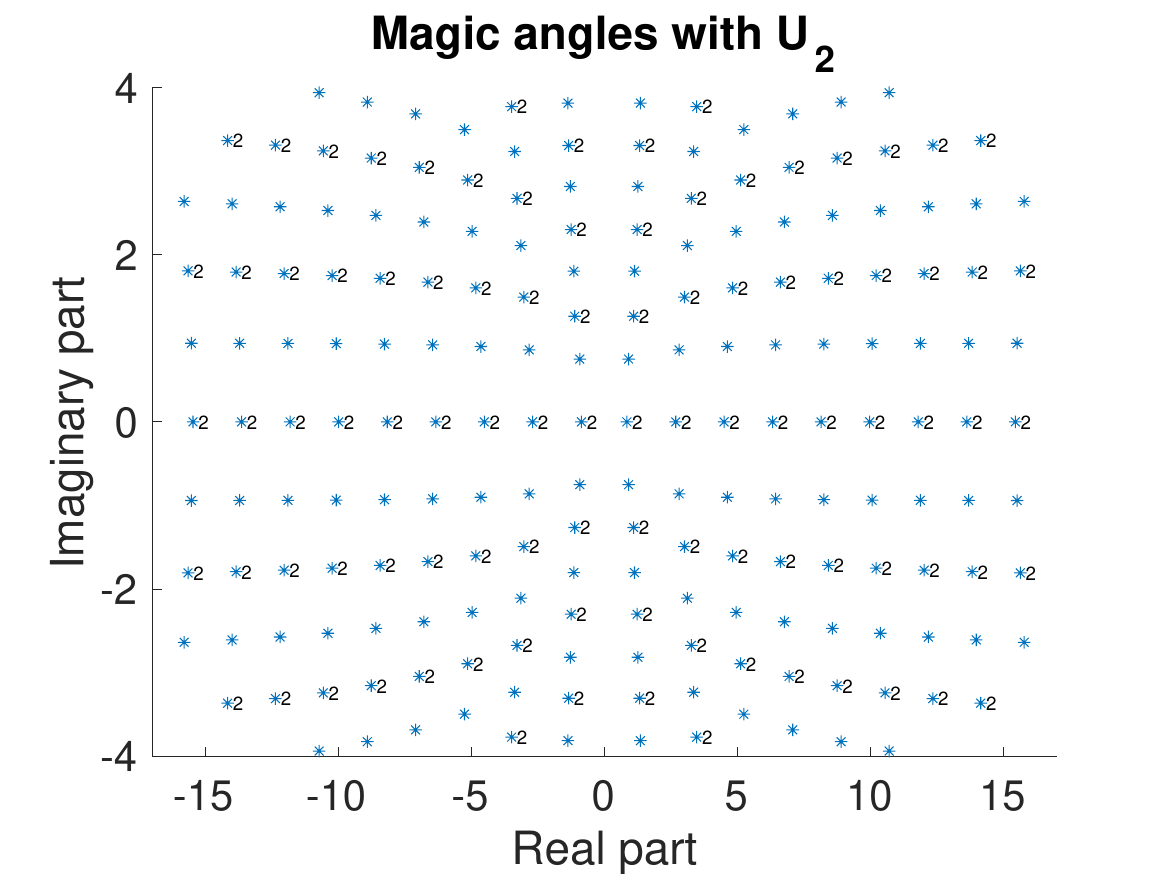}
\caption{Magic angles $\alpha$ derived from potentials $U_{\pm}=U_1(\pm \bullet)$ (left) and $U_{\pm}= U_2 (\pm \bullet) $ (right) in 
\eqref{eq:Hamiltonian}. The multiplicity of the flat bands $u$ of $(D(\alpha)+k)u_{k}=0$ is illustrated by the numbers (no number $\rightarrow$ simple magic angle, 2 $\rightarrow$ two-fold degenerate magic angle) in the figure. The movie
\url{https://math.berkeley.edu/~zworski/Interpolation.mp4} shows 
the magic angles for interpolation between these potentials:
$ U ( z ) = ( \cos \theta -\sin \theta )  U_1 ( z ) + \sin \theta U_2 ( z ) $;
multiplicity one magic angles are coded by $ \red{*} $ and multiplicity two by $ \blue{*} $. }
\label{fig:degeneracy}
\end{figure}

\begin{table}[!htb]
\begin{minipage}{.45\linewidth}
\input alpha_table_1

\end{minipage}
\begin{minipage}{.45\linewidth}
\input alpha_table_2

\end{minipage}
\medskip
\caption{First $11$ real magic angles, rounded to 6 digits, for $U_{\pm}=U_1(\pm \bullet)$ (left) and $U_{\pm}= U_2 (\pm \bullet)$ (right).
The $ \alpha$'s for $ U_1 $ are simple and the ones on the right are double.}
\label{table:double}
\end{table}

\begin{figure}
\includegraphics[width=7cm]{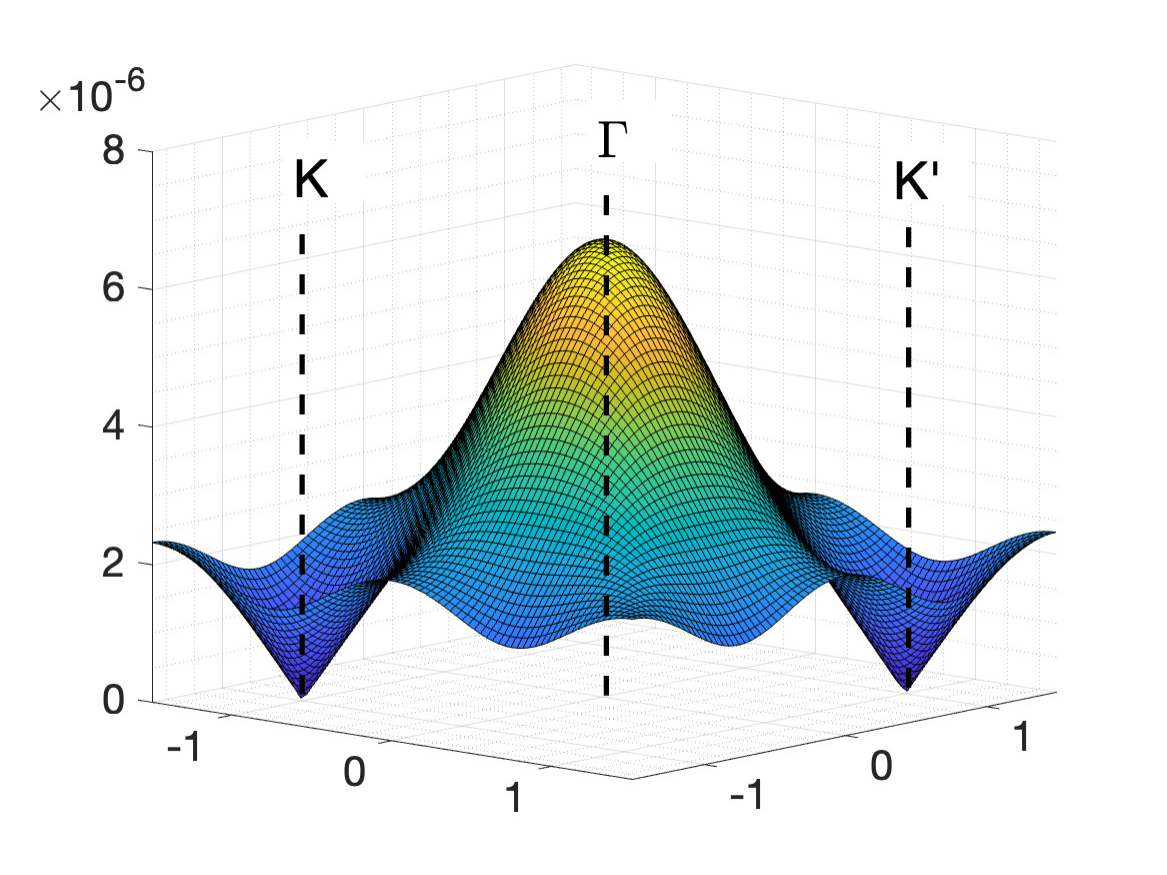}
\includegraphics[width=7cm]{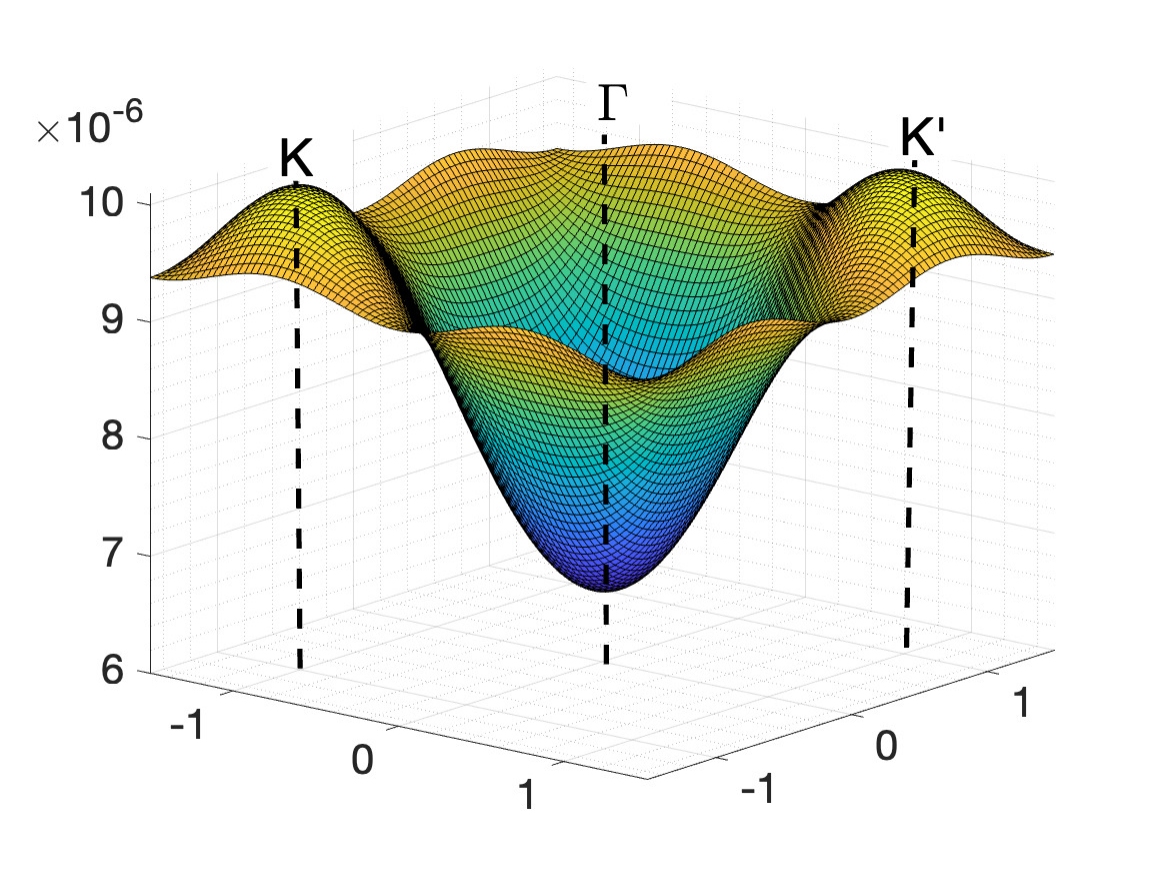}
\caption{\label{fig:bands1} Let $\alpha \approx 0.853799$ as in Table \ref{table:double}, lowest two Bloch band with positive energy close to the first magic angle with $U_{\pm}= U_2 (\pm \bullet)$. We plot $ E_1(k)$ (left) and $ E_2(k)$ (right). } 
\end{figure}

\begin{figure}
\includegraphics[width=7cm]{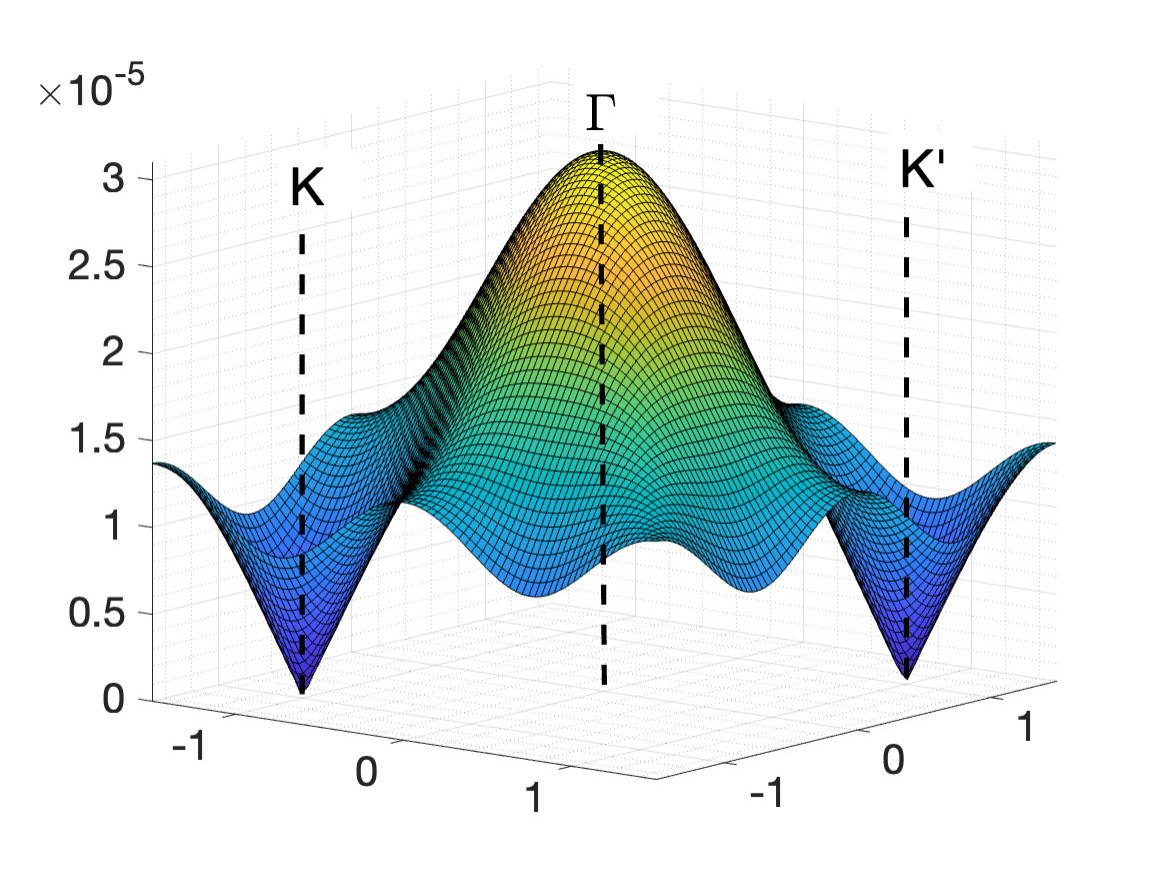}
\includegraphics[width=7cm]{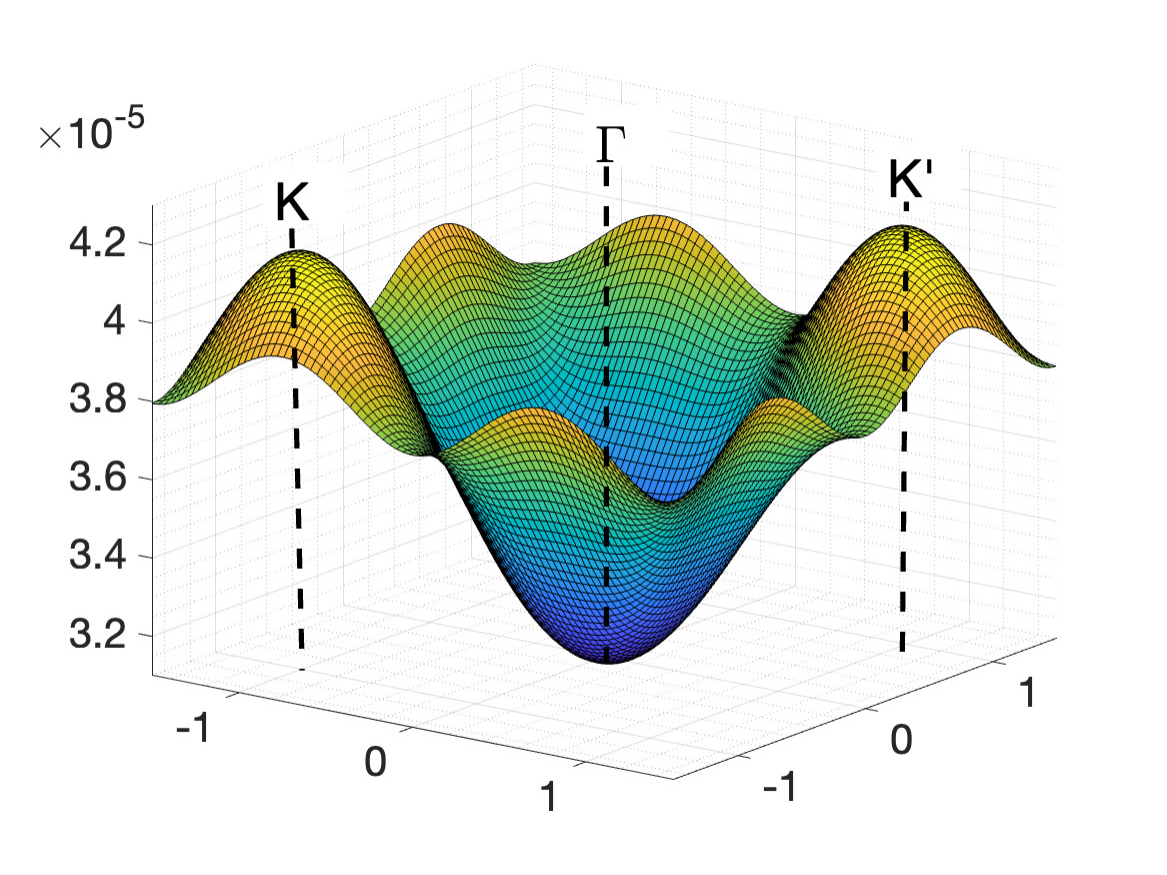}
\caption{\label{fig:bands2} Let $\alpha \approx 0.9628 + 0.9873i$ the first complex magic angle for $U_{\pm}= U_1 (\pm \bullet)$, lowest two Bloch band with positive energy close to the first degenerate magic angle. We plot $ E_1(k)$  (left) and $E_2(k)$ (right). } 
\end{figure}

 \begin{theo}[Degenerate magic angles]
\label{mult}
For the Bistritzer--MacDonald potential, $U_+=U_1$ and $U_-=U_1(-\bullet)$, defined in \eqref{eq:potential}, there exist infinitely many $\alpha\in \mathcal A$ which are not simple.
\end{theo}
Theorem \ref{sec:races_to_infinity} in \S \ref{s:inf} states this for a larger class of potentials satisfying the assumptions of \cite[Theorem 5]{bhz1} with an additional non-degeneracy condition, see Theorem \ref{General}. 

It is natural to ask if multiplicities always occur and if multiplicities of higher degree are also ubiquitous. If we do not demand that \eqref{eq:defU} holds, 
then, generically in the sense of Baire, magic angles are either simple or two-fold degenerate:
\begin{theo}[Generic simplicity]
\label{t:sim}
For Hamiltonians \eqref{eq:Hamiltonian} satisfying \eqref{eq:symmetries},
there exists a generic subset (an intersection of open dense sets),
$ \mathscr V_0 \subset \mathscr V $, where the space of 
matrix-valued potentials, $ \mathscr V $,  is defined in \eqref{eq:Vscr},
such that if $V \in \mathscr V_0$ then (see \eqref{eq:defm})
\[      m ( \alpha ) \leq 2 . \]
More precisely, 
when $ \alpha $ is simple then 
 \begin{equation}
 \label{eq:gens1} \dim\ker_{L^2_{0,2}}(D(\alpha)) =1 \text{ and }\dim\ker_{L^2_{0,0}}(D(\alpha))=\dim\ker_{L^2_{0,1}}(D(\alpha))=0 \end{equation}
 and when it is double, 
 \begin{equation}
 \label{eq:gens2} \dim\ker_{L^2_{0,2}}(D(\alpha)) =0 \text{ and } \dim\ker_{L^2_{0,0}}(D(\alpha))=\dim\ker_{L^2_{0,1}}(D(\alpha))=1. 
 \end{equation}
\end{theo} 

\noindent
{\bf Remark.} It may seem at first the conclusions \eqref{eq:gens1},\eqref{eq:gens2} follow
from Theorem \ref{theo:rigidity}. However, in that theorem, we assumed also \eqref{eq:defU}, which does not necessarily hold for potentials in $ \mathscr V_0 $.

The Chern number and Berry curvature associated to the degenerate flat band have similar properties to the case of simple flat bands. In particular, we have the following result proved
in  \ref{sec:Chern}. More specifically, we prove that the Bloch vector bundle decomposes into a trivial bundle of rank $ m ( \alpha ) - 1 $ and a line bundle isomorphic to that of a simple band. 
\begin{theo}[Flat band topology]
\label{t:Chern} 
For  $\alpha \in \mathcal A$, the Chern number of the rank $ m ( \alpha )$  vector bundle $ E $ associated to $\ker_{L^2_0}(D(\alpha)+k)$
(see \eqref{eq:defE}) is 
\begin{equation}
\label{eq:cherno}
c_1 ( E ) = -1 .
\end{equation}
In addition, the trace of the curvature, $H$, is non-negative and satisfies $H(k) = H(\omega k)$, $ H ( k ) = H ( - k) $.
\end{theo}
In Section \ref{sec:numerics}, we collect numerical observations on the possibility of having eigenvalues of $T_k$ of algebraic multiplicity $2$ but geometric multiplicity $1$ and thus corresponding to simple magic angles. We also discuss features of the Berry curvature for two-fold degenerate magic angles. {That motivates the presentation of the properties of the curvature in Theorem \ref{t:Chern}.}

 \smallsection{Acknowledgements} 
We would like to thank 
Mengxuan Yang for helpful discussions.Special thanks go also to Zhongkai Tao for his help with simplifying Section \ref{sec:Chern}. 
TH and MZ  were partially 
supported
by the National Science Foundation under the grant DMS-1901462
and by the Simons Foundation under a ``Moir\'e Materials Magic" grant. 

\smallsection{Data availability statement}
The datasets generated during and/or analysed during the current study are available from the corresponding author on request.

\smallsection{Conflict of interest statement}
The authors have no conflict of interest to declare.

\section{Properties of the Hamiltonian}
\label{s:propH}

In Section \ref{s:intr} the moir\'e lattice is given by $4\pi i (\ZZ \omega \oplus \ZZ \bar \omega)$ which is consistent with the notation in the physics literature \cite{BM11,magic}. We followed it to make the article accessible to a broader audience. From Section 2 onwards, we introduce a simple change of variables $ z_{\rm{new}} = \frac43 \pi i z_{\rm{old}} $ so that the lattice becomes $\Lambda:=\mathbb{Z} \oplus \omega \mathbb{Z}$ with dual lattice $\Lambda^* :=\frac{4\pi i}{\sqrt{3}} \Lambda$.  In doing so, we simplify mathematical expressions involving, for instance, Jacobi theta functions. 
This new coordinate system has been introduced in \cite{dynip}, see also \cite[Appendix A]{bhz2}.

Thus we work now with \eqref{eq:Hamiltonian} but now we assume
 \begin{equation}
 \label{eq:news}
U_\pm (z+\gamma) = e^{\pm i  \langle \gamma,K \rangle} U_\pm (z), \ \ \gamma \in \Lambda, 
\quad U_\pm 
 (\omega z) = \omega U_\pm (z).
 \end{equation}
 Here and elsewhere, $ \langle z, w \rangle:=\Re(z \bar w)$,  $ \pm K $ are the nonzero
 points of high symmetry, $ \omega K \equiv K \! \! \mod \! \Lambda^* $, $ K = \frac 43 \pi $.

  The analogue of \eqref{eq:defU} is given by 
 \begin{equation}
 \label{eq:defUn}
 U_+ ( z ) = U ( z ),  \ \ \ U_- ( z ) = U ( -z ) , \ \ \  \quad \overline{U(\bar z)} =-U(-z),
 \end{equation}
and the Bistritzer--MacDonald potential is now  $ U ( z ) = - \frac 43 \pi i U_1 ( \frac 4 3 \pi i z ) $,
where $ U_1 $ is given in \eqref{eq:potential}.

The Hamiltonian is still of the form
\begin{equation}
\label{eq:RescPot}
H(\alpha):= \begin{pmatrix} 0 & D(\alpha)^* \\ D(\alpha) & 0 \end{pmatrix} \text{ with }D ( \alpha ) := \begin{pmatrix}
2 D_{\bar z } & \alpha U_+ ( z ) \\
\alpha U_- (  z )& 2 D_{\bar z } \end{pmatrix}.
\end{equation}
We then define 
$$\rho(z):=\operatorname{diag}(\chi_{k_j}(z)), \quad k_2 = -k_1 = K , \ \ \in \CC/\Lambda^*, \quad \chi_k(z):=e^{i \langle z, k \rangle}, $$
so that for $\gamma \in \Lambda$
$$  V ( z + \gamma ) = \rho ( \gamma )^{-1} V ( z ) \rho ( \gamma) , \ \ \ 
V(z):= \begin{pmatrix} 0 & U_+(z)\\ U_- (z) & 0 \end{pmatrix}. $$
The modified potential, $V_{\rho}(z):=\rho(z) V(z) \rho(z)^{-1}$, is $ \Lambda$-periodic and thus
\[ \rho(z) D(\alpha) \rho(z)^{-1} = D_{\rho}(\alpha), \quad D_{\rho}(\alpha) :=\operatorname{diag}((2D_{\bar z}-k_j)_{j=1}^2) + V_{\rho}(z).\]
Using the rotation operator $\Omega u (z) = u(\omega z)$, rotating by $2\pi/3$, satisfying $\Omega D(\alpha) = \omega D(\alpha) \Omega,$ we can define $\mathscr{C} =\operatorname{diag}(1, \bar\omega) \Omega$ such that $\mathscr C H = H \mathscr C$ and translation operator $\mathscr L_\gamma u(z):=\rho(\gamma) u(z+\gamma).$
By using the translation $\mathscr L_{\gamma}$, we can define,  for $ k \in \mathbb C $, the spaces
\[ H_k^s :=H_k^s (\CC/\Lambda ,\CC^n):=\{ u \in H^s_{\text{loc}}(\CC;\CC^n): \mathscr L_{\gamma} u =e^{i \langle k, \gamma \rangle} u, \gamma \in \Lambda\}, \text{ with } L^2_k:=H_k^0,\]
where $n =1$ corresponds to the first, $n=2$ to the upper two, and $n=4$ to all components of $\mathscr L_{\gamma}.$

When $  k \in \mathcal K := \{ K , -K , 0 \} + \Lambda^*  $ we also define
\[ L^2_{k,p} = L^2_{k,p} ( \mathbb C/\Gamma ; \mathbb C^n ) := \{ u \in L^2_k : 
u ( \omega z ) = \bar \omega^p u ( z )\} , \ \  L^2_k = \bigoplus_{p \in \ZZ_3} L^2_{k,p} . \]
We can then define a generalized Bloch transform
\[ \begin{gathered}    \mathcal  B u ( z, k ) := \sum_{ \gamma \in \Lambda } e^{ i \langle  z + \gamma, k \rangle}   \mathscr L_\gamma  u ( z ) ,  
 \ \ \  \mathcal B u ( z , k + p ) =  
e^{ i \langle z , p \rangle}\mathcal B u  ( z, k )  , \ \ p \in \Lambda^*, \ \ \ u \in \mathscr S ( \mathbb C ) , 
\\
\mathscr L_{\alpha  } \mathcal  B u ( \bullet, k ) = 
\sum_{\gamma} e^{ i \langle z + \alpha + \gamma , k \rangle } \mathscr L_{\alpha + \gamma } 
u ( z ) =  \mathcal B u ( \bullet, k ) , \ \ \alpha  \in \Lambda
\end{gathered} \]
such that
\begin{equation}
\label{eq:conj}    \begin{gathered} \mathcal B D ( \alpha ) = ( D ( \alpha ) - k ) \mathcal B , \ \ \
D ( \alpha ) - k = e^{ i \langle z, k \rangle } D ( \alpha )  e^{ - i \langle z, k \rangle }, \\
\mathcal B H ( \alpha ) = H_k ( \alpha ) \mathcal B , \ \ \ 
H_k ( \alpha )  := e^{ i \langle z, k \rangle } H( \alpha ) e^{ - i \langle z, k \rangle } = 
\begin{pmatrix} \ \ 0 & D ( \alpha)^* - \bar k \\
D ( \alpha ) -k  & 0 \end{pmatrix}  . \end{gathered}  \end{equation}
In particular, we say $H(\alpha)$ exhibits a flat band at energy zero if and only if $0 \in \bigcap_{k \in \mathbb C} \Spec(H_k(\alpha)).$
To study the set of $\alpha$ at which $H(\alpha)$ exhibits a flat band at zero, we define the set of Dirac points $\mathcal K_0 :=\{K,-K\}+ \Lambda^*$ such that for $k \notin \mathcal K_0$ we can define the compact \emph{Birman-Schwinger} operator
\begin{equation}
\label{eq:BS}
T_k = R(k)V(z): L^2_{0} \to L^2_{0}, \quad R(k)= (2D_{\bar z}-k)^{-1}. 
\end{equation}

This operator then characterizes the set of magic angles in the sense stated in the next Proposition. 
\begin{prop}[{\cite[Theorem 2]{beta}},{\cite[Proposition 2.2]{bhz2}}]
\label{p:magicD}
There exists a discrete set $ \mathcal A$ such that
\begin{equation}
\label{eq:magic} 
\Spec_{ L^2 _0 } D ( \alpha ) = \left\{ \begin{array}{ll}  \mathcal K_0  & \alpha \notin 
\mathcal A , \\
 \mathbb C & \alpha \in \mathcal A . \end{array} \right.
\end{equation} 
Moreover, 
\begin{equation}
\label{eq:specc}  \alpha \in \mathcal A \ \Longleftrightarrow \ \exists \, k \notin \mathcal K_0, \ 
\alpha^{-1} \in \Spec_{ L^2 _0 } T_k \ \Longleftrightarrow \forall \, k \in \mathcal K_0,
\alpha^{-1} \in \Spec_{ L^2 _0 } T_k, 
\end{equation}
where $ T_k $ is a compact operator given by 
\begin{equation}
\label{eq:defTk}   T_k := R ( k )  V ( z ) : L^2 _{0} \to L^2 _{0} , \ \ \ R( k ) :=  ( 2 D_{\bar z } - k )^{-1}  
\end{equation}
\end{prop}
In particular, the spectrum of $T_{k_0}$ is independent of $k_0 \notin \mathcal K_0$ and characterizes parameters $\alpha \in \CC$ at which the Hamiltonian exhibits a flat band at zero energy. Since the parameter $\alpha$ is inherently connected with the twisting angle, we shall refer to $\alpha$'s at which \eqref{eq:specc} occurs as \emph{magic} and
denote their set by $ \mathcal A \subset \mathbb C $. We then square the operator $T_{k_0}^2 = \operatorname{diag} (A_{k_0},B_{k_0})$ where $A_{k_0} = R(k_0)U(z)R(k_0)U(-z).$ Setting $k_0=0$, we notice that $T_0$ leaves the subspaces $L^2_{0,j}$ invariant. By projecting the spaces $L^2_{0,j}$ onto the first component, we can define $A_0$ on spaces $L^2_{0,j}.$

\medskip

\noindent
{\bf Remark.} If $\alpha \in \mathcal A$ is simple, then $1/\alpha$ is an eigenvalue of $T_0$ with  eigenvalue of geometric multiplicity $1$ and the Hamiltonian exhibits a \emph{two-fold degenerate flat band} at energy zero. If $\alpha \in \mathcal A$ is two-fold degenerate, then $1/\alpha$ is an eigenvalue of $T_0$ with eigenvalue of geometric multiplicity $2$ and the Hamiltonian exhibits a \emph{four-fold degenerate flat band} at energy zero. It follows from Theorem \ref{p:zeros} that we can drop the minima in the above definition. 

\medskip

Suppose that the potential $U(z)$ satisfies the symmetries given in \eqref{eq:news}, namely  
\[ U ( z + \gamma ) = e^{i \langle  \gamma ,K \rangle}  U ( z ) , \ \ \ U ( \omega z ) = \omega U ( z ) . \]
Since $ U $ is then periodic with respect to $ 3 \Lambda $ ($ 3 K \equiv 0 \!\! \mod \! \Lambda^* $), expanding in Fourier series gives
\[  U = \sum_{  p \in  \Lambda^* /3 }   a_p e^{ i \langle z , p \rangle  } \].
The translational symmetry now writes:
$$\forall p \in \Lambda^*/3, \ \ \forall \gamma \in \Lambda, \quad   a_p  e^{ i \langle \gamma, p \rangle } = a_p e^{ i \langle \gamma, K\rangle  } .$$
Identifying the Fourier coefficients now gives that for all $ p \in \Lambda^*/3 $, 
\[ a_p \neq 0 \ \Longrightarrow \   \forall \, \gamma \in \Lambda, \ \ 
\langle \gamma, p \rangle = \langle \gamma, K \rangle \  \Longrightarrow \ 
p \equiv K \! \! \mod \Lambda^* . \]
In other words, we see that (changing notation) 
\begin{equation}
 U ( z ) = \sum_{ p \in \Lambda^* } a_p e^{ i \langle p + K , z \rangle } . 
 \end{equation}
We now investigate the rotational symmetry: it is equivalent to 
\[  \sum_{ p \in \Lambda^* } a_p e^{ i \langle \bar \omega  p + \bar  \omega K , z \rangle}  = 
f ( \omega z ) = \omega f ( z ) = \sum_{ p \in \Lambda^*} \omega a_p e^{ i \langle  p +K , z \rangle }. \]
Now, $ \bar \omega p + \bar \omega K = \bar \omega p - r^{-1} ( \omega ) + K $, where we defined the rescaling map 
\begin{equation}
\label{eq:z}
z :\Lambda^* \to \Lambda,\quad z(k):={\sqrt{3}k}/{4\pi i}.
\end{equation}
(Although $ z $ is a complex variable, the notation is justified as it is a map from $ k$-space to $ z$-space.)
Hence, the right hand side of the previous equality rewrites
\[  f ( \omega z ) = \sum_{  p \in \Lambda^* } a_{ \omega p + r^{-1} ( \bar \omega ) } 
e^{ i \langle p + K, z \rangle } ,  \]
that is $  a_p = \omega a_{\omega p + r^{-1} ( \bar \omega ) }.$
The previous discussion justified the following characterization of potentials $U(z)$ satisfying the symmetries given in \eqref{eq:news}
\begin{equation}
\label{eq:Fourier}
U(z) \text{ satisfies } \eqref{eq:news} \ \iff \  U ( z ) = \sum_{ p \in \Lambda^* } a_p e^{ i \langle p + K , z \rangle } \text{ and } \forall p\in \Lambda^*, a_p = \omega a_{\omega p + r^{-1} ( \bar \omega ) }.
\end{equation}
In other words, the values of $ a_p $  are determined on the orbits of 
\[ \kappa: \Lambda^* \in p \mapsto \omega p + r^{-1} ( \bar \omega ) , \ \ \ {\rm{Orb}}( p ) = 
\{ p , \omega p + r^{-1} ( \bar \omega ) , \bar \omega p - r^{-1} ( \omega ) \}, \ \ 
a_{\kappa (p) } = \bar \omega a_p  . \]
So, for instance, the BM potential, up to a factor, comes from the orbit of $ p = 0 $.

%

In addition, there exist a number of further anti-linear symmetries of the chiral Hamiltonian
\[ Qv(z) = \overline{v(-z)}, \quad \mathscr Qu(z) = \begin{pmatrix} 0 & Q \\ Q & 0 \end{pmatrix} u(z),\]
satisfying $QD(\alpha)Q = D(\alpha)^*$ with $Q:L^2_{k,p}(\CC/\Lambda;\CC^2) \to L^2_{k,-p}(\CC/\Lambda; \CC^2)$ with $\mathscr Q: L^2_{k,p}(\CC/\Lambda;\CC^4) \to L^2_{k,-p+1}(\CC/\Lambda;\CC^4)$ satisfying $H(\alpha)\mathscr Q = \mathscr Q H(\alpha)$
and 
\[ \mathscr Ev(z):=Jv(-z), \quad J:=\begin{pmatrix} 0 & 1 \\ -1 & 0 \end{pmatrix}\]
with $\mathscr E:L^2_{\pm K,\ell}(\CC/\Lambda; \CC^2) \to L^2_{\mp K,\ell}(\CC/\Lambda; \CC^2)$ and 
\begin{equation}
\label{eq:E}
\mathscr E:L^2_{0,\ell}(\CC/\Lambda; \CC^2) \to L^2_{0,\ell}(\CC/\Lambda; \CC^2)\text{ satisfying } \mathscr E D(\alpha) \mathscr E^* = -D(\alpha).
\end{equation}
Finally, we also introduce their composition $\mathscr A: L^2_{k,p}(\CC/\Lambda;\CC^2) \to L^2_{k,-p}(\CC/\Lambda; \CC^2)$
\begin{equation}
\label{eq:A}
 \mathscr A:=\mathscr E Q, \text{ with } \mathscr Av(z):=\begin{pmatrix} 0 & 1 \\ -1 & 0 \end{pmatrix} \overline{v(z)}
 \end{equation}
with $\mathscr A D(\alpha) \mathscr A= -D(\alpha)^*.$

Using the above symmetries, we observe that
\begin{prop}
\label{prop:equiv_spec}
The spectrum of $T_0$ satisfies $\Spec_{L^2_{0,p}}(T_0) = \Spec_{L^2_{0,-p+1}}(T_0).$
In particular, for $m \ge 2$ we find $\tr_{L^2_{0,p}}T_0^{2m} = \tr_{L^2_{0,-p+1}}T_0^{2m}.$
\end{prop}
\begin{proof}
Let $v \in L^2_{0,p}$ satisfy $T_0 v =-\lambda v$ then by multiplying by $2D_{\bar z}$ we find $D(1/\lambda)v =0$. Thus, $D(1/\lambda)^* Qv =0$ with $Qv \in L^2_{0,-p}.$ We thus have 
\[ 0=D(1/\lambda)^*  Qv =D(1/\lambda)^* R(0)^* (2 D_{z}) Qv  .\]
We conclude that $ (2 D_{z}) Qv \in L^2_{0,-p+1}$ is an eigenvector to $T_0^*$ with eigenvalue $-\bar \lambda.$
\end{proof}
\section{Zeros, spectral gap, and rigidity}
The zeros always fall into three point characterized by high symmetry:
$ \omega z \equiv z \mod \Lambda$. That determines them (up to $ \Lambda $)
as $ 0 $, $ \pm z_S $, where 
\[   z_S = i /\sqrt 3, \ \ \  \omega z_S = z_S - ( 1 + \omega )  , \]
is known as the {\em stacking point}.

\subsection{Theta function argument}
\label{s:trans}
We use the following notation
\begin{equation}
\label{eq:theta2}
\begin{gathered} 
 \theta ( z) =  \theta_{1} ( \zeta | \omega ) :=  - \sum_{ n \in \mathbb Z } \exp ( \pi i (n+\tfrac12) ^2 \omega+ 2 \pi i ( n + \tfrac12 ) (\zeta + \tfrac
12 )  ) , 
\end{gathered}
\end{equation}
\[ \theta ( \zeta + m   ) = (-1)^m \theta  ( \zeta ) , \ \ \theta ( \zeta + n \omega ) = (-1)^n e^{ - \pi i n^2 \omega - 
2 \pi i \zeta  n } \theta  ( \zeta  ) , \] 
and the fact that $ \theta $ has simple zeros at $ \Lambda $ (and not other zeros) -- 
see \cite{tata}. 
We can then define 
\begin{equation}
\label{eq:defFk}  F_k ( z ) = e^{\frac i 2   (  z -  \bar z ) k } \frac{ \theta ( z - z ( k ) ) }{
\theta( z) } ,  \ \ \ z (k):=  \frac{ \sqrt 3 k }{ 4 \pi i } , \ \  z:  \Lambda^* \to \Lambda . 
\end{equation}
In particular, we have then
\begin{equation}
\label{eq:propFk}
\begin{gathered}F_k ( z + m + n \omega ) = e^{  -  n k \Im \omega } 
e^{ 2 \pi i n z ( k ) } F_k ( z ) = F_k ( z ) , \\
  ( 2 D_{\bar z } + k ) F_k ( z ) 
 = c(k)   \delta_0 ( z ) , \ \ c(k) :=  2 \pi i {\theta ( z ( k ) ) }/{ \theta' ( 0 ) } .
\end{gathered} \end{equation}

One then has that for $u \in \ker_{L^2_{0}}(D(\alpha))$ vanishing at a point $w$ one has
\begin{equation}
\label{eq:flat_band}
(D(\alpha) + k)F_k(z-w)u(z)=0.
\end{equation}
In particular this means that vanishing of the vector valued function $ u $ at some point, implies existence of 
a flat band at $ 0 $: every $ k $ is an eigenvalue of $ D ( \alpha ) $. Presented in a slightly different way, this observation was the basis
of the analysis in \cite{magic}.

\subsection{Zeros}

We first formalise the theta function argument of \cite{magic} in a slightly different way than in 
\cite{beta}. (Where we used the fact that a non-zero Wronskian between shifted protected states implies 
existence of $ (D ( \alpha ) + k )^{-1} $ -- see \cite[Proposition 3.3]{beta}.)
\begin{prop}
\label{p:theta}
Suppose that $ \mathbf u_K (\alpha ) \in \ker_{ H_0^1 } ( D ( \alpha ) + K ) $ is a family of protected stated of the chiral model (see \cite[Theorem 1]{survey} and references given there). Then 
\begin{equation}
\label{eq:theta}  \alpha \in \mathcal A \ \Longleftrightarrow \ 
\exists \, z_0 \in \mathbb C/\Lambda \text{ such that } \mathbf u_K (\alpha ,  z_0 ) = 0 . 
\end{equation}
\end{prop}
\begin{proof}
The proof of $ \Longleftarrow $ is given in  \S \ref{s:trans}, see also \cite[\S 6]{survey}.
On the other hand if
$ \alpha \in \mathcal A $ then for every $ k $ (or equivalently some $ k \notin \{ K , - K \} +  \Lambda^* $)
there exists $ \mathbf v_{k} \neq 0 $ such that $( D ( \alpha ) + k ) \mathbf v_k = 0 ${, where we drop the subscript of $u$ and $v$ in the following. } Then the Wronskian, 
\begin{equation}
\label{eq:Wr}  W = W ( \mathbf u , \mathbf v ) = u_1 v_2 - u_2 v_1 , \ \ \ \mathbf u = \begin{pmatrix} u_1 \\ u_2 \end{pmatrix} , \ \ \mathbf v = \begin{pmatrix} v_1 \\ v_2 \end{pmatrix} , \end{equation}
satisfies $ 2 D_{\bar z } W  = -( K + k)  W $:
\[ \begin{split} 2 D_{\bar z } W & = (2 D_{\bar z } u_1) v_2  + u_1 (2 D_{\bar z } v_2)  - (2 D_{\bar z } u_2) v_1 -
u_2 (2 D_{\bar z } v_1) \\
& =  ( - \alpha U( z )  u_2 - K u_1 ) v_2 + u_1 ( - \alpha U ( - z ) v_1 - k v_2 ) \\
& \ \ \  - v_1 ( - \alpha U( -z) u_1 - K u_2 ) - u_2 ( 
- \alpha U (  z ) v_2 -  kv_1 )   \\
& =- (k+K) ( u_1 v_2 - u_2 v_1 ) = -( k + K ) W . 
\end{split} \]
Putting $ k_0 := - (k + K) \notin \Lambda^* $, the general solution of this equation is given by 
\[ W ( z, \bar z ) = e^{ \frac i 2 ( k_0 \bar z + \bar k_0 z ) } w ( z) , \ \ \ w \in \mathscr O ( \mathbf C ) . \]
Since $ W $ is periodic and $ z \mapsto e^{{-} \frac i 2 ( k_0 \bar z + \bar k_0 z ) }  $ is a bounded function,
$ w $ has to be constant, and for $ k_0 \notin \Lambda^* $, that constant has to vanish. 
If follows that for $ z \in \Omega := \complement \mathbf u^{-1} ( 0 ) $, an open 
dense set as $ \mathbf u \not \equiv 0 $ is real analytic (this follows 
from the ellipticity of the equation and analyticity of $ U $ -- see \cite[Theorem 8.6.1]{H1}),
$ \mathbf v ( z , \bar z ) = F ( z, \bar z ) \mathbf u ( z, \bar z ) $, $ F \in C^\infty ( \Omega ) $. 
By applying $ 2 D_{\bar z } $ to both sides we see that $ D_{\bar z } F = ( K -k ) F $. 
Hence 
for some $ f \in C^\infty ( \Omega   ) $, 
\begin{equation}
\label{eq:u2v}  \mathbf v( z, \bar z )  = e^{ - \frac i 2 ( k_1 \bar z + \bar k_1 z ) } f ( z) \mathbf u ( z, \bar z ) ,  \ \ \ \ 
\partial_{\bar z } f|_{ \Omega } = 0 , \ \ \ k_1 := K - k .  \end{equation}
 As
in the proof of \cite[(4.4)]{bhz1} we see that $ f $ is in fact meromorphic. 
In fact, for a fixed $ z_1 $ we put 
\[ G_0 (  \zeta ,  \bar \zeta )  := v_1 ( z, \bar z)|_{ z = z_1 + \zeta } , 
\ \ \  G_1 ( \zeta, \bar \zeta ) := e^{ - \frac i 2 ( k_1 \bar z + \bar k_1 z ) }  u_1 ( z, \bar z ) |_{ z = z_1 + \zeta } . \]
As already remarked $ \mathbf u $ and $ \mathbf v $ are real analytic and hence 
$ G_j \in \mathscr O ( B_{\mathbb C^2 } ( 0 , \delta ) ) $. 
With $ g ( \zeta ) := f ( z_1 + \zeta ) $, we have 
\[ G_0 ( \zeta , \xi ) = g ( \zeta ) G_1 ( \zeta, \xi ), \ \ \  z_1 + \zeta \in \Omega.  
\]
Now choose $ \xi_0 \in B_{\mathbb C } ( 0, \delta/2) $ such that $ \zeta \mapsto G_1 ( \zeta, \xi_0) $ is 
not identically zero (if no such $ \xi_0 $ existed, $ G_1 \equiv 0$, and hence, from the equation, $ \mathbf u  \equiv 0 $). But then 
$ \zeta \mapsto g (\zeta ) := G_1 ( \zeta, \xi_0 )/G_2 ( \zeta, \xi_0 ) $ is meromorphic near $ \zeta = 0 $ and, as $ z_1 $ was
arbitrary $ f $ is meromorphic everywhere.

For $ v $ to be periodic
$ f $ cannot be constant and hence it has to have poles. But that means that $ \mathbf u $ has to 
have zeros.
\end{proof}

We are now ready to proof Theorem \ref{p:zeros} which is a refinement of Proposition 1. 
\begin{proof}[Proof of Theo. \ref{p:zeros}]
Fix $k \in \mathbb{C}/\Lambda^*$ and assume $\alpha$ is magic. Then there exists a nonzero function $\mathbf{u}_k(\alpha) \in \ker(D(\alpha)+k)$ that vanishes somewhere. Let $\mathbf{u} := \mathbf{u}_k(\alpha)$, and suppose that $\mathbf{u}$ has $m$ distinct zeros $z_0, \dots, z_{m-1}$, each of multiplicity one. 
This can be assumed without loss of generality as we can multiply $ \mathbf u $ by a periodic meromorphic
function with poles the zeros of $ \mathbf u $ and simple zeros.  

We first show that $\dim \ker(D(\alpha)+k) \geq m$.
Assume $m \geq 2$ (the case $m=1$ is trivial). Choose points $w_j \notin \{z_0, \dots, z_{m-1}\}$ such that
\[
w_0 + w_j \equiv z_0 + z_j \quad \text{for } j = 1, \dots, m-1,
\]
and define, for $j = 1, \dots, m-1$,
\[
\mathbf{u}_j(z, \bar{z}) := \frac{ \theta(z - w_0) \theta(z - w_j) }{ \theta(z - z_0) \theta(z - z_j) } \cdot \mathbf{u}(z, \bar{z}).
\]
The prefactor is a meromorphic function with simple poles at $z_0$ and $z_j$ (see \cite[\S 3.1]{bhz2} for properties of $\theta$), but the structure of the zeros of $\mathbf{u}$ ensures that $\mathbf{u}_j \in \ker(D(\alpha)+k)$. These functions are linearly independent: if
\[
c_0 + \sum_{j=1}^{m-1} c_j \frac{ \theta(z - w_0) \theta(z - w_j) }{ \theta(z - z_0) \theta(z - z_j) } \equiv 0,
\]
then evaluating at $z = z_j$ for $j > 0$ shows $c_j = 0$, hence all $c_j = 0$. Thus, we obtain $m$ linearly independent functions in $\ker(D(\alpha)+k)$. 

Now for the reverse inequality. Suppose $\mathbf{u}, \mathbf{v}_1, \dots, \mathbf{v}_{M-1}$ span $\ker(D(\alpha)+k)$, where $M := \dim \ker(D(\alpha)+k)$. Then the Wronskians $W(\mathbf{u}, \mathbf{v}_j)$ vanish identically, since they are constant and vanish at the zeros of $\mathbf{u}$. As in \cite[Eq.~(4.1)]{bhz2}, this implies that each $\mathbf{v}_j(z, \bar{z}) = f_j(z) \mathbf{u}(z, \bar{z})$ for some meromorphic function $f_j(z)$, with $f_0 \equiv 1$, and the set $\{f_j\}_{j=0}^{M-1}$ linearly independent.

The functions $f_j$, $ 0\leq j \leq M-1 $ lie in the space $L(D)$ of meromorphic functions\footnote{See \href{https://terrytao.wordpress.com/2018/03/28/246c-notes-1-meromorphic-functions-on-riemann-surfaces-and-the-riemann-roch-theorem/}{Terry Tao's blog} for a quick introduction; formula (7) there for the version of the Riemann–Roch Theorem used here.}  
%
where $ D $ is the divisor defined by the zeros on $ \mathbf u $.
Therefore,
\[
M = \dim \ker(D(\alpha)+k)  \leq \deg ( D ) = m.
\]
Combining with the previous inequality, we conclude
\[
\dim \ker(D(\alpha)+k) = m.
\]

Moreover, since every nonzero element of the kernel has exactly $m$ zeros (counted with multiplicity), and the multiplicity is independent of $k$ by the theta function argument (see  \S \ref{s:trans} and \cite[Lemma 4.1]{bhz2}), the dimension of the kernel  is constant over $k \in \mathbb{C}/\Lambda^*$. By continuity of the spectrum of $H_k(\alpha)$ and compactness of $\mathbb{C}/\Lambda^*$, it follows that the flat bands at energy zero are isolated from the rest of the spectrum by a nonzero gap.
\end{proof} 

\noindent
{\bf Remarks.} 1. This quickly settles
\cite[Problem 14]{survey}.

\noindent
2. Although we invoked a basic version of the Riemann--Roch theorem, the proof that the number of poles
of $ f_j$'s has to be greater than $ M$ is explicit. If the poles are all simple then
\[ f_j ( z ) = \sum_{ k=1}^{K(j)} \lambda_k (j) \frac{\theta' ( z - a_k ( j ) )  }{ \theta ( z - a_k (j ) ) } + c(j), \ \ \ \ 
\sum_{ k=1}^{K(j)} \lambda_k (j) = 0 .\]
For $ f_0 \equiv 1$, $ f_1, \cdots , f_{M-1} $ to be linearly independent we need 
$ \sum_{j=1}^{M-1} K(j) \geq M $. It is not difficult to modify this construction to the case of poles of higher order.

\begin{figure}
\includegraphics[trim={3cm 0 0 0},width=7cm]{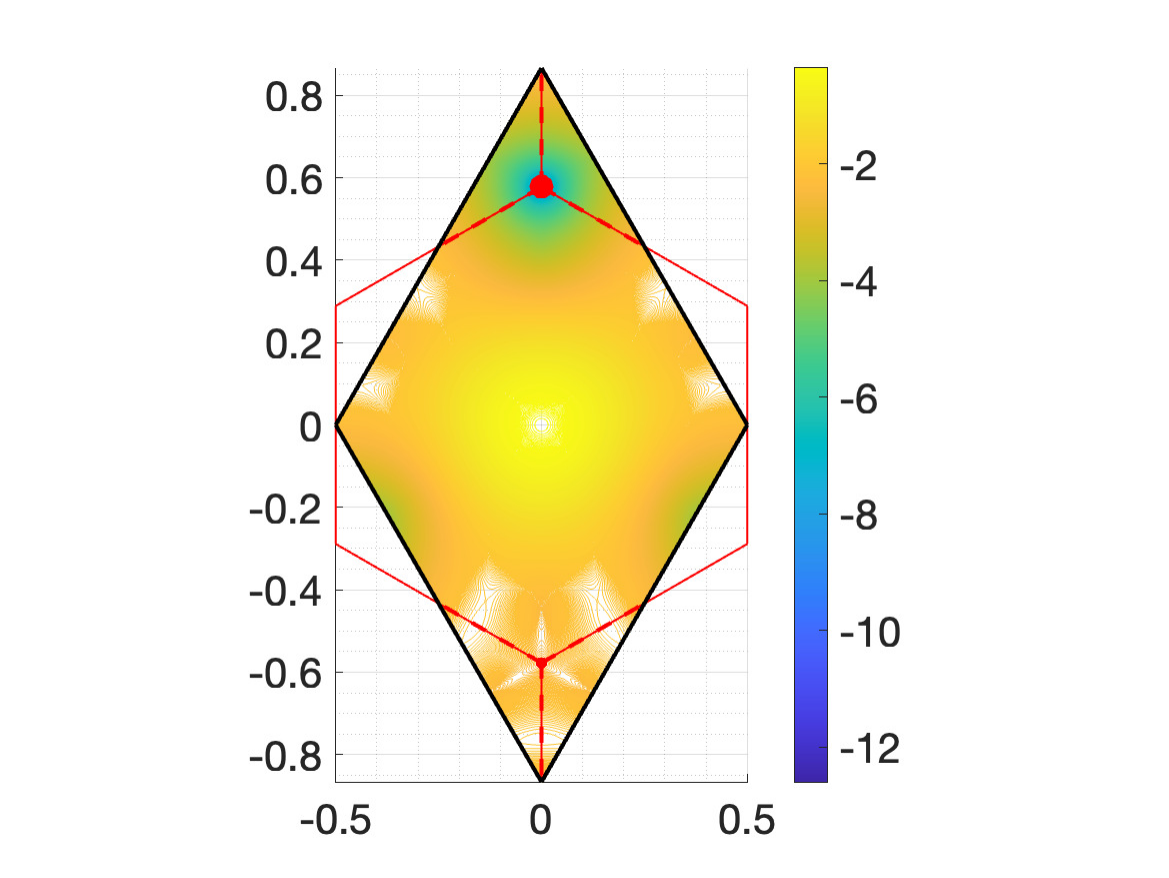}
\includegraphics[trim={3cm 0 0 0},width=7cm]{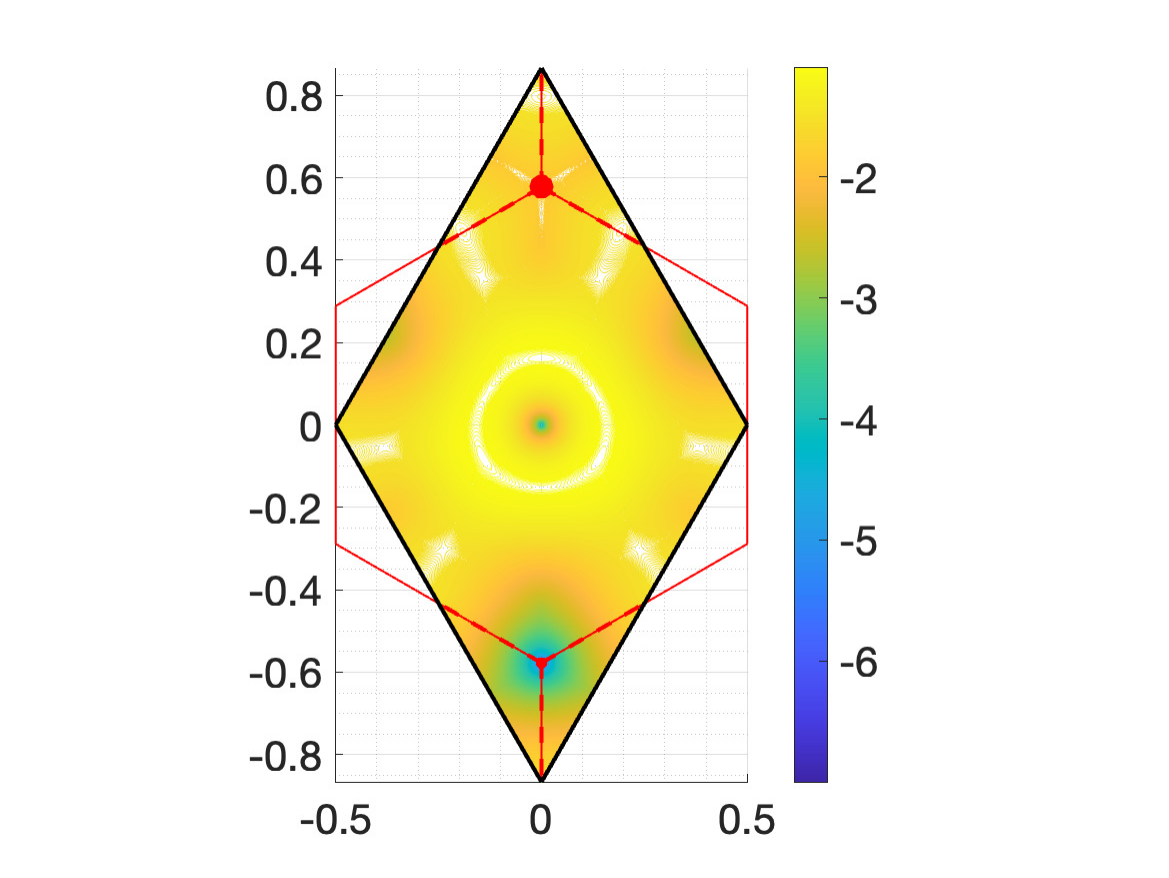}\\
\includegraphics[trim={3cm 0 0 0},width=7cm]{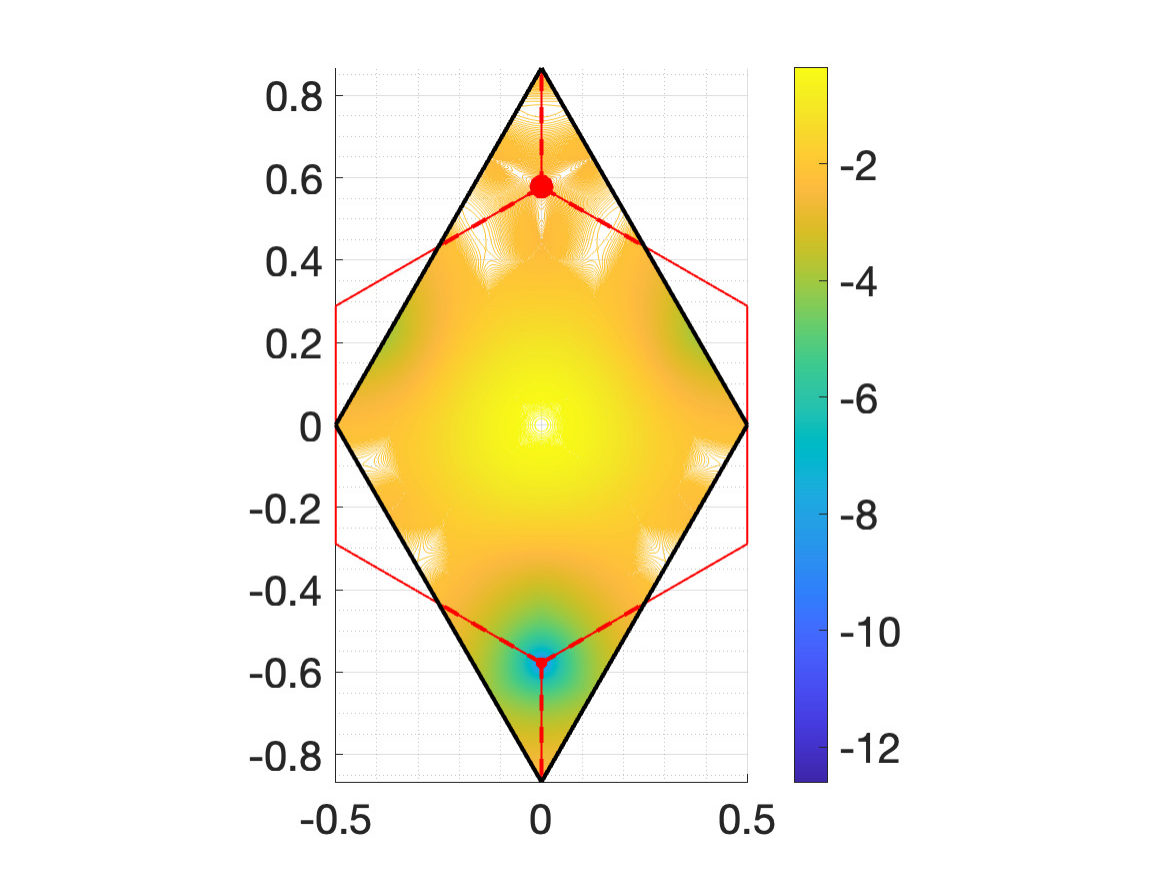}
\includegraphics[trim={3cm 0 0 0},width=7cm]{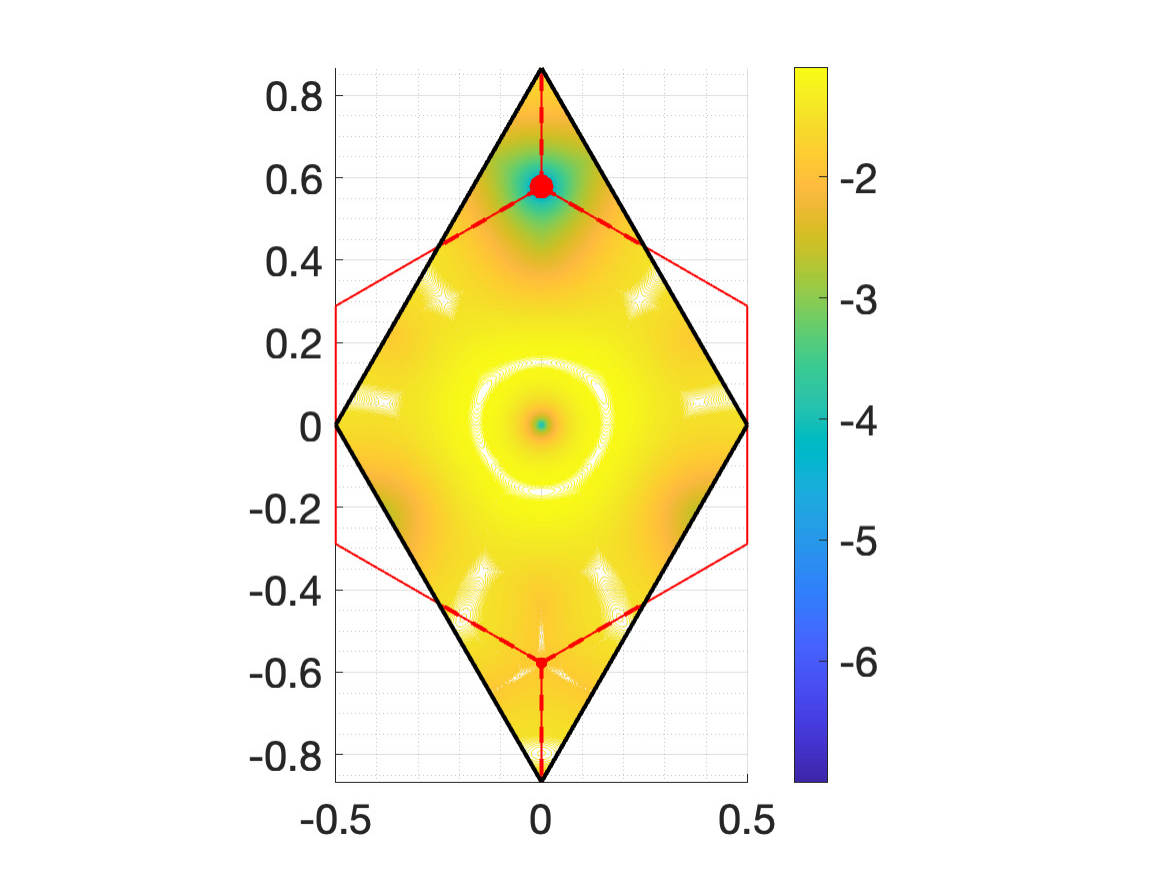}
\caption{\label{fig:zeros}Modulus of flat band wavefunctions of $\ker_{X}(D(\alpha))$ at first magic angle $\alpha=0.853799$ with $X=L^2_{i,j}$ with $i=K$(top), $i=-K$ (bottom), $j=0$ (left), $j=1$ (right) for potential $U_{\pm}= U_2 (\pm \bullet)$ in \eqref{eq:potential}.}
\end{figure}

\begin{figure}
\includegraphics[trim={3cm 0 0 0},width=7cm]{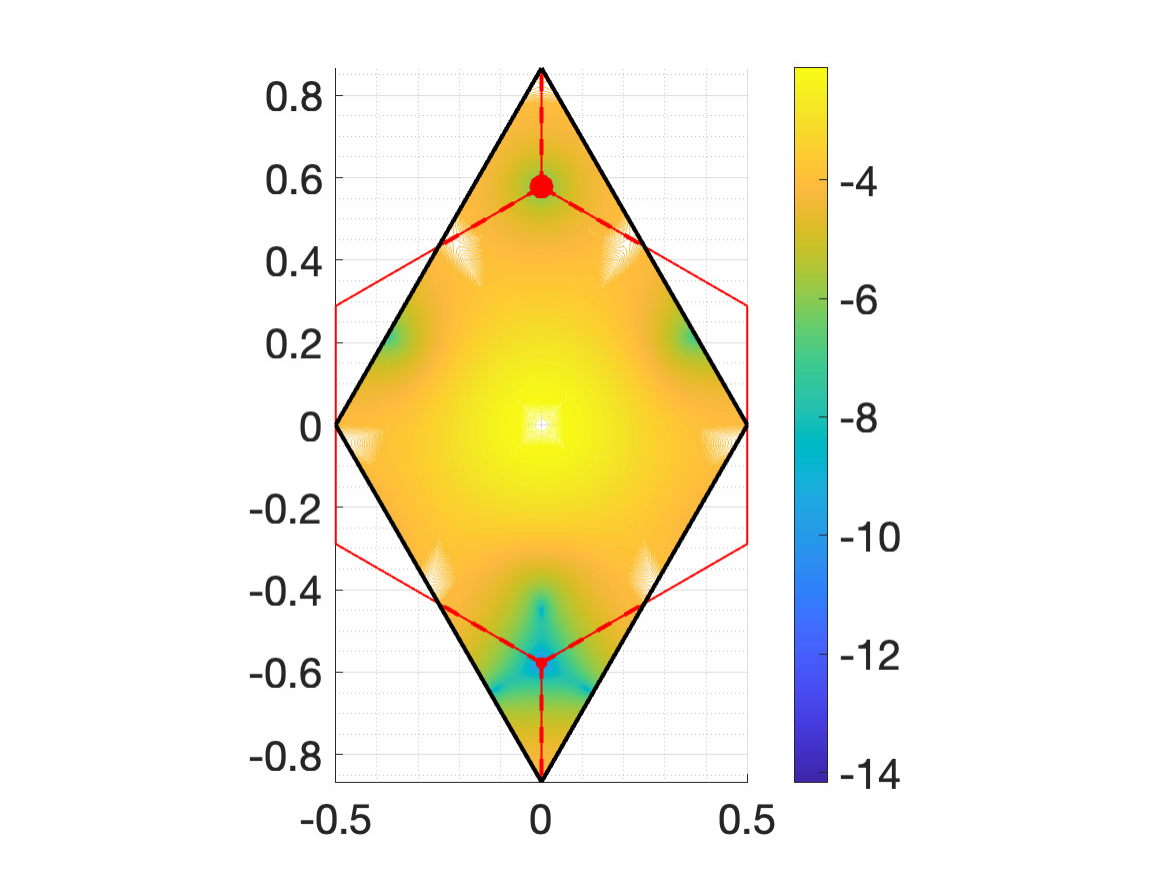}
\includegraphics[trim={3cm 0 0 0},width=7cm]{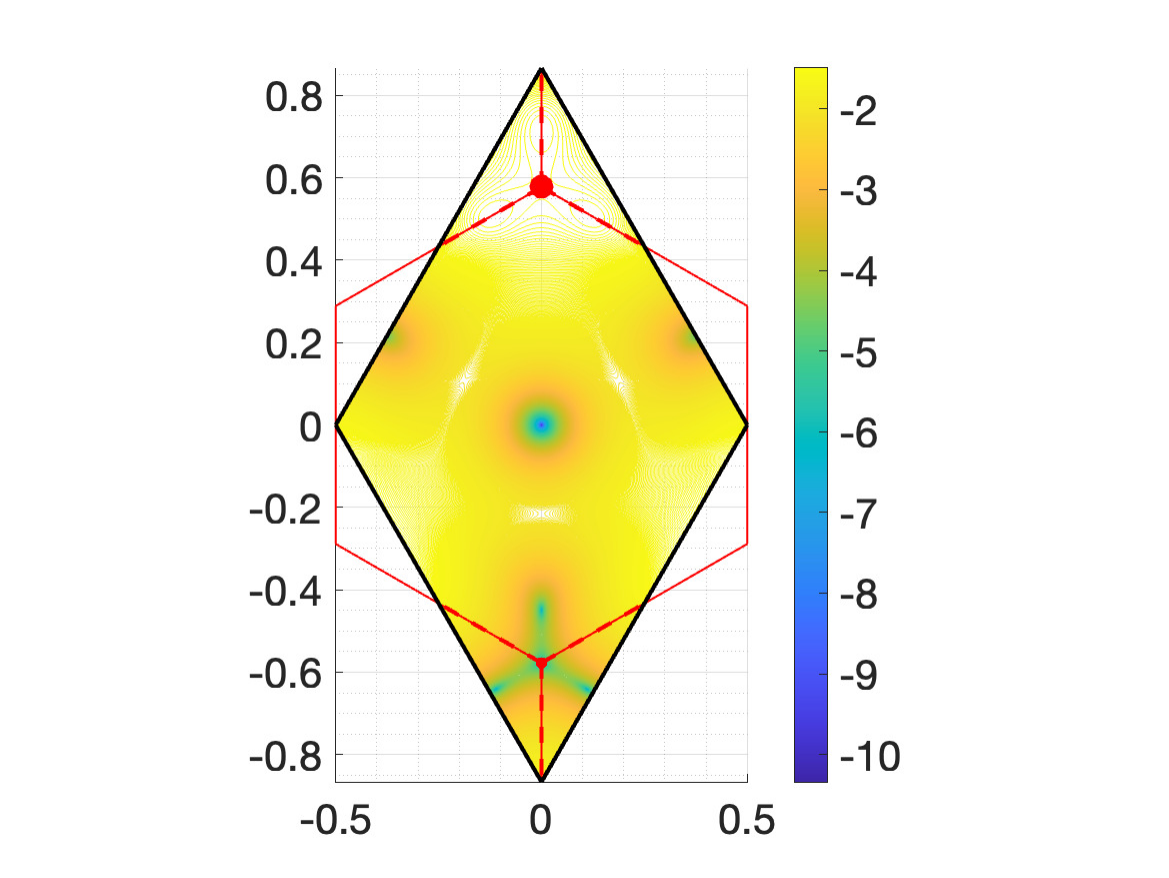}\\
\includegraphics[trim={3cm 0 0 0},width=7cm]{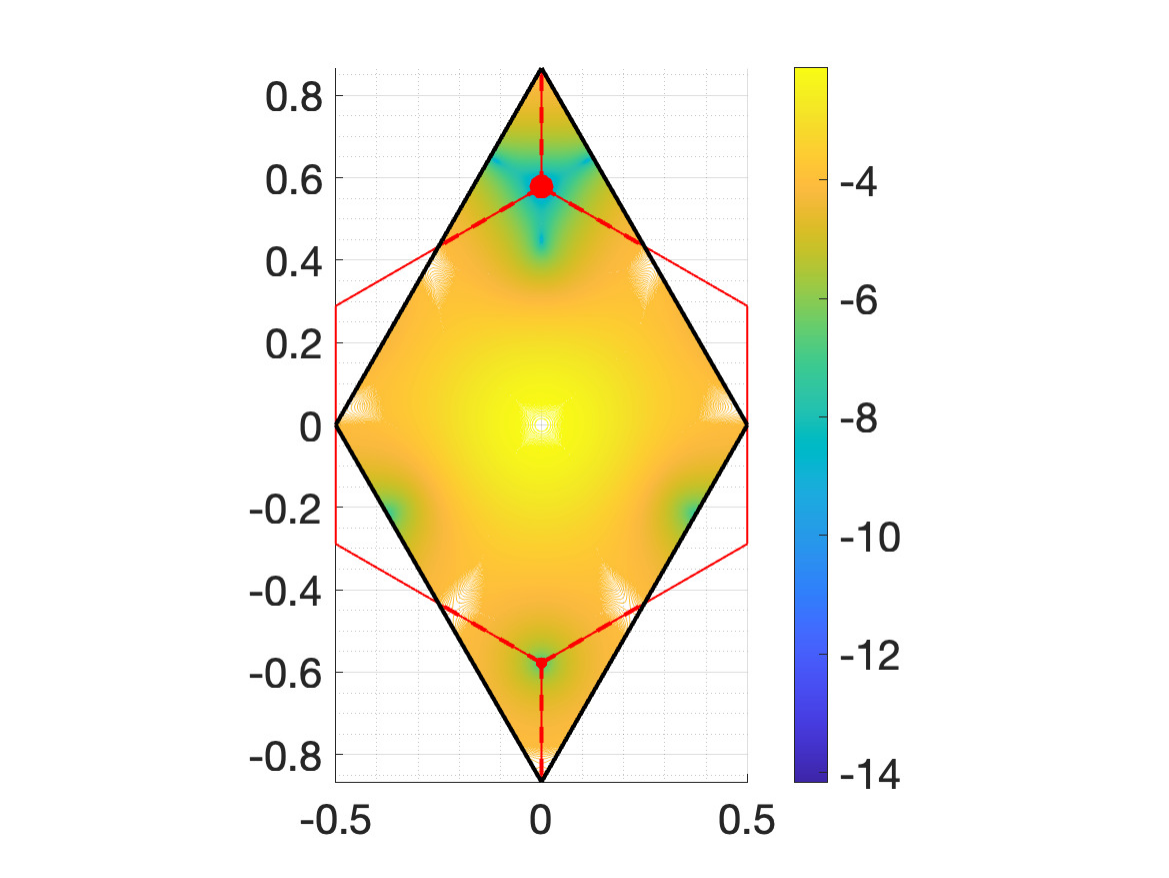}
\includegraphics[trim={3cm 0 0 0},width=7cm]{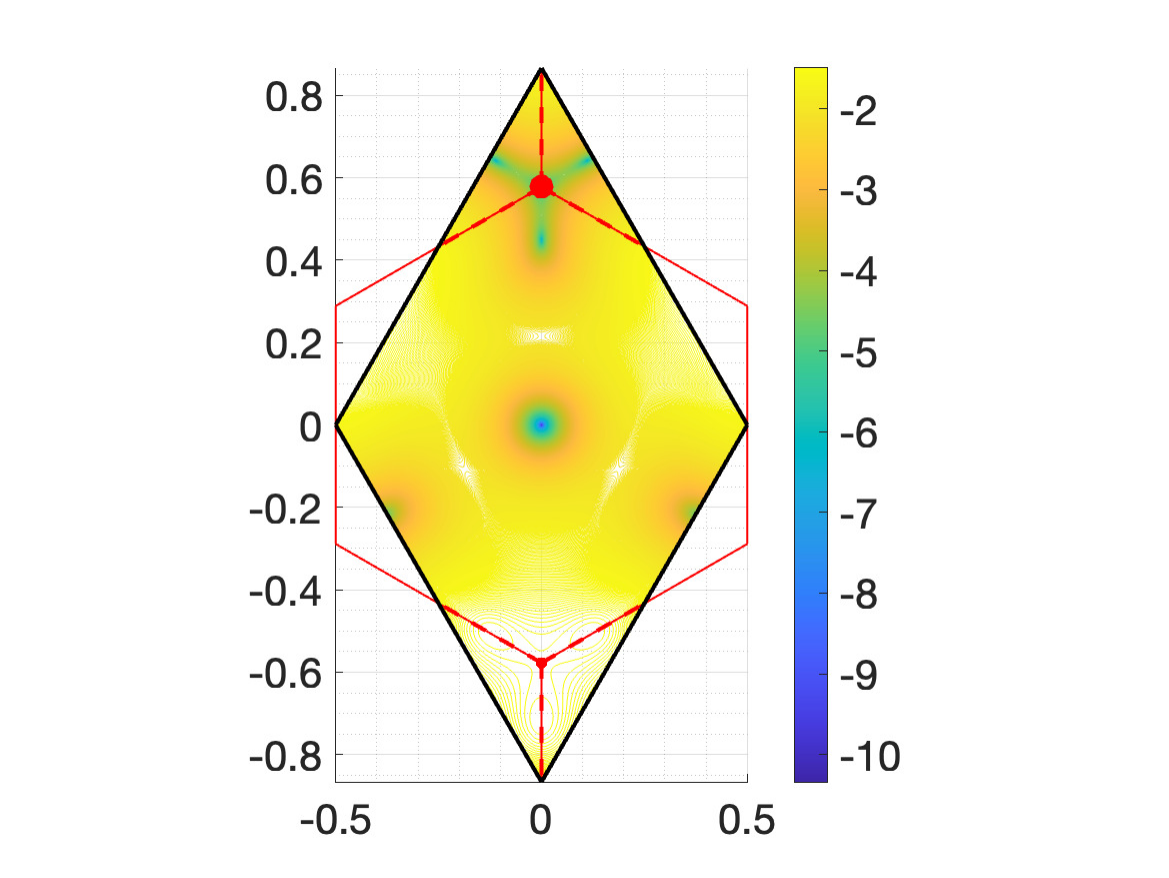}
\caption{\label{fig:zeros2}Flat band wavefunctions of $\ker_{X}(D(\alpha))$ at first magic angle $\alpha=0.853799$ with $X=L^2_{0,0}$ (left) and $X=L^2_{0,1}$ (right) for potential $U_{\pm}= U_2 (\pm \bullet)$ in \eqref{eq:potential} upper component, top and lower component, bottom.}
\end{figure}

\subsection{Rigidity}
\label{sec_rig}
In this subsection we provide
\begin{figure}
\includegraphics[width=7.5cm]{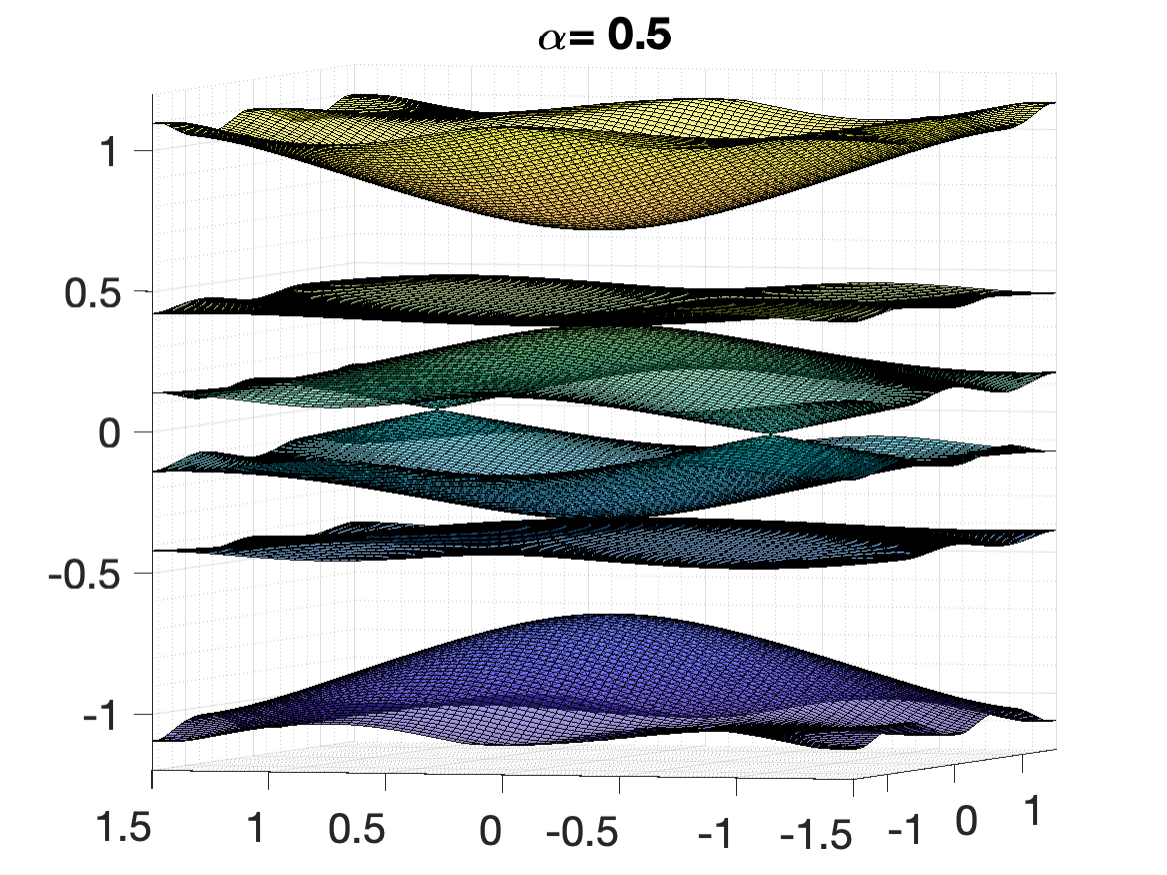}
\includegraphics[width=7.5cm]{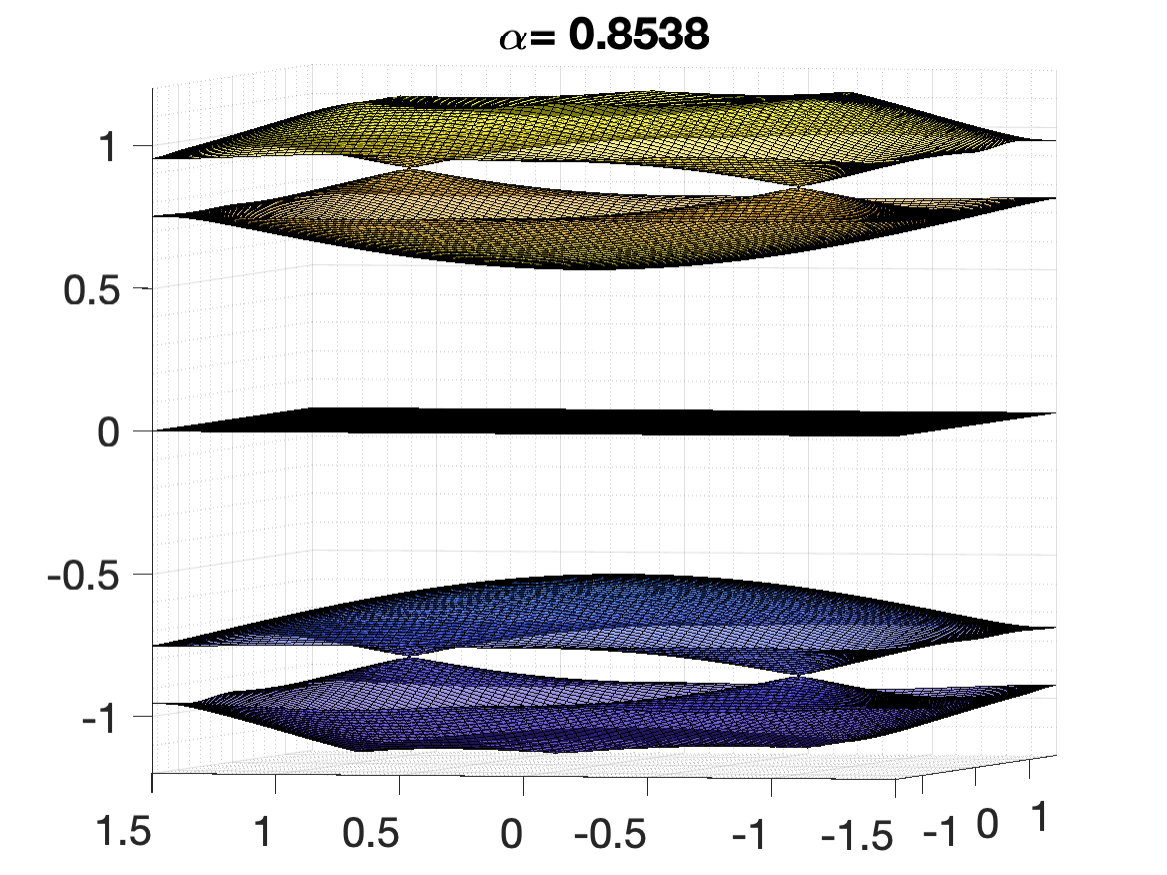}\\
\includegraphics[width=7.5cm]{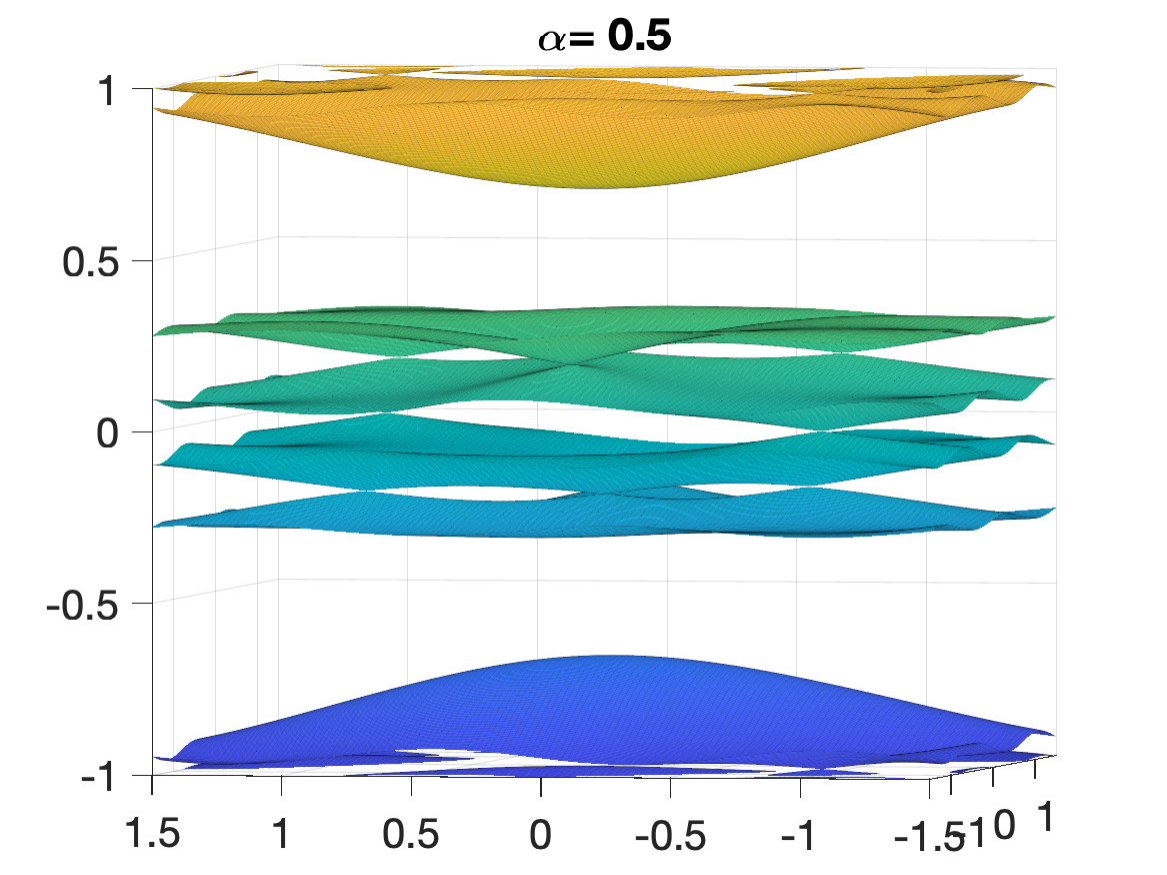}
\includegraphics[width=7.5cm]{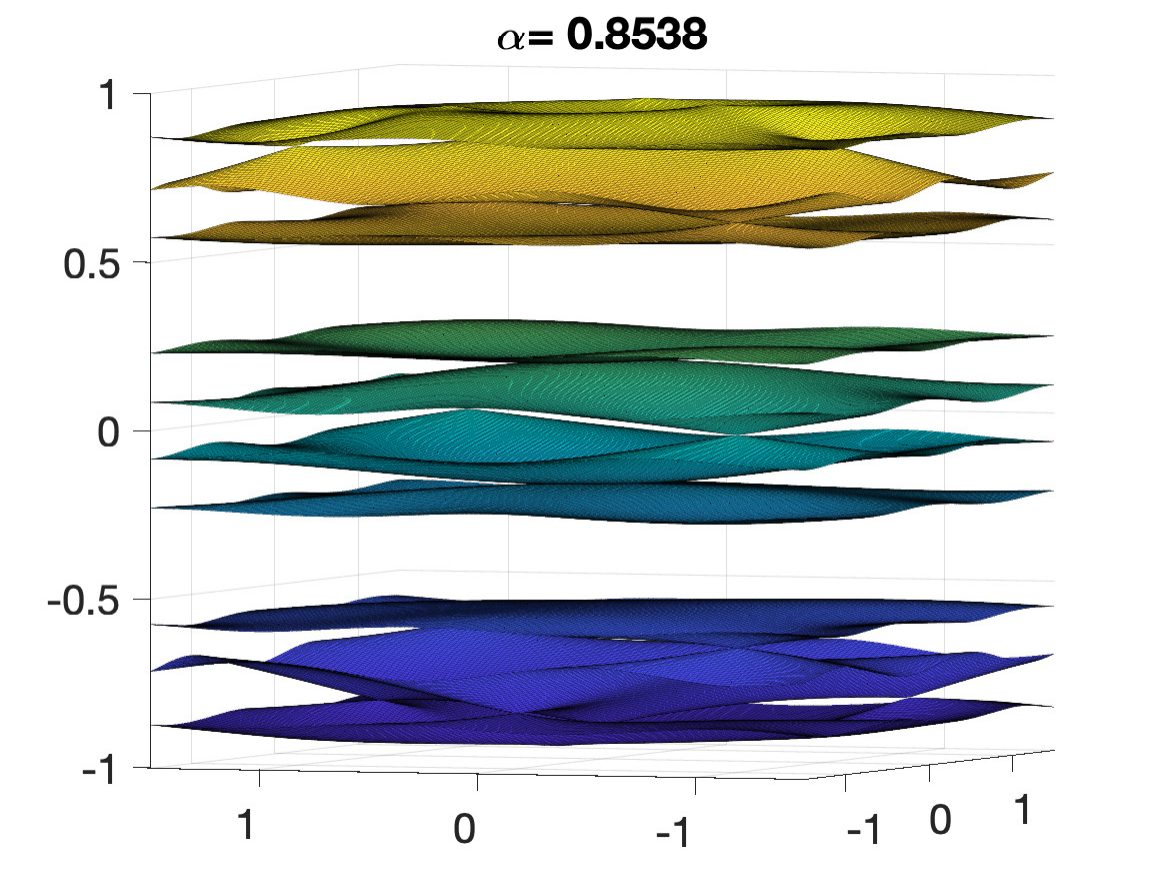}
\caption{Bands of Hamiltonian \eqref{eq:RescPot} at $\alpha=0.5$ (left) and $\alpha=0.8538$ (right) with potential $U_{\pm}:=U_2(\pm \bullet)$ \eqref{eq:potential}. 
Bands of full continuum Bistritzer-MacDonald Hamiltonian \cite[(1)]{suppl} with same $\alpha$ and $\beta = 0.7 \alpha$ with potential $U_{\pm}= U_2 (\pm \bullet)$ \eqref{eq:potential} and anti-chiral potential $V(z) := 2\partial_{z} U_2(z)$. }
\label{fig:beyond_chiral}
\end{figure}

\begin{proof}[Proof of Theorem \ref{theo:rigidity}]
If $ \Omega u ( z ) := u ( \omega z ) $ then 
\begin{equation}
\label{eq:rots}    \Omega D ( \alpha )  = \omega D ( \alpha ) \Omega , \ \ \  
\ker_{L^2_0}(D(\alpha)) = \bigoplus_{p \in \mathbb Z_3} \ker_{L^2_{0,p}}(D(\alpha)).
\end{equation}
Theorem \ref{p:zeros} shows that all non-zero elements of $\ker_{L^2_0}(D(\alpha))$ have $m(\alpha)$ zeros counting multiplicities. 
That allows us to consider the invariant subspaces $L^2_{0,p}$, separately. 

Let $u \in \ker_{L^2_{0,0}}(D(\alpha))$. Since $\omega z_S = z_S -(1+\omega),$ we find
\[\begin{split} u(\pm z_S + \zeta ) &=  u(\pm \omega z_S +  \omega \zeta ) = u(\pm z_S \mp (1+\omega) + \omega \zeta ) \\
&= \operatorname{diag}(e^{-i\langle \mp(1+\omega),K \rangle},e^{i\langle \mp(1+\omega),K \rangle}) \mathscr L_{\mp (1+\omega)} u(\pm z_S + \omega \zeta), \\
&= \operatorname{diag}(\omega^{\pm 1}, \omega^{\mp 1}) u(\pm z_S +  \omega \zeta ) 
\end{split}\]
that, is  $ u ( \pm z_S + \omega \zeta ) =\operatorname{diag}(\omega^{\mp 1}, \omega^{\pm 1})
u ( \pm z_S + \zeta ) $. This shows that 
\[  M_{u}(\pm z_S) \not \equiv 0 \!\!\! \mod 3. \]
  We now recall the symmetry $\mathscr E:L^2_{0,p}(\CC/\Lambda; \CC^2) \to L^2_{0,p}(\CC/\Lambda; \CC^2)$ (valid under the assumption \eqref{eq:defU}):
  \begin{equation}
  \label{eq:Esy} \mathscr E D(\alpha) \mathscr E^* = -D(\alpha) , \ \ \ \mathscr Ev(z):=Jv(-z), \ \ \  J:=\begin{pmatrix} \ \ 0 & 1 \\ -1 & 0 \end{pmatrix}. \end{equation}
It follows that  $ M_{ u}(z_S) = M_{ u}(-z_S) $. 
 In particular, any element of $\ker_{L^2_{0,0}}(D(\alpha))$ has zeros at $z=\pm z_S$ (for both signs). 
    Zeros at any other point are three-fold degenerate by rotational symmetry \eqref{eq:rots} and that 
     shows that 
    \begin{equation}
    \label{eq:mod123}
\sum_{z \in \mathbb{C} / \Lambda} M_{{u}}(z) \bmod 3 \in \{1,2\}.
    \end{equation}

The same conclusion holds for the subspace $\ker_{L^2_{0,1}}(D(\alpha))$  as can been seen using the
properties of the Weierstra{\ss}    $ \wp $-function, $ \wp ( z ):= \wp ( z; \omega , 1 ) $ \cite[\S I.6]{tata}:
\begin{equation*}
\wp ( \omega z ) = \omega \wp (z )  \ \text { and } \ 
  \wp ( z ) = 0  \ \Longrightarrow \ z = \pm z_S + \Lambda,  \ \ \wp' ( \pm z_S ) \neq 0 .
  \end{equation*} 
This shows that     
\begin{equation}
\label{eq:wpz} \wp ( z )\ker_{L^2_{0,1}(\CC/\Lambda; \CC^2)}(D(\alpha)) =  \ker_{L^2_{0,0}(\CC/\Lambda; \CC^2)}(D(\alpha)) \end{equation} 
and hence \eqref{eq:mod123} holds for elements of $\ker_{L^2_{0,1}}(D(\alpha))$ as well.

Finally, suppose that $u \in \ker_{L^2_{0,2}}(D(\alpha))$. Since then 
$ u ( \omega z ) = \omega u ( z ) $ (see \eqref{eq:rots}) 
we see that 
$ M_u(0) \equiv 1\!\!\!\  \mod 3$. 
 As above, we find
\[\begin{split} u(\pm z_S + \zeta ) &=\bar \omega  u(\pm \omega z_S +  \omega \zeta ) = 
\bar \omega u(\pm z_S \mp (1+\omega) + \omega \zeta) \\
&=\bar \omega \operatorname{diag}(e^{-i\langle \mp(1+\omega),K \rangle},e^{i\langle \mp(1+\omega),K \rangle}) \mathscr L_{\mp (1+\omega)} u(\pm z_S +  \omega \zeta) \\
&=\bar \omega \operatorname{diag}(\omega^{\pm 1}, \omega^{\mp 1}) u(\pm z_S + \omega \zeta) 
\\ &=  \operatorname{diag}(\omega^{\frac12 ( 1 \mp 1) }, \omega^{\frac12 ( 1 \pm 1) }) u(\pm z_S 
+ \omega \zeta) 
,\end{split}\]
that is $ u ( \pm z_S + \omega \zeta ) = 
\operatorname{diag}(\omega^{\frac12 ( \pm 1 -1 ) }, \omega^{\frac12 (\mp 1 - 1) }) 
u ( \pm z_S +  \zeta )
$.  Hence if $ u ( z_S + \zeta ) = ( u_1 ( \zeta ) , u_2 ( \zeta ) )^t $ 
then $  u_1 ( \omega \zeta ) = u_1 ( \zeta ) $ and $ u_2 ( \omega \zeta ) = \bar \omega u_2 ( \zeta ) $
and the order of the zero of $ u_1 $ is 0 mod 3 and that of $ u_2 $ is 2 mod 3. This shows that 
$ M_{u}(\pm z_S)  \bmod 3 \in \{0,2\} $. 
 Using the $\mathscr E$ symmetry \eqref{eq:Esy}, we conclude that 
 \begin{equation}
 \label{eq:Muz} M_{u}( z_S)  + M_{u} ( - z_S ) \not \equiv 2 \! \!\! \mod 3 . \end{equation}
 Other points $z \notin \{0,\pm z_S\}$ satisfy $M_u(z)=M_{u}(\omega z) = M_{u}(\omega^2 z)$ by rotational symmetry. Combining this observation with the fact that the multiplicity as $ 0$ is 1 mod 3 and \eqref{eq:Muz} we obtain 
    \[
\sum_{z \in \mathbb{C} / \Lambda} M_{{u}}(z) \not \equiv 0 \! \!\!  \mod 3 . 
    \]
In all three cases, the total multiplicity of zeros is not divisible by $3$, ruling out the case $m(\alpha)  \equiv 0 \bmod 3 $. 

We also see that if $ m ( \alpha ) = 1 $ then $ \dim \ker_{L^2_{0,2} } D ( \alpha ) = 1 $, as otherwise
\eqref{eq:wpz} would imply $ m ( \alpha ) > 1 $. If $ m (\alpha ) = 2 $ the only possibility is 
$ \dim \ker_{ L^2_{0,0} } D ( \alpha ) = \dim \ker_{L^2_{0,1} } D ( \alpha ) = 1 $. 
\end{proof}

The location of zeros for flat band Bloch functions at a degenerate magic angle are illustrated in Figures \ref{fig:zeros} and \ref{fig:zeros2}.

\section{Trace computations}
\label{s:trace}

To prove the existence of degenerate magic angles (Theorem \ref{mult}) we argue by contradiction
using the Birman--Schwinger operator $ T_k $ defined in \eqref{eq:BS}. 
From Theorem \ref{theo:rigidity}, we see that in the case if all the $ \alpha $'s were all simple 
then the traces of $ T_k^{2p}  $ restricted to $ L^2_{0,0} $ or $ L^2_{0,1} $ would have to 
vanish.  For a general $ k $, the operator $T_k$ does \emph{not} preserve the rotational invariant subspaces $L^2_{0,j}$.  To achieve that we set $ k = 0 $ so that the proof reduces to showing that $\tr((T_0)_{L^2_{0,0}}^{2\ell})\neq 0$ for some value of $\ell$. That is done using the previous rationality condition $\tr((T_0)_{L^2_{0,0}}^{2\ell})=q_{\ell}\pi/\sqrt 3$ for $q_{\ell}\in \mathbb Q$ obtained before by the authors \cite[Theorem 1]{bhz1} and some elementary arguments involving transcendental numbers. 

\subsection{Traces on rotationally invariant subspaces}
\label{tracessection}
 We recall that an orthonormal basis of $L_{0}^2(\mathbb C/3\Lambda; \mathbb C)$ is given by setting 
$$e_{\nu}(z):=e^{i\langle \nu,z\rangle}/\sqrt{\mathrm{Vol}(\mathbb C/3\Lambda)},\quad \nu \in \Lambda^*+K,\ \ \langle \nu,z\rangle:=\mathrm{Re}(\bar z \nu).  $$
We see that $ \Omega e_{\nu}=e_{\bar \omega \nu}.$ This means that an orthonormal basis of $L^2_{{0,j}}$ is given by
$$e_{[\nu]}(z)=\frac{1}{\sqrt3}\big(e_{\nu}(z)+\omega^j e_{\omega\nu}(z)+\bar\omega^{j}e_{\bar\omega \nu}(z)\big),\;\; \nu \in\Lambda^*+K,\;\; [\nu]=\{\nu,\omega\nu,\bar\omega\nu\}.$$
Following our approach developed in \cite{bhz1}, we compute the sum of powers of magic angles by computing traces of the operator $T_k$ defined in \eqref{eq:defTk}. Since odd powers of $T_k$ have vanishing traces it suffices to compute the traces of powers of the Hilbert-Schmidt operator
\begin{equation}
\label{eq:Ak}
A_k:=R(k)U(z)R(k)U(-z):L^2_{0}\to L^2_{0},\quad k\notin ( K +  \Lambda^*) 
\cup ( - K + \Lambda^* ) =\mathcal K_0.
\end{equation}
Due to the relation
$$\forall k\notin \mathcal K_0,\quad \Omega^{-1} A_k\Omega=A_{\omega k}, $$
we see that subspaces $L^2_{0,j}$ are \emph{not} in general invariant by $A_k$. This makes a direct 
application of the strategy of \cite{bhz1} impossible. However, we see that the operator $A_0$ does preserve this smaller subspace. From now on, we therefore specialize to $k=0$.
For $\ell\geq 2$, one can compute the trace on $L^2_{{0,j}}$:
$$\tr\big((A_0)^{\ell}_{|L^2_{{0,j}}}\big)=\sum_{[\nu],\nu \in \Lambda^*+K}\langle A_0^{\ell}e_{[\nu]},e_{[\nu]}\rangle. $$
Now, we write that, using bilinearity of the scalar product
$$3\langle A_0^{\ell}e_{[\nu]},e_{[\nu]}\rangle=\sum_{h=0}^2\langle A_0^{\ell}e_{\omega^h\nu},e_{\omega^h\nu}\rangle+\sum_{k\neq h}\omega^{j(k-h)}\langle A_0^{\ell}e_{\omega^h\nu},e_{\omega^k\nu}\rangle.$$
Thus, when summing on $[\nu]$, the first term gives a third of the trace on $L^2_{0}$, (which was computed in \cite{beta} for $\ell=2$ and $U_0=U_1$ and shown to be equal to $4\pi/\sqrt 3$)
\begin{equation}
\label{eq:Traces}
\begin{split} 
\tr((A_0)^{\ell}_{|L^2_{{0,j}}}) & =\tfrac 1 3\tr(A_0^{\ell})+\tfrac 1 3\sum_{[\nu],\nu \in \Lambda^*+K}\sum_{k\neq h}\omega^{j(k-h)}\langle A_0^{\ell}e_{\omega^h\nu},e_{\omega^k\nu}\rangle\\
& =:\tfrac 1 3\tr(A_0^2)+\tfrac 1 3\mathcal R_{\ell,j}. \end{split} 
\end{equation}
\subsection{Existence of degenerate magic angles}
Our strategy now consists in using \cite[Theorem1]{bhz1} and the fact that $\pi/\sqrt 3$ is transcendental to contradict the conclusion of Theorem \ref{theo:rigidity}. More explicitly, we will prove the following statement:

\begin{theo}
\label{General}
We consider the Hamiltonian \eqref{eq:RescPot} with a potential $U \in C^{\infty}(\mathbb C/3\Lambda) $ and  $U_+:=U$ and $U_-:=U(-\bullet)$ satisfying the first two symmetries of \eqref{eq:news} with only finitely many non-zero Fourier modes $a_p\in \pi \mathbb Q(\omega/\sqrt 3)$, appearing in the decomposition \eqref{eq:Fourier}. Then, if we denote $\mathcal A(U)$ the set of $($complex$)$ magic angles for the potential $U$ and if $\mathcal A(U)\neq \emptyset$, there exists $\alpha \in \mathcal A(U)$ which is not simple. This also applies to the Bistritzer–MacDonald potential $U_+:=U_1$ defined in \eqref{eq:potential}.
\end{theo}
\begin{proof}
\underline{Step 1 (Existence of non-zero trace):} We start by noticing that the existence of a magic angle is equivalent to the existence of a non-vanishing trace 
$$\exists \ell\geq 2,\quad \tr((A_0)_{|L^2_0}^{\ell})\neq 0. $$
This follows from the properties of the regularized Fredholm determinant, cf. \cite{bhz1}. 

\noindent
\underline{Step 2 (Trace is transcendental):}  We fix such an $\ell$ for which $\tr((A_0)_{|L^2_0}^{\ell})\neq 0$. Using \cite[Theorem $5$]{bhz1}, and the hypothesis on the potential, this implies that $\tr(A_0^{\ell})\in \pi \mathbb Q(\omega)$. Since the trace is non-zero by assumption, this proves that $\tr(A_0^{\ell})$ is transcendental.  

\noindent
\underline{Step 3 ($\mathcal R_{\ell,j}$ is a finite sum):} We now show that the sum defining the remainder $\mathcal R_{\ell,j}$ in \eqref{eq:Traces} is always a finite sum, under the assumption that the potential has only finitely many non-zero Fourier mode. 
We start with the formula defining the remainder
$$\mathcal R_{\ell,j}:= \sum_{[\nu],\nu \in \Lambda^*+K}\sum_{k\neq h} \omega^{j(k-h)}\langle A_0^{\ell}e_{\omega^h\nu},e_{\omega^k\nu}\rangle.$$
The summand $\langle A_0^{\ell}e_{\omega^h\nu},e_{\omega^k\nu}\rangle$ is non-zero only if $ A_0^{\ell}e_{\omega^h\nu}$ has a non-vanishing Fourier mode corresponding to $e_{\omega^k\nu}$. Now, if we look at the definition of $A_0$ (see \ref{eq:Ak}), we see that the $R(k)$ part acts diagonally (with coefficients in $ ( i \pi)^{-1} \mathbb Q(\omega)$ as we chose $k=0$) on the Fourier basis, on the other hand, the $U(z)$ and $U(-z)$ parts act as a finite sum of weighted shifts on this basis (it is here where we use the assumption of having finitely many non-vanishing Fourier modes). Moreover, by assumption, the weights are elements of $(i\pi)\mathbb Q(\omega)$.This means that there exists a finite subset $\mathcal F_U^{\ell}\subset 3\Gamma^*$ such that
\begin{equation}
\label{eq:a_eta}
\forall \nu \in \Lambda^*+K,\quad A_0^{\ell}e_{\nu}=\sum_{\eta\in \mathcal F_U^{\ell}} a_{\eta}e_{\nu+\eta},\quad a_{\eta}\in \mathbb Q(\omega).  
\end{equation}
But this means that there exists a constant $R>0$ such that for any $\eta \in \mathcal F_U^{\ell}$, we have $|\eta|\leq R$. In particular, if $\langle A_0^{\ell}e_{\omega^h\nu},e_{\omega^k\nu}\rangle$ is non-zero, then $|\omega^h\nu-\omega^k\nu|\leq R$. Now, because $h\neq j$, this inequality is false outside a compact set for $\nu$. But because $\nu$ is on a lattice, which is discrete, we conclude that the above inequality is true for at most a finite number of $\nu$. Thus, the sum defining $\mathcal R_{\ell,j}$ is finite.

\noindent
\underline{Step 4 $\mathcal R_{\ell,j} \in \mathbb Q(\omega)$:} Finally, for the non-zero terms of the sum, we use \eqref{eq:a_eta} again to conclude that $\langle A_0^{\ell}e_{\omega^h\nu},e_{\omega^k\nu}\rangle=a_{\eta}\in \mathbb Q(\omega).$ This proves that $\mathcal R_{\ell,j} \in \mathbb Q(\omega).$

\noindent
\underline{Step 5 (Proof by contradiction):} Since $\mathcal R_{\ell,j}\in \mathbb Q(\omega)$ is algebraic and thus $\tr((A_0)^{\ell}_{|L^2_{{0,j}}})\neq 0$ by \eqref{eq:Traces}. This contradicts the conclusion of Theorem \ref{theo:rigidity}; thus proving the existence of non-simple magic angle for the potential $U$.
\end{proof}

\section{Infinite number of degenerate magic angles}
\label{s:inf}

We now adapt the argument, already used in \cite[Theorem $6$]{bhz1}, to prove that the number of non-simple magic angles is actually infinite. This actually refines the previous theorem by showing there is an infinite number of non-simple magic angles.

In the next theorem we use the same notation and assumptions as in Theorem \ref{General}. The 
definition of multiplicity is given in \eqref{eq:U2A}. 
\begin{theo}
\label{sec:races_to_infinity}
Let 
\[  \mathcal A_m(U) := \{ \alpha \in \mathcal A ( U ) : m_U ( \alpha ) \geq 2 \} \]
be  the set of non-simple magic angles. Then
\begin{equation}
\label{eq:implication} 
|\mathcal A(U)|>0 \
  \Longrightarrow \  |\mathcal A_m(U)|=+\infty. 
\end{equation}
In particular, the set of magic angles for the Bistritzer--MacDonald potential $ U =U_1$ (see \eqref{eq:potential}) is infinite.

In addition, if for $N\geq 0$, and  $a=(a_{p})_{\{p\in \Lambda^*; \Vert p \Vert_{\infty} \leq N\}}$,  $U_a$ is given by \eqref{eq:Fourier} with coefficients $a$, then \eqref{eq:implication} holds for a generic $($in the sense of Baire$)$ set of coefficients $a=(a_{p})_{\{p\in \Lambda^*; \Vert p \Vert_{\infty} \leq N\}}\in \mathbb C^{(2N+1)^2}$ which contains $(\pi\mathbb Q(\omega/\sqrt{3}))^{(2N+1)^2}.$ Here, we used the notation $\|p\|_{\infty}=\|\tfrac{4\pi i}{\sqrt 3}(p_1+p_2\omega)\|_{\infty}:=\max(p_1,p_2).$
\end{theo}
\begin{proof}
We start by observing that since $\pi$ is transcendental on $\mathbb Q$, it is also transcendental in $\mathbb Q(\omega/\sqrt{3}).$
Now, we shall assume that there exist only finitely many non-simple eigenvalues of $A_{0}^2$ on $L^2_{0}$. This implies, by Theorem \ref{theo:rigidity} that $(A_0)^{\ell}_{|L^2_{{0,1}}}$ has only finitely many eigenvalues, we denote them by $\lambda_i \in \mathbb C$ for $i=1,..,N$.  Then we define the $n$-th symmetric polynomial
\[ e_n (\lambda_1 , \ldots , \lambda_N )=\sum_{1\le  j_1 < j_2 < \cdots < j_n \le N} \lambda_{j_1} \dotsm \lambda_{j_n}.\]
Newton identities show that this polynomial can be expressed as 
\begin{equation}
\label{eq:Newton} e_n(\lambda_1 , \ldots , \lambda_N ) = (-1)^n  \sum_{m_1 + 2m_2 + \cdots + nm_n = n \atop m_1 \ge 0, \ldots, m_n \ge 0} \prod_{i=1}^n \frac{(-\tr (A_0)^{2i}_{|L^2_{{0,1}}})^{m_i}}{m_i ! i^{m_i}}\end{equation}
where $e_n=0$ for $n>N.$ 
The fact that $\mathcal A(U)\neq \emptyset$ implies, by Theorem \ref{General} that $\mathcal A_m(U)\neq \emptyset$. Now, this means that there is a non-vanishing trace of $(A_0)^{\ell}_{|L^2_{{0,2}}}$. Choose $m_0$ to be the minimal power for which the trace is non-zero. Choose $n=m_0\times K$ where $K$ is a large integer, and using the fact that $e_n=0$, we deduce that $\pi$ is the root of the polynomial of degree $K$ with coefficients in $\mathbb Q \left(\frac{\omega}{\sqrt{3}}\right)$ given by
 \[ \begin{split} \sum_{m_1 + 2m_2 + \cdots + nm_n = n \atop m_1 \ge 0, \ldots, m_n \ge 0} \prod_{i=1}^{m_0\times K} (\tr(A_0)^{2i}_{|L^2_{{0,1}}})^{m_i} & = \sum_{m_1 + 2m_2 + \cdots + nm_n = n \atop m_1 \ge 0, \ldots, m_n \ge 0} \prod_{i=1}^{m_0\times K} (\underbrace{\frac 13\tr(A_0)^{2i}_{|L^2_{{0}}}}_{\in\mathbb Q \left(\frac{\omega}{\sqrt{3}}\right)\pi}-\underbrace{\mathcal R_{i,1}}_{\in \mathbb Q \left(\frac{\omega}{\sqrt{3}}\right)})^{m_i}\\
 & =0 . \end{split} \]

The power $m_1 \cdots m_n$ of $\pi$ is maximized, among the tuples we sum by the unique choice $m_i=\delta_{i,m_0}K$. By choice of $m_0$, this gives that the above polynomial has a non-zero leading coefficient and is therefore non-zero. This contradicts the fact that $\pi$ is transcendental and concludes the proof.

Now, let $a=(a_{p})_{\{p\in \Lambda^*; \Vert p \Vert_{\infty} \leq N\}}\in \mathbb C^{(2N+1)^2} \in \mathbb C^{(2N+1)^2}$ and assume that $\mathcal A(U_a) \neq \emptyset$. Then, we can find an open neighborhood of $ a $,  $\Omega_a \ni a$, such that for coefficients $b=(b_{p})_{\{p\in \Lambda^*;\Vert p \Vert_{\infty} \le N\}} \in \Omega_a$ we have $\mathcal A(U_b) \neq \emptyset$. 
Take  $q=(q_{p})_{\{p\in \Lambda^*;\Vert p \Vert_{\infty} \le N\}} \in (\pi\mathbb Q(\omega/\sqrt{3}))^{(2N+1)^2}\cap \Omega_a$ 
for which we then have $\vert \mathcal A(U_q) \vert = \infty.$ 
Continuity of eigenvalues of $T_k $ as the potential $ U $ changes shows that the 
$ V_{m,a} := \{ b \in \Omega_a :  | \mathcal A(U_{b}) | \geq m \}  $
is open and dense in $\Omega_a$. Hence, the set 
coefficients for which $ 0 < | \mathcal A(U_b) | < \infty $ is given by $\bigcup_{m \in \mathbb N} \bigcup_{q \in  ( \mathbb Q+i\mathbb Q)^{2N+1}} \Omega_q \setminus V_{m,q}$. It is then meagre and does not contain $(\pi\mathbb Q(\omega/\sqrt{3}))^{(2N+1)^2}.$
\end{proof}

\section{Numerical evaluation of the trace and existence of non-real magic angle} 
In this section the potential $U_{\pm}$ will be taken to be equal to $U_1(\pm \bullet)$ defined in \eqref{eq:potential}. In \eqref{eq:Traces}, we have proven that the traces on the rotational-invariant subspaces can be written as
\begin{equation}
\label{eq:UU}
\tr((A_0)^{\ell}_{|L^2_{{0,j}}})  =\tfrac 1 3\tr(A_0^{\ell})+\tfrac 1 3\mathcal R_{\ell,j},
\end{equation}
where the remainder was shown to be a finite sum. Although the first term $\tr(A_0^{\ell})$ is a priori an infinite sum, the authors provided in \cite[Theo. 7]{bhz1} a semi-explicit formula which can be evaluated rigorously with computer assistance for $U=U_1$ and small values of $\ell$. From \cite[Table 1]{bhz1}\footnote{The traces $\tr(A_0^{2})$ and $\tr(A_0^{4})$ were explicitly computed "by hand" in \cite{beta} and strictly speaking, the following argument relies on computer assistance only for obtaining the exact value of $\tr(A_0^{3})$.}, we see that
$$\tr((A_0^{2})_{|L^2_0})=\frac{4\pi}{\sqrt 3}, \quad  \tr((A_0^{3})_{|L^2_0})=\frac{96\pi}{7\sqrt 3}, \quad\tr((A_0^{4})_{|L^2_0})=\frac{40\pi}{\sqrt 3}.$$
We can read off from the above  $\tr(A_0^{2})\tr(A_0^{4})<\tr(A_0^{3})^2$. If all magic angles were real, then by $\ell^p$-interpolation 
$\tr(A_0^{2})\tr(A_0^{4})\geq \tr(A_0^{3})^2$, which is a contradiction. In other words, we have proven that
\begin{prop}
Let $U=U_1$ be the potential defined in \eqref{eq:potential}, then $\mathcal A \cap \CC\setminus \RR \neq \emptyset.$
\end{prop}
Our goal here is to mimic this argument on rotational-invariant subspace by computing the finite remainders $\mathcal R_{\ell,j}$ using computer assistance to find the exact results.

From doing so, we obtain the following result.
\begin{prop}
\label{prop:Tristanconvex}
For the Bistritzer-MacDonald potential $U_1$ defined in \eqref{eq:potential}, we have 
$$ \tr((A_0^{2})_{|L^2_{0,1}})=\tr((A_0^{2})_{|L^2_{0,0}}) = \frac{4\pi}{3\sqrt{3}}-3\approx -0.581601<0$$ 
and $\tr((A_0^{2})_{|L^2_{0,2}}) = \frac{4\pi}{\sqrt{3}}+6 \approx 8.4184.$ For the higher powers, we find 
\[ \tr((A_0^{3})_{|L^2_{0,2}}) = \frac{32 \pi}{7 \sqrt{3}} + \frac{810}{49}  \approx 24.8223 \text{ and }  \tr((A_0^{4})_{|L^2_{0,2}}) = \frac{40 \pi}{3 \sqrt{3}} + \frac{4374}{91}  \approx 72.2499.\]
This implies the inequality
$$ \tr((A_0^{2})_{|L^2_{0,2}})\tr((A_0^{4})_{|L^2_{0,2}})<\tr((A_0^{3})_{|L^2_{0,2}})^2.$$
We conclude that for any $j\in \mathbb Z_3$, there is a non-real magic angle $\alpha_j \in \CC\setminus \RR$ with corresponding eigenfunction $u\in L^2_{0,j}$ of $T_k$. By Theorem \ref{theo:rigidity}, we conclude the existence of non-real and non-simple magic angles.
\end{prop}
We note that as the traces depend continuously on the potential $U$, the inequalities $$\tr((A_0^{2})_{|L^2_{0,1}})=\tr((A_0^{2})_{|L^2_{0,0}})<0 \text{ and }\tr((A_0^{2})_{|L^2_{0,2}})\tr((A_0^{4})_{|L^2_{0,2}})<\tr((A_0^{3})_{|L^2_{0,2}})^2$$ remain true for small perturbations of $U$ and so does the existence of a non-real and non-simple magic angle. As stated in the introduction, the potential $U_2$, defined in \eqref{eq:potential}, leads to real and doubly-degenerate magic angles. We then see numerically that $\tr((A_0^{2})_{|L^2_{0,1}})=\tr((A_0^{2})_{|L^2_{0,0}})>0$, see Figure \ref{fig:traces}. To interpolate between these two opposite behaviors, we introduce the potentials 
\begin{equation}
\label{eq:Utheta}
U_{\theta}(z):= U ( z ) = ( \cos \theta -\sin \theta )  U_1 ( z ) + \sin \theta U_2 ( z ), 
\end{equation}
see \url{https://math.berkeley.edu/~zworski/Interpolation.mp4} for a movie showing the dependence of the set of magic angle when $\theta$ varies.

In Figure \ref{fig:traces} we show $\tr((A_0^{2})_{|L^2_{0,0}}), \tr((A_0^{2})_{|L^2_{0,2}})$ as a function of $\theta$, verifying that the inequality $\tr((A_0^{2})_{|L^2_{0,0}})<0$ holds for a large range of values $\theta$.

\begin{figure}
\includegraphics[width=10cm]{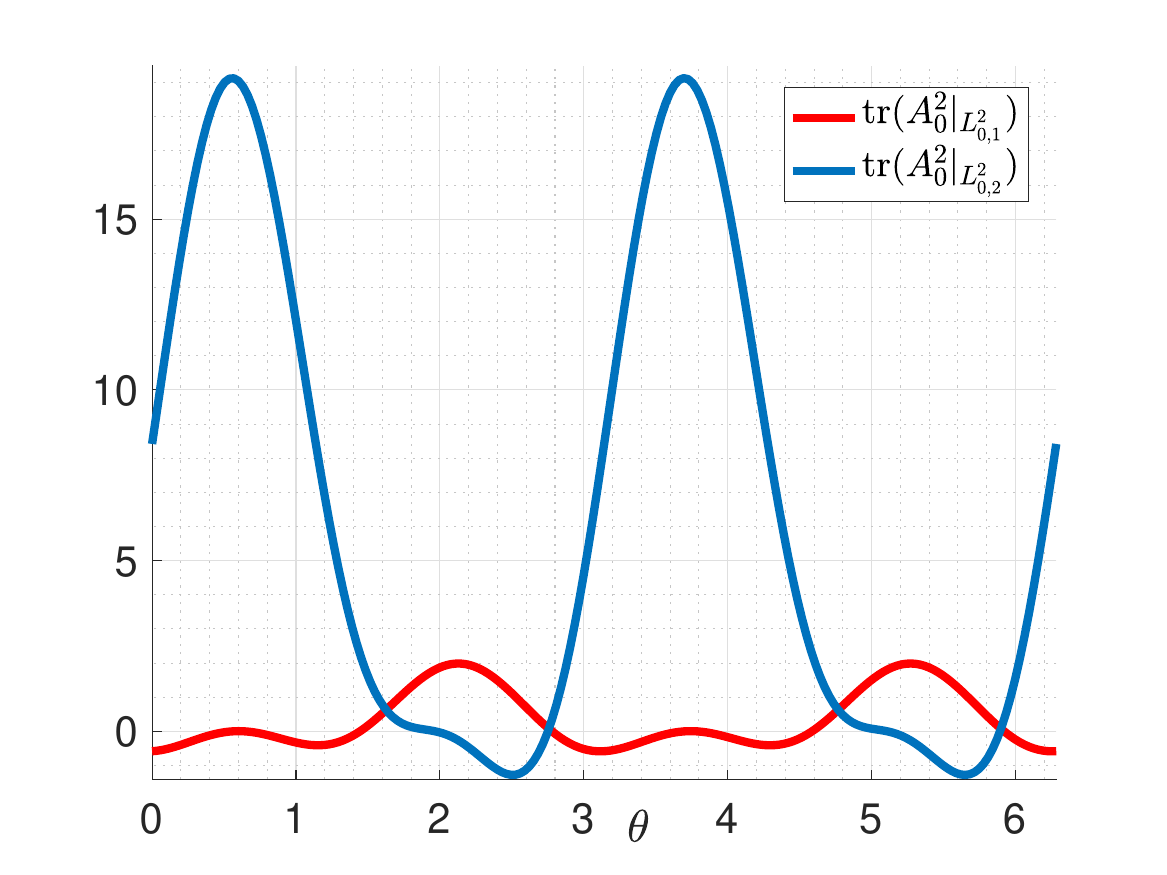}
\caption{\label{fig:traces}${\color{red}\tr((A_0^{2})_{|L^2_{0,1}})}$ and ${\color{blue}\tr((A_0^{2})_{|L^2_{0,2}})}$ for potentials $U_{\pm}(z):=U_{\theta}(\pm z)$ in \eqref{eq:Utheta}. While for $\theta=0$, $U_{\theta=0}=U_1$ we see that $\tr((A_0^{2})_{|L^2_{0,2}})>0$ and $\tr((A_0^{2})_{|L^2_{0,1}})<0$. For $\theta = 2\pi 7/8 \approx 5.5$ and $U_{\theta=2\pi 7/8}=U_2$ we have $\tr((A_0^{2})_{|L^2_{0,2}})<0$ and $\tr((A_0^{2})_{|L^2_{0,1}})>0$, instead.  }
\end{figure}

\noindent
{\bf Remark.}
This previous computation could be made rigorous at the cost of adapting the algorithm used in \cite[Theo. 7]{bhz1} to the potential $U_{\theta}$ in order to compute the first term in \eqref{eq:UU}.

\section{Generic simplicity in each representation}
\label{s:gensimp}

\subsection{Generalized potentials}
We now consider the general class of potentials $U_{\pm}(z)$ satisfying 
\begin{equation}
\label{eq:newU}
U_{\pm}(\omega z) = \omega U_{\pm}(z), U_{\pm}(z+\gamma) = e^{\mp 2i \langle \gamma, K \rangle} U_{\pm}(z), \quad \gamma \in \Gamma. 
\end{equation}
We do
{\em not} however assume $\overline{U_{\pm}(\bar z)} = - U_{\pm}(z)$ and then define 
\[ V(z) := \begin{pmatrix} 0 & U_{+}(z) \\ U_-(z) &0 \end{pmatrix} \text{ such that } D_V(\alpha) = 2D_{\bar z} + \alpha V(z).\]
 
It is convenient to use the following Hilbert space of {\em real analytic} potentials defined
using the following norm: for fixed $ \delta > 0 $, 
\begin{equation}
\label{eq:norm}
\| V \|_\delta^2 := 
\sum_{\pm} \sum_{ k \in \Lambda^*/3 } |a^\pm_{k } |^2 e^{ 2 | k| \delta}, \   \ \ \ 
U_{\pm} ( z ) = \sum_{  k \in K + \Lambda^* } a^\pm_{k } e^{ \pm i \langle z ,k \rangle } . 
\end{equation}
Then we define $ \mathscr V = \mathscr V_\delta $ by 
\begin{equation}
\label{eq:Vscr}
V \in \mathscr V \ \Longleftrightarrow \ 
\text{ $ V$ satisfies \eqref{eq:newU}, } \ 
\| V \|_{\delta } < \infty . 
\end{equation}
We note that we have as before,
\[ \mathscr L_{\mathbf a } D_V ( \alpha ) = D_V ( \alpha)  \mathscr L_{\mathbf a } ,  \ \ \ 
\Omega D_V ( \alpha ) = D_V ( \alpha ) \Omega .\]
We also recall the antilinear symmetry $\mathscr A: L^2_{k,j}\to L^2_{k,-j}$ defined by
\begin{equation}
\label{eq:defAsc}  \mathscr A :=  \begin{pmatrix} \ \   0 & \Gamma  \\ - \Gamma & 0  \end{pmatrix}, \ \ 
\Gamma v ( z ) = \overline{ v ( z ) } , \ \ \ \ 
\mathscr A D_V ( \alpha ) \mathscr A =   {-}D_V ( \alpha )^* . \end{equation}

\subsection{Proof of generic simplicity}

Our proof of Theorem \ref{t:sim} is an adaptation of the argument for generic simplicity of resonances
by Klopp--Zworski \cite{klop} -- see also \cite[\S 4.5.5]{res}.

We then use the decomposition
\[ L^2_{0} = \bigoplus_{j=0}^2 L^2_{{0,j}} , \ \ \ 
L^2_{{0,j} } \simeq L^2 ( F ) ,\]
where $ F $ is a fixed fundamental domain of $ G_3 $. 
For $ V \in \mathscr V $ and $R = (2D_{\bar z})^{-1}$
\[    V : L^2_{{0,j} } \to L^2_{{0,j-1} } , \ \ \ R :  L^2_{0,j-1 } \to L^2_{0,j } 
\ \Longrightarrow \ R V :  L^2_{0,j } \to L^2_{0,j } . 
 \]

Before proceeding we record the following regularity result:
\begin{lemm} 
\label{l:anal}
Suppose that for some $ \lambda \in \mathbb C  $ and $ k \in \mathbb N $ and
$ w \in L^2 ( \mathbb C/3\Lambda; \mathbb C ) $, 
  $ ( R V - \lambda)^k w = 0 $. 
Then $ w \in C^\omega ( \mathbb C / 3\Lambda; \mathbb C) $, that is, $ w $ is {\em real analytic}.
The same conclusion holds if $ ( V^* R^* - \lambda )^k w = 0 $. 
\end{lemm}
\begin{proof}
We prove a slightly more general statement that $ 
( RV - \lambda)^k w = f \in C^\omega (\mathbb C/3\Lambda ; \mathbb C^2 ) $ 
implies that 
$ w \in C^\omega ( \mathbb C / 3\Lambda; \mathbb C^2  ) $. We proceed by induction on 
$ k $. For $ k = 0 $, $ w = f $. If $ k > 0 $, we put 
$ \widetilde {w } := ( RV - \lambda )^{k-1} {w } $ and note that (the case of  $ \lambda = 0 $ 
is even simpler)
\[   D_V ( - 1/\lambda )  \widetilde {w } = 
2\lambda^{-1} D_{\bar z }  ( R V - \lambda ) \widetilde {w } = 
2 \lambda^{-1}   D_{\bar z }  { f }  \in C^\omega . \]
This means that $ \widetilde {w } $ is a solution of
an elliptic equation with analytic coefficients, hence it is analytic \cite[Theorem 9.5.1]{H1}. 
The inductive hypothesis now shows that $ w $ is analytic. 

In the case of $ ( V^* R^* - \lambda )^k {w } = 0 $, we proceed similarly but
put $ \widetilde w := R^*  ( V^* R^* - \lambda )^{k-1} {w } $, so that
\[  D_V ( - 1/ \bar \lambda )^*  \widetilde {w } = 
2  \bar \lambda^{-1}  D_z  R^* ( V^* R^* - \lambda ) 
(V^* R^* - \lambda )^{k-1} {w } = 2 \bar \lambda^{-1} D_z  R^* { f}  \in C^\omega. \]
Since $ ( V^* R^* - \lambda)^{k-1} w = 2 D_z  \widetilde {w} $ the inductive
argument proceeds as before.
\end{proof} 

The next lemma shows that we have generic simplicity for operators restricted
to the three representations:
\begin{lemm}
\label{l:sim}
There exists a generic subset of $ \mathscr V_j $ of $ \mathscr V $ such that
for $ V \in \mathscr V_j $,  the eigenvalues of  $ RV |_{ L^2_{0,j} }  $
are {\em simple}.
\end{lemm}
\begin{proof} 
We follow the presentation in the proof of \cite[Theorem 4.39]{res} with modifications
needed for our case. We fix $ j $ and consider all operators as acting on
$\mathscr H := L^2_{{0,j}} $. 
The eigenvalue multiplicity is defined using the resolvent:
\[  m_V ( \lambda  ) := \frac{1}{2 \pi i } \tr \oint_{\lambda} ( \zeta - RV )^{-1} 
d\zeta , \]
where the integral is over a sufficiently small positively oriented circle around $ \lambda $. We then define
\begin{equation}
\label{eq:Ethetar}   \mathscr E_r := \{ W \in \mathscr V :  m_W ( \lambda ) \leq 1, \   
\lambda \in \mathbb C \setminus D (0, r )  \}.  \end{equation}
We want to show that for $ r > 0 $,  $ \mathscr E_r $ is open and dense. That will show that the set
\[  \mathscr E := \{ W \in \mathscr V :  \ \forall \, \lambda , \ 
m_W ( \lambda ) \leq 1  \} = \bigcap_{ n \in \NN } \mathscr E_{\frac 1 n } \]
is generic (and in particular, by the Baire category theorem, it has a nowhere dense complement). 

Suppose that $ R W  $ has exactly one eigenvalue $ \lambda_0 $ in 
$  D ( \lambda , r ) $ and $\Spec(RW) \cap D(\lambda,2r) = \{\lambda_0\}$. Putting $ \Omega := D ( \lambda , r ) $  we then define
\begin{equation}
\label{eq:PiWOm} \Pi_W ( \Omega ) := \frac{1}{2 \pi i }  \int_{\partial \Omega } ( \zeta - RW  )^{-1} d\zeta, \ \ \  m_W ( \Omega ) := \tr \Pi_W ( \Omega ) . \end{equation}
If $ V \in \mathscr V $ and $ \| V \|_{\delta } $ is sufficiently small then for $ \zeta \in \partial \Omega $, 
\[  ( R ( W + V ) - \zeta )^{-1} = ( R W   - \zeta )^{-1} ( 
I +  R V  ( R W - \zeta )^{-1} )^{-1}  , \]
exists and we can define $ \Pi_{W + {V } } ( \Omega )  $ as 
in \eqref{eq:PiWOm}.
 This also shows that if $ \|{V } \|_\delta < \varepsilon  $ for sufficiently small $ \varepsilon $
then for $ \zeta \in \partial \Omega $, 
\[ ( R W - \zeta )^{-1} - 
( R ( W + {V } ) - \zeta )^{-1} = \mathcal O_\varepsilon ( \| {V } \|_\delta )_{ \mathscr H \to \mathscr H } 
 . \]
It follows that 
$  \| \Pi_W ( \Omega ) - \Pi_{W + {V } } ( \Omega ) \|_{ \mathscr H \to 
\mathscr H } \leq C_\varepsilon  \| {V } \|_\delta  $. 
In particular, if we take $ \|{V } \|_{\delta } <  1/C_\varepsilon  $, then $ 
\Pi_W ( \Omega ) $ and $ \Pi_{W + {V } } ( \Omega ) $ have the same rank
\begin{equation}
\label{eq:mVconstant} 
\text{ $ m_{ W + V } ( \Omega ) $ is constant for $ \| V \|_\delta $ sufficiently small. } 
\end{equation}
This immediately implies that $ \mathscr E_r $ is 
open: if $ \lambda $ is a simple eigenvalue of $ R W $  then $ m_W ( \Omega ) = 1 $ this 
values does not change under small perturbations.

 Now we want to show that $ \mathscr E_r $ is dense. This follows from the following
statement 
\begin{equation}
\label{eq:WdeltaW} 
\begin{gathered} 
\forall \;  W \in \mathscr V, \ \varepsilon > 0 \ \ \exists\;
 V \in \mathscr V \ \ \ 
 W + V  \in \mathscr E_r , \ \ \ \ \| V\|_{\delta } < \varepsilon. 
 \end{gathered}
\end{equation}
As the number of eigenvalues of $ RW $ outside $ D ( 0 , r ) $ is finite, 
it is enough to prove a local statement 
as it can be applied successively to obtain \eqref{eq:WdeltaW} (once an eigenvalue is simple it stays simple for sufficiently small perturbations). That is, it is enough
to show that
\begin{equation}
\label{eq:WdeltaW1} 
\begin{gathered} 
\forall \;  W \in \mathscr V , \ \varepsilon > 0 \ \ \exists\;
 V \in \mathscr V  \ \forall \, \lambda \in \Omega \\
 m_{W +  V } ( \lambda ) \leq 1, \ \ \ \| V \|_{\delta } < \varepsilon. 
 \end{gathered}
\end{equation}
As in \cite{klop} we proceed by induction and start by noting that one of two 
cases has to occur:
\begin{equation}
\label{eq:option1}
\begin{gathered}
\forall \, \varepsilon > 0 \ \ \exists \, V \in \mathscr V, \ \lambda
\in \Omega \ \ \ 
1 \leq m_{ W + V } ( \lambda ) < m_{ W + V } ( \Omega ) , \ \ \ \| V \|_\delta < \varepsilon , 
\end{gathered}
\end{equation}
or
\begin{equation}
\label{eq:option2}
\begin{gathered}
\exists \, \varepsilon > 0 \ \forall \, V  \in \mathscr V, 
\| V \|_\delta < \varepsilon \ 
\exists \,  \lambda  = \lambda ( V ) 
\in \Omega \ \ \ 
m_{ W + V } ( \lambda ) = m_{ W + V } ( \Omega ) .
\end{gathered}
\end{equation}
The first case implies that adding an arbitrarily small $ V $ to $ W $ produces at least two distinct eigenvalues of $ R ( V + W)  $. The second case implies that for any small perturbation 
preserves maximal multiplicity. 

We will now show that \eqref{eq:option2} {\em cannot occur}. 
For that assume that
$   m_W ( \lambda ) = M $ and that \eqref{eq:option2} holds. For $ V \in \mathscr V $, $ \|V \|_{\delta } < \varepsilon $, put, in the notation of \eqref{eq:PiWOm}, 
\[  k( V ) := \min\{ k : (  R( W + V)  - \lambda ( V ) )^k \Pi_{ W + V } ( \Omega ) = 0 \} .\]
Then $ 1 \leq k ( V ) \leq M $ and $ 
 \mathscr V \ni V \mapsto k ( V ) $ is a {\em lower semi-continuous} function. In fact, if $ \| V_j  - V \|_{\mathscr V } \to 0 $
and 
 then, from \eqref{eq:PiWOm}, we 
see that $(  R (  W + V_j ) - \lambda ( V_j ) )^k \Pi_{ W + V_j } ( \Omega ) = 0$, then 
$ (  R (  W + V ) - \lambda  ( V ) )^k \Pi_{ W +V  } ( \Omega ) = 0$.

We also define
\[  k_0  := \max\{ k ( V ) : V \in \mathscr V , \| V \|_{\delta }
< \varepsilon/2 \}.\]
It follows that if $ k ( V' ) = k_0 $ then $ k ( V + V' ) = k_0 $ for $ \| V \|_{\delta} < \rho $, with a sufficiently small $ \rho $.  Hence we can replace $ W $ by $ W + V' $, decrease $ \varepsilon$  and assume that
\begin{equation}
\label{eq:PVMim}
\begin{gathered} ( R (W + V ) - \lambda( V ) )^{k_0} \Pi_{ V + W } ( \Omega ) = 0 , \\ 
( R (  W + V)  - \lambda ( V ) )^{k_0-1} \Pi_{ V + W } ( \Omega ) \neq 0 , \\ 
m_{ W + V } ( \lambda ( V ) ) = \tr \Pi_{ V + W } = M > 1  , \ \ \ \forall \, V , \ \ 
\|V \|_{\delta}  < \varepsilon.
\end{gathered}
\end{equation}

\noindent
To see that \eqref{eq:PVMim} is impossible we first assume that $ k_0 > 1 $. 
Take $ V = V ( t ) = W + t V $, $  \| V \|_{C^M} < \varepsilon $, $ t \in [ -1 ,1 ]$.
For $ h , g \in \mathscr H  $ we define
(dropping  $ \Omega $ in  $ \Pi_\bullet ( \Omega ) $)
\[ \begin{split}    & w ( t)  :=  ( R ( W + t V )  - \lambda(t) )^{k_0-1} \Pi_{ W +t V } h , \\ & 
\widetilde {w } (t ) := ( (  W^* + t  V^* )R^*  - \overline {\lambda (t)} )^{k_0 - 1} \Pi_{W+tV } ^* g . \end{split} \]
By our assumption \eqref{eq:PVMim} we can choose $ g $ and $ h $ so that 
$ w := w ( 0 ) \not \equiv 0 $ and $ \widetilde {w} := \widetilde {w }( 0 ) \not \equiv 0 $. 
Lemma \ref{l:anal} then implies that 
\begin{equation}
\label{eq:w0tw0}  \supp w = \supp \widetilde {w} = \mathbb C / 3 \Lambda . 
\end{equation}
Since $ \lambda ( t ) $ is assumed to be the only eigenvalue of $ R V ( t ) $ in $ \Omega $ 
and since it has fixed algebraic and geometric multiplicity, 
the functions $ t \mapsto \lambda ( t V ) , \Pi_{ W + t V } , w ( t ) $ depend smoothly on  $ t $. 
Hence, we can differentiate:
\[ \begin{split} 
0 & = \frac{ d }{ d t } (  R ( W + tV ) - \lambda ( t ) )^{k_0} \Pi_{W + t V } h \\
&  =
\sum_{ \ell = 0}^{k_0-1} (  R ( W + t V )  - \lambda ( t ) )^\ell R V (  R ( W + t V )  - \lambda ( t ) )^{k_0-1- \ell} \Pi_{W + t V } h \\
& \ \ \ \ \ + 
( R(  W + tV )  - \lambda  ( t ) ) H(t) 
\end{split}\]
where $ H  ( t ) \in \mathscr H $.  We now put $ t = 0 $ and take
the $ \mathscr H $ inner product with $ \widetilde {w } $: the term with
$ H  (0 ) $ disappears as $ ( R   W  - \lambda( 0 )  )^{k_0} \Pi_{W}^* \equiv  0 $
as do all the terms with $ \ell > 0 $. 
Consequently, we obtain 
\[ \forall \,  V \in \mathscr V \ \  \langle V {w } ,
R^*  \widetilde {w } \rangle = 0 .\]
Since $ V \in L^2_{{0,1} } $, $ w \in L^2_{{0,j}}  $,  $ R^* \widetilde {w } \in L^2_{{0,j+1}}$, we conclude
that  (with $ \circ_j $ denoting components of $ \bullet = w, \widetilde{w } $)
\begin{equation}
\label{eq:UpUm}   \langle U_+  w_2 , R^* \widetilde w_1 \rangle_{L^2 ( F ) } + 
\langle U_- w_1 , R^* \widetilde w_2 \rangle_{L^2 ( F ) } 
= 0 , 
\end{equation} 
where $ F $ is a fundamental domain of the joint group action defined by $\mathscr L$ and $\mathscr C$. 
Since $ V $ is arbitrary on $F$, 
this implies that $ \bar {w}  ( z) ( R^* \widetilde {w } ) ( z ) \equiv 0 $, 
which in turn contradicts \eqref{eq:w0tw0}.

 It remains to consider the case of $ k_0 = 1 $ in \eqref{eq:PVMim}. In that
 case the finite rank projection $ \Pi_W $ can be written as (with the notation,
 $ ( f \otimes g ) ( u ) := f \langle u , g \rangle $)
 \begin{equation}
 \label{eq:PiW}   \Pi_W = \sum_{ j =1}^M w_j \otimes \widetilde {w} _j, \ \ \ 
 \langle {w}_j , \widetilde {w}_k \rangle = \delta_{ jk} , \ \ \ ( RW - \lambda_0 ) {w}_j = 0, \ \ 
 ( W^* R^* - \bar \lambda_0 ) \widetilde {w}_k = 0 . \end{equation}

 Then, 
 \[ \begin{split} 0 & =  \frac{d}{dt} \left[ (  \lambda ( t )  - R ( W + t V ) )  \Pi_{ W + t V } \right] \\
 & =\lambda'(t)  \Pi_{ W + t V } - R V \Pi_{W + t V } + 
( \lambda ( t ) - R ( W + t V ) ) \frac{d}{dt}  \Pi_{ W + t V } \end{split}\] 
Applied to $ w_j $ and paired with $\widetilde {w}_k $ we get at $ t = 0$,
\[ 0 = \lambda'(0 ) \delta_{jk}   -  \langle R V w_j , \widetilde {w}_k \rangle . \] 
Hence we need to show that for $ j \neq k $
\begin{equation}
\label{eq:wjk}  
\langle R V w_j , \widetilde {w}_k \rangle = 0 , \ \ \forall \, V \in \mathscr V \ \Longrightarrow \ 
w_j = \widetilde {w} _k = 0 . \end{equation}
But that is done as in the discussion after \eqref{eq:UpUm}.

We have now proved that \eqref{eq:option1} holds and we use it now to 
prove \eqref{eq:WdeltaW1} by induction on $ m_{ W } ( \lambda_0 ) $ where $ \lambda_0 $ is the 
unique eigenvalues of $ RW  $ in $ D( \lambda_0, 2 r )$, $ \Omega := D ( \lambda_0, r ) $.
 If $ m_W ( \lambda_0 ) = 1 $ there is nothing to prove. Assuming that 
we proved \eqref{eq:WdeltaW1} for $ m_W ( \lambda_0 ) < M $ assume that 
$ m_W ( \lambda_0 ) = M $. From \eqref{eq:option1} we see that we can find $ V $, 
$ \| V_0 \|_{\delta } < \varepsilon/ 2 $ such that $ m_{ W + V_0 } ( \Omega ) = 
m_{ W } ( \Omega ) $ (see \eqref{eq:mVconstant}) and such that
all eigenvalues in 
$ \Omega $, $ \lambda_1, \cdots, \lambda_k$, satisfy $ m_{ W + V_0 } ( \lambda_j ) < M $. 
We now find $ r_j $ such that, 
\[ \begin{gathered}
D ( \lambda_j , 2 r_j ) \subset \Omega, \ \ \  D ( \lambda_j, 2 r_j ) \cap D ( \lambda_k , 2 r_k ) = 
\emptyset , \  \ j \neq k ,\\
 \{  \lambda_j \} = D( \lambda_j , 2 r _j ) \cap \Spec ( R (  
W + V_0 ) ) . \end{gathered} \]
 We put $ \Omega_j := D ( \lambda_j, r_j ) $ and
apply \eqref{eq:WdeltaW1} successively to $ W + V_0 + \cdots V_{j-1} $, 
$ j = 1, \cdots , k $,  in 
$ \Omega_j $ with $ \|V_{j} \|_\delta < \varepsilon/2^{j+1} $. That gives the desired 
$ V = \sum_{j=0}^k V_j $. 
\end{proof}

\section{Generic simplicity}

In this section we complete
the proof of Theorem \ref{t:sim}.

We already showed in Proposition \ref{prop:equiv_spec}\footnote{We stated Proposition \ref{prop:equiv_spec} for a smaller class of potentials than the generalized tunnelling potentials considered here, see \eqref{eq:newU}, but the proof only uses only translational and rotational symmetries which are still satisfied for generalized tunnelling potentials } that $\Spec_{L^2_{0,0}}(RW) = \Spec_{L^2_{0,1}}(RW)$ and know from the previous Lemma that we can ensure simplicity of spectra of $RW$ in each representation $L^2_{0,j}$. We shall now see that we can split spectra of $RW$ in $L^2_{0,0},L^2_{0,1}$ from the one in $L^2_{0,2}.$
\begin{lemm}
\label{l:split}
Suppose that 
\[  \Spec_{ L^2_{{0,j } } } ( R W ) \cap D( \lambda_0 , 2r )  = \{ \lambda_0 \}  , \ \ \  j \in \ZZ_3 ,  \   r > 0 ,\]
and $ \lambda_0 $ is a simple eigenvalue of $ RW|_{L^2_{{0,j } } } $. 
Then, for every $ \varepsilon > 0 $ there exists
$ V \in \mathscr V $, $ \| V \|_{\delta } < \varepsilon$,  such that for some $\lambda_1 \neq \lambda_2$
\begin{equation}
\begin{split}
\label{eq:split} 
&\Spec_{ L^2_{{0,2 } } } ( R ( W + V ) ) \cap D( \lambda_0 , r )  = \{ \lambda_2 \}, \\&\Spec_{ L^2_{{0,j } } } ( R ( W + V ) ) \cap D( \lambda_0 , r )  = \{ \lambda_1 \}, j \in \{0,1\}. \ \  
\end{split}
\end{equation}
\end{lemm}
\begin{proof}
As in \eqref{eq:PiW} we have  $ {w}_k , 
\widetilde {w}_k \in L^2_{{0,j_k} } $, such that $ \langle w_k, \widetilde {w}_k \rangle = 1$, and 
\[ ( 2  \lambda_0  D_{\bar z }   - W ) {w} _k = 0 , \ \ \  ( 2 \bar \lambda_0 
 D_z  - W^* ) R^* \widetilde {w}_k = 0 .\]
Since the eigenvalue $ \lambda_0 $ is assumed to be simple, \eqref{eq:defAsc} gives
\begin{equation}
\label{eq:w2wt}  R^* \widetilde {w}_{p}  = \gamma_{1-p}  \mathscr A 
w_{1-p} = \gamma_{1-p} \begin{pmatrix} \ \bar w_{(1-p)2} \\ - \bar w_{(1-p)1} \end{pmatrix}, \ \
w_p = \begin{pmatrix} w_{p1} \\ w_{p2} \end{pmatrix}  , \ \ 
\gamma_{p} \in \mathbb C^* . \end{equation}
We can split an eigenvalue with eigenvectors $ w_k $,  if we can find
$ V $ such that (see \eqref{eq:UpUm} for the notation)
\[  \begin{gathered} \langle V w_2 , R^*\widetilde  {w}_2 \rangle_{L^2 ( F ) }  \neq 
 \langle V w_0 , R^*\widetilde   {w}_0\rangle_{L^2 ( F ) } , \text{ with }\\
 \langle V w_2, R^* \widetilde  {w}_2 \rangle =   \bar{\gamma_2}  \int_F  \left( U_+( z)  w_{22}^2 ( z)  -   U_-(z) w_{ 21}^2( z )\right) d m ( z) \text{ and }\\
 \langle V w_0 , R^* \widetilde  {w}_0 \rangle = \bar \gamma_1 \langle V w_0 , \mathscr A w_1 \rangle=  \bar{\gamma_1}  \int_F  \left( U_+( z)  w_{02} ( z)w_{12} ( z)   -   U_-(z) w_{ 01}( z )w_{11}( z )\right) d m ( z)
\end{gathered} \]
 where we used \eqref{eq:w2wt} to obtain the last equality. 
 If for all (analytic) $U_\pm  $ the terms were equal it would follow that
$\bar{\gamma_2} w_{2\ell}^2 = \bar{\gamma_1} w_{0\ell}w_{1\ell}$ for $\ell \in \{1,2\}$. This implies that $w_{2\ell}$ vanishes at $0, \pm z_S$. However, the zeros at $\pm z_S$ have to be at least of order $2$ since by rotational and translational symmetry
\begin{equation}
\label{eq:order_of_zeros}
\begin{split} w_2(z\pm z_S) &= \bar \omega w_2(\omega (z \pm z_S)) =  \bar \omega w_2(\omega z \pm z_S \mp (1+\omega)) \\
&=\bar \omega  \operatorname{diag}(\omega^{\pm 1}, \omega^{\mp 1}) w(\omega z \pm z_S). \end{split}
\end{equation}
This means that for instance at $z_S$ we have $w_2(z\pm z_S) = \operatorname{diag}(1,\omega) w(\omega z \pm z_S)$ which means that the first component has to vanish at least to third order and the second component at least to second order. 
This implies that $w_2$ has at least $5$ zeros counting multiplicities and this is impossible by the usual theta function argument \cite[Lemma $4.1$]{bhz2}.
\end{proof}
We can now finish
\begin{proof}[Proof of Theorem \ref{t:sim}]
Lemma \ref{l:sim} (strictly speaking its proof) 
and Lemma \ref{l:split} now show
that for every $ r > 0 $, the set
\[  \mathscr V_r := \{ V :  \text{ $ R V \vert_{ L^2_{0,1}\oplus L^2_{0,2} } $  has simple eigenvalues in $ \mathbb C \setminus 
D ( 0 , r ) $} \} \] 
is open and dense. We then obtain $ \mathscr V_0 $ by taking the intersection of 
$ \mathscr V_{1/n } $.
\end{proof}

\section{The Chern number of a two-fold degenerate flat band} 
\label{sec:Chern}
In this section we compute the Chern number of the flat band in the case of 2-fold degeneracy.
We start by a general discussion of of the Chern connection 
and the Berry connection in the case holomorphic vector bundles. Although
we stress our case of the two torus, \S\S \ref{s:cc} and \ref{s:Berry} apply to vector bundles over
more general manifolds.

\subsection{The Chern connection} 
\label{s:cc}
Suppose that $  \pi : E \mapsto  X $ 
is a holomorphic vector bundle over a torus $ X =  \mathbb C/\Lambda^*  $ (see \cite[\S 2.7]{notes} for a quick introduction 
sufficient for our purposes or  \cite{Wells} for an in-depth treatment), and that $ E  $ is a sub-bundle of a trivial Hilbert bundle over $ X $, $ X \times \mathscr H $,  where $ \mathscr H $ is
a Hilbert space. This gives a hermitian structure on $ E $: for $ k \in X  $, 
we introduce an inner product on the fibers $ E_k := \pi^{-1} ( k ) $, using $ E_k \subset \mathscr H $:
\[   \langle \zeta , \zeta' \rangle_k := \langle \zeta, \zeta' \rangle_{\mathscr H } , \ \ \ \zeta, \zeta' \in E_k . \]
We then have two natural connections on $ E $, the {\em Chern connection}, available when 
the bundle is holomorphic and equipped with hermitian structure, and a hermitian connection\footnote{$ D :
C^\infty ( X , E ) \to C^\infty ( X, E \otimes T^*X ) $ is a connection if for any $ f \in C^\infty ( X ) $, 
$ D ( f s ) = f D s + s df $. A connection $ D $ is hermitian if  $ d \langle s ( k ) , s' ( k ) \rangle_k  = \langle D s ( k ), s' ( k ) \rangle_k + 
\langle s( k ) , Ds'( k ) \rangle $.},  available
for any smooth vector bundle embedded in a Hilbert bundle. In the context of vector bundles of 
eigenfunctions, the latter is called the {\em Berry connection} and we adopt this terminology for the
general case as well.

We first define the Chern connection. For that we choose a local holomorphic trivialization $ U \subset X $, $ \pi^{-1} ( U ) \simeq U \times \mathbb C^n $, for which the hermitian metric is given by 
\begin{equation}
\label{eq:defGz} \langle \zeta , \zeta\rangle_k = \langle G ( k ) \zeta , \zeta \rangle = \sum_{i,j=1}^n G_{ij}(k) \zeta_i \bar \zeta_j  \ \ \zeta \in \mathbb C^n , \ \ 
k \in U . \end{equation}
We note that if $ \{ u_1 ( k ) , \ldots , u_n ( k ) \} \subset \mathscr H $ is a basis of $ E_k $ for $ k \in U $, 
and $ U \ni k \to u_j ( k ) $ are holomorphic, then $ G ( k ) $ is the Gramian matrix:
\begin{equation}
\label{eq:gram} G ( k ) := \left( \langle u_i ( k ) , u_j ( k ) \rangle_{\mathscr H} \right)_{ 1\leq i, j \leq n } . 
\end{equation}
If $ s : X \to E $ is a section, then the Chern connection $
D_C : C^\infty (X;  E ) \to C^\infty (X ; E \otimes T^* X) $, over $ U $ is given by (using only 
the local trivialization and \eqref{eq:defGz})
\begin{equation}
\label{eq:Chern}
\begin{gathered}    D_C s ( k ) := d s ( k ) +\eta_C ( k )  s ( k )  , \\ \eta_C( k )  :=  G(k)^{-1} \partial_k G ( k ) \, dk \in
C^\infty ( U ,\Hom( \mathbb C^n , \mathbb C^n ) \otimes (T^* U)^{1,0} ) .
\end{gathered}
 \end{equation}
Here  $ \partial_k $ denotes the holomorphic derivative and the notation $ (T^* U)^{1,0} $ indicates that only $ dk $ and not $ d \bar k $ appear in the 
matrix valued 1-form $ \eta_C $,  $ \eta_C = \eta_C^{1,0} $. We also recall that  $ D_C $ is the unique hermitian connection with this property -- see \cite[Theorem 2.1]{Wells}.  

For the definition of the Berry connection we only require that $ E \to X $ is a smooth
vector bundle which is a subbundle of $ X \times \mathscr H$, where $ \mathscr H $ is a Hilbert space. 
That means for $ k \in X $ we have a well defined orthogonal projection $ \Pi ( k ) : \mathscr H 
\to E_k := \pi^{-1} ( k ) $ and an inclusion map $ \iota : E \hookrightarrow X \times \mathscr H$. The formula
for the Berry connection is then given by 
\begin{equation}
\label{eq:Berry} D_B s ( k ) := \Pi ( k ) ( d ( \iota \circ s ) ( k ) ). 
\end{equation}
To find a local expression similar to \eqref{eq:Chern} we use the Gramian \eqref{eq:gram}. If
$ s ( k ) = \sum_{ j=1}^n s_j^U  ( k ) u_j ( k )  =: A (k ) s^U ( k ) $,  $ A( k ) : \mathbb C^n  \to \mathscr H $
(so that $ A( k ) $ provides a 
local trivialization) then $ \Pi ( k ) = A ( k ) G ( k ) ^{-1} A(k)^* $ and 
\begin{equation}
\label{eq:Berry1} 
\begin{gathered} 
\begin{split} D_B s ( k ) & = \Pi ( k ) \sum_{ j=1} ^n  \left( d s_j^U ( k ) u_j ( k ) + s_j^U ( k ) d u_j ( k )  \right) 
\\ &  = 
A ( k ) ( d s^U ( k )  + \eta_B ( k)  s^U ( k ) ) , \end{split}  \\
\eta_B (k) =  G ( k )^{-1} B ( k ) \in C^\infty ( U , \Hom( \mathbb C^n, \mathbb C^n )  \otimes T^* U ) , \\ 
B ( k )_{  \ell j  } := \langle d u_j( k ) , u_\ell ( k )  \rangle_{\mathscr H }  \in C^\infty ( U, T^* U ) . 
\end{gathered}
\end{equation}
These formulas  hold for choices of $ u_j $ which are not necessarily holomorphic. However if, as in 
\eqref{eq:gram},  $ k \mapsto u_j ( k)  $ are holomorphic, then
\begin{equation}
\label{eq:B2G} \begin{split} (\partial_k G ( k ))_{ij} dk  & =   \langle \partial_k u_i ( k ) , u_j ( k ) \rangle_{\mathscr H} dk  + 
\langle u_i ( k ) ,  \partial_{\bar k } u_j ( k ) \rangle_{\mathscr H} dk \\
& =   \langle \partial_k u_i ( k ) , u_j ( k ) \rangle_{\mathscr H} dk  \\
& = \langle d u_i ( k ) , u_j ( k ) \rangle  = B ( k )_{ij} ,  \end{split}  \end{equation}
since $ \partial_{\bar k } u_j ( k )   = 0 $ and $ d w = \partial_k w \, dk + \partial_{\bar k }w  \, d \bar k $. In particular, that means that in the notation of \eqref{eq:Chern} and \eqref{eq:Berry} 
\begin{equation}
\label{eq:C2B} \begin{split}  \text{$ U \ni k \mapsto u_\ell ( k )$ holomorphic } & \ \Longrightarrow \ \  \eta_C ( k ) = \eta_B ( k ) , \ k \in U 
 \\ &  \ \Longrightarrow \ \    D_C = D_B ,
\end{split}
\end{equation}

We record this standard fact as
\begin{prop}
Suppose that $ X $ is a complex manifold and $ E \mapsto X $ is a holomorphic vector
bundle with a holomorphic embedding $ \iota : E \to X \times  \mathscr H $ into a trivial 
Hilbert bundle. Then the Berry connection \eqref{eq:Berry} and the Chern connection 
\eqref{eq:Chern} defined using the hermitian structure on $ \mathscr H $
are equal.
\end{prop}

\noindent
{\bf Remark.} As was pointed out to us by Michael Singer, the conclusion \eqref{eq:C2B} 
could be deduced directly from the uniqueness of the Chern connection mentioned after 
\eqref{eq:Chern}: using \eqref{eq:Berry} we have 
$  D_B^{(0,1)} s( k ) = \Pi ( k ) (  d^{(0,1)} ( \iota \circ s ) ( k )  ) $. But as the embedding $ \iota $
(an inclusion, in our case) is holomorphic this implies that $ D_B^{(0,1)} s ( k ) = 0 $ for holomorphic sections.
This and being hermitian characterize the Chern connection.
We should also stress that the discussion above does not depend on the 
fact that $ X $ has complex dimension one.

The curvature of a connection $ D $ is given by 
\begin{equation}
\label{eq:curv}   \Theta := D \circ D  , 
\end{equation}  
which is a globally defined
two form with values in $\Hom ( E, E ) $. In a local trivialization in which $  D = d + \eta $, we have
$ \Theta = d \eta + \eta \wedge \eta $. For the Chern connection, for $ X $ of any dimension
$ \Theta = \bar \partial \partial \eta_C $ since \eqref{eq:Chern} shows that 
 $ \partial \eta_C = - \eta_C \wedge \eta_C $ (when $ X $ has a complex dimension one, 
 this is obvious as $ dk \wedge dk = 0 $). 
It is then immediate from \eqref{eq:C2B} that
\begin{equation}
\label{eq:B2C}
\Theta := D_C \circ D_C = D_B \circ D_B , 
\end{equation}
that is,  in the holomorphic case,  the curvatures defined using the Chern curvature or the Berry curvature agree for
holomorphic vector bundles embedded in trivial Hilbert bundles. 

The Chern class (a Chern number in the case of $ \mathbb C/\Lambda^* $) is given by 
\[     c_1 ( E ) := \frac{i}{ 2 \pi } \int_{\mathbb C/\Lambda^* }   \tr \Theta \in \mathbb Z , \]
where we note that over $ U \subset \mathbb C/\Lambda^* $ for which we defined 
\eqref{eq:gram}, 
\begin{equation}
\label{eq:g2C} 
\begin{split} 
\tr \Theta &  = \partial_{\bar k } \tr G ( k )^{-1} \partial_k G ( k ) \, d \bar k \wedge d k \\
& = 
\partial_{\bar k } \partial_k \log g ( k ) \, d \bar k \wedge d k  , \ \ \ \  g ( k ) := \det G ( k ) ,
\end{split} 
\end{equation}
where we used Jacobi's formula \cite[(B.5.14)]{res}. 
In particular,
\[ H ( k ) := \partial_{\bar k } \partial_k \log g ( k ) = g(k)^{-2}( g(k) \partial_{\bar k } \partial_k g(k) - \vert \partial_k g(k)\vert^2).\]

For any holomorphic hermitian vector bundle the trace of the curvature of the Chern connection,
$ \tr \Theta $ can be interpreted as a curvature of a line bundle. If $ \pi  : E \to X $ has rank 
$ n$, we obtain a line bundle $ \pi : L :=  \wedge^n E \to X $. It inherits hermitian 
structure from $ E $. If we define the Chern connection on $ \wedge^n E $ as in \eqref{eq:Chern} 
(using only holomorphy and the hermitian structure) we obtain a new curvature $ \Theta_L $ which 
is a differential two form on $ X $, and 
\[   \Theta_L = \tr \Theta . \]

In case when $ E $ embeds holomorphically in $ X \times \mathscr H $ we can then take, as in \eqref{eq:gram}, $ k \mapsto u_j ( k ) \in \mathscr H $, $ j = 1, \cdots,  n $,
a local holomorphic basis of $ E $. Then for 
\begin{equation}
\label{eq:defPhi}  \Phi ( k ) := \wedge_{ j=1}^n u_j ( k ) \in  \wedge^n E_k \subset \wedge^n \mathscr H , \end{equation}
we have
\[  \| \Phi ( k ) \|^2_{\wedge^n \mathscr H} = 
\det \left( (\langle u_j (k ), u_\ell ( k ) \rangle_{\mathscr H}  ) _{ 1\leq j,\ell \leq n }\right) = \det G ( k ) = g ( k )  . \]
In particular when $ X = \mathbb C /\Lambda^* $, we obtain, as in \cite[(5.10)]{bhz2}, 
{ \begin{equation}
\label{eq:defH}  \Theta_L = H( k )  \ d \bar k \wedge d k , 
\end{equation} }
where $ H $ is given by 
\begin{equation}
\label{eq:H2P}   H ( k ) = \| \Phi ( k ) \|^{-4} \left( \| \Phi ( k ) \|^2 
 \| \partial_k  \Phi ( k ) \|^2 - | \langle \partial_k \Phi (k) , 
 \Phi (k) \rangle |^2 \right) \geq 0 ,   
 \end{equation}
where $   \| \bullet \| = \| \bullet \|_{ \wedge^n \mathscr H}$. 

\medskip

\noindent
{\bf Remark.} 
From a physics perspective the  construction of the line bundle $ \wedge^n E $, 
in the case of $ E \subset X \times \mathscr H $ can be interpreted as 
the Slater determinant of the individual Bloch functions on the fermionic $n$-particle Hilbert space. We thus find that the trace of the curvature of the rank $n$ vector bundle coincides with the curvature of the line bundle described by the $n$-particle wavefunction.

\subsection{The Berry curvature}
\label{s:Berry}
For completeness we derive the standard formula for the curvature of the Berry connection 
\eqref{eq:Berry}:

\begin{prop}
\label{p:BC}
Suppose that $ \pi : E \to X $ is a complex vector bundle over a manifold $ X $ and that
there exists an embedding $ \iota : E \to X \times \mathscr H $ into a trivial Hilbert bundle.
Then the curvature of the connection \eqref{eq:Berry} is given  
in terms of the orthogonal projection $ \Pi ( k ) : \mathscr H \to E_k := \pi^{-1} ( k ) $, as 
\begin{equation}
\label{eq:Pi2Th}  \Theta  = \Pi \, d \Pi \wedge d \Pi |_E  , \end{equation}
and is a differential two form with values in $ \Hom ( E , E ) $.
\end{prop}
\begin{proof}
This is a local computation so for some $ U \subset X $ we can choose a
smoothly varying orthonormal basis $ \{ u_j ( k ) \}_{ j =1}^n $, $ k \in U $.  Then in the notation 
of \eqref{eq:Berry1}  (we drop the dependence on $ k $ in $ A ( k ) $, $ \Pi ( k ) $ and $ E_k $)
\begin{equation} 
\label{eq:AA}    A : \mathbb C^n \to \mathscr H, \ \  A^* : \mathscr H \to \mathbb C^n , \ \ \ 
A A^* = \Pi, \ \   A^* A = I_{\mathbb C^n } . 
\end{equation}
With the trivialization given by $ A $, we have (using \eqref{eq:AA})
\[  D s = A^*  ( \Pi ( d ( A s ) ) = A^* \Pi A ds + A^* \Pi dA s = ds + A^* dA s =: ds + \eta ds   . \]
Hence, in this trivialization, the curvature 
is a differential two form with values in $ \Hom ( \mathbb C^n, \mathbb C^n ) $:
\[ \begin{split} A^* \Theta A & =   d \eta + \eta \wedge \eta = d ( A^* dA ) + A^* dA \wedge A^* d A\\
&  = dA^* \wedge dA + A^* dA \wedge  A^*dA  . \end{split} \]
The curvature $ \Theta = D_B \circ D_B $ which is a differential form with values in $ \Hom ( E,  E ) $, 
is then given by 
\begin{equation}
\label{eq:Th2A} \begin{split} \Theta & =  \Pi \Theta \Pi  = A ( dA^* \wedge dA  + A^* dA \wedge A^* dA  ) A^* \\
& = 
A dA^* \wedge dA A^* + A A^* d A \wedge A^* dA A^* \\
& = A dA^* \wedge dA A^* + A d A^*A \wedge dA^* A A^*  , \end{split}  \end{equation}
where we used $ d ( A^* A ) = 0 $.

The right hand side in \eqref{eq:Pi2Th} is given by 
\[  \begin{split} AA^*  d ( A A^* ) \wedge d ( A A^* ) & = 
AA^* ( ( dA A^* + A dA^* ) \wedge ( d A A^* + A d A^* ) 
\\
& = A A^* \left( d A \wedge ( A^* dA ) A^* + d A \wedge ( A^*A ) dA^* \right. \\
& \ \ \ \ \ \left. + A dA^* \wedge d A A^* + 
A  dA^* \wedge d A A^* \right) .
\end{split}  \]
From \eqref{eq:AA} we see that $ A^* A = I_{\mathbb C ^n } $ and
that $ A^* dA = - dA^* A $. Hence, 
\[ \begin{split} \Pi d \Pi \wedge d \Pi  & = 
AA^*  \left( - d A \wedge d A^* A A^* + d A \wedge dA^* \right.  \\
& \ \ \ \ \ +  \left.   A dA^* \wedge d A A^* + 
A  dA^* \wedge A d   A^* \right)  .
\end{split} \]
Acting on $ E $, $ AA^* = I_E $ and hence the first two terms in the bracket cancel:
\[   \Pi d \Pi \wedge d \Pi  |_E = A d A^* \wedge dA A^* |_E + A dA^* \wedge  A d A^* |_ E . \]
But from \eqref{eq:Th2A} that is the same as the action of $ \Theta $ on $ E $. 

\end{proof}

\subsection{Proof of Theorem \ref{t:Chern}}
\label{s:flat2}
We now consider  
\begin{equation}
\label{eq:defVk}  V ( k ) := \ker_{ H^1_0 }  ( D ( \alpha ) + k ) \subset L^2_0 .\end{equation}
This defines a (trivial) vector bundle $ \widetilde E \to \mathbb C $:
\[    \widetilde E := \{ ( k , v ) : v \in V ( k ) \} \subset \mathbb C \times L^2_0 ( \mathbb C/\Lambda; \mathbb C^2 )  . \]
To define a vector bundle over the torus $ \mathbb C /\Lambda^* $ we define an equivalence relation
on $ \mathbb C \times L^2_0 ( \mathbb C/ \Lambda ; \mathbb C^2 ) $:
\begin{equation}
\label{eq:deftau} 
\begin{gathered} \exists \, p \in\Lambda^* \   ( k, u ) \sim_\tau ( k + p , \tau ( p )^{-1} u ) , \ \  [ \tau ( p ) u ] ( z ) = e^{ i \langle z, k \rangle } v ( z ) , 
\end{gathered}
\end{equation}
and notice that $ \tau ( p )^{-1}  V ( k ) = V ( k + p ) $. Using this (see \cite[Lemma 8.4]{notes} or
\cite[Lemma 5.1]{bhz2}), 
\begin{equation}
\label{eq:defE}     E := \widetilde E\, /\sim_\tau  \, \to \mathbb C /\Lambda^* . 
\end{equation}
is a holomorphic vector bundle over $ \mathbb C/\Lambda^* $. 

Since $ \Pi ( k + p ) = \tau(p)^{-1} \Pi ( k ) \tau ( p ) $, the Berry connection defined by 
\eqref{eq:Berry} on $ \widetilde E$, satisfies
\[   ( D_B s ) ( k + p ) = \Pi ( k + p ) ( d ( \iota \circ s ) ( k +  p )) = 
\tau^{-1} ( p ) \Pi ( k ) d ( \iota \circ \tau ( p ) s ( k + p ) )  .\]
Hence, if for $ k \in U \subset X $,  $( k , s ( k )  ) \sim_\tau ( k' , s' ( k' ) ) $
then $ k' = k + p $, $ s' ( k  + p ) = \tau ( p)^{-1} s ( k ) $, for some $ p \in \Lambda^* $ and 
\[  ( k , D_B s ( k ) ) \sim_\tau ( k', D_B s' ( k' ) ) . \]
This means that $ D_B $ is a well defined connection on $ \widetilde E $. Since the Chern 
connection is intrinsically defined on $ \widetilde E $ using holomorphic and hermitian structures, 
the two connections are equal.

If $ m ( \alpha ) = m $ then by Theorem \ref{p:zeros} $ u \in \ker_{H^1_0} D ( \alpha ) $ has exactly $ m $ zeros and let us first
assume that they are simple (this can always be arranged by multiplication by a meromorphic function). Let us denote them by 
$ z_1, \cdots , z_m $. Then 
\begin{equation}
\label{eq:vectorspace}
V ( k ) =\left\{ \sum_{\ell = 0}^m \zeta_\ell F_k ( z - z_\ell ) u_0 (z ),  \zeta \in \mathbb C^m  \right\}, \ \ k \notin \Lambda^* ,  \end{equation}
where $ F_k $'s were defined in \eqref{eq:defFk}.
When $ p \in \Lambda^* $ 
we have 
\begin{equation}
\label{eq:Vp}  V ( p ) =  \{ e^{ i \langle p, z \rangle } u_0 ( z ) f ( z ) :  f \in L( D ) \}, \end{equation}
where $ D $ is the divisor defined by the zeros of $ u_0 $.
We can write elements of $ L (D )$ as follows (see \cite[\S I.6]{tata}): 
\begin{equation}
\label{eq:deff}  f ( z) = \mu_0 + \sum_{\ell =1}^m \mu_\ell  \frac{\theta' ( z - z_\ell )}{\theta ( z  - z_\ell) } , \ \ \ \ \sum_{\ell=1}^m \mu_\ell = 0 . \end{equation} 
We now consider 
\[ G_k ( z ) := e^{ - i \langle k , z  \rangle }  F_k ( z ) , \ \ \ \ G_p ( z ) = e_p ( 0)^{-1} , \ p \in \Lambda^*.  \]
We recall from \cite[Lemma 3.3]{bhz2} that for $ p \in \Lambda^* $, 
\[ \begin{gathered}  F_{ k + p } ( z ) = e_{p} ( k)^{-1} \tau ( p ) F_k ( z ) ,  \ \ \ 
\tau( p ) v ( z ) := e^{ i \langle p, z \rangle } v ( z ) \\ e_p ( k ) := 
\frac{\theta ( z ( k ) ) }{ \theta (z ( k + p )) } = ( -1 )^n ( -1)^m e^{ i \pi n^2 \omega + 2 \pi i n  z ( k ) } , \ \ \ 
z ( p ) = m + n \omega . 
 \end{gathered}
 \]
 Hence,  for $ p \in \Lambda^* $, 
 \[ \begin{split}  \theta ( z ( k + p ) )^{-1}  G_{k+p } ( z  ) & = \frac{ \theta ( z ( k  ) } {\theta ( z ( k + p ) ) } 
 e^{ - i \langle ( k + p ) , z  \rangle } \theta ( z ( k ) )^{-1} F_{k+p} ( z ) \\
 & = 
 e_p ( k )  e^{ - i \langle ( k + p ) , z }  e_p ( k)^{-1} e^{ i \langle p , z - z_\ell \rangle} F_k ( z )
  \\ &  = 
 e^{ - i \langle k , z - z_\ell \rangle }  \theta ( z ( k ) )^{-1} F_k ( z ) \\
 & =  \theta ( z ( k ) )^{-1} G_k ( z  ) , \end{split} \] 
 that is, 
 $ k \mapsto \theta ( z ( k ) )^{-1} G_k ( z - z_\ell ) $ is periodic with respect to $ \Lambda^* $. 
 If $ \lambda = ( \lambda_1 , \cdots, \lambda_m ) \in\mathbb C^m $, $ \sum_{ j=1}^{m} \lambda_j = 0 $, then 
 \begin{equation}
 \label{eq:defGkl}  G_k ( z , \lambda ) := \sum_{ j =1}^m \lambda_j \theta( z ( k ) )^{-1} ( G_k ( z - z_j ) - 
 e_k ( 0 )^{-1} ) , 
 \end{equation}
 is also periodic and smooth in $ k$ at $ \Lambda^* $: for $ p \in \Lambda^* $ we have
 \[ \begin{split}  G_p ( z, \lambda ) & =  G_0 ( z, \lambda ) = \sum_{j=1}^m \lambda_j   (z' ( 0 )\theta' ( 0 ))^{-1}  
 \partial_k \left( e^{ - i \langle z, k \rangle } F_k ( z - z_j ) \right) |_{k=0 } \\ & = 
\sum_{j=1}^m \lambda_j  (z' ( 0 )\theta' ( 0 ))^{-1} \left( \tfrac i 2 ( z - z_j ) -  z' ( 0 ) \frac{ \theta' ( z - z_j ) }{ \theta ( z- z_j ) } 
\right)
\\ & = \mu_0 + \sum_{ j = 1}^m \mu_j \frac{ \theta' ( z - z_j ) }{ \theta ( z- z_j ) } , \ \ \ \ \  \sum_{ j=1}^m \mu_j = 0. 
 \end{split} \]
Hence in view of \eqref{eq:Vp} and \eqref{eq:deff} we can extend 
\eqref{eq:vectorspace} to all $ k \in \mathbb C $ as follows:
\[ V ( k ) = \left\{ e^{ i \langle z, k \rangle}  G_k ( z , \lambda )  u_0 (z ) +
\lambda_0 F_k ( z - z_1 ) u_0 ( z ) : \lambda_0 \in \mathbb C,  \lambda = ( \lambda_1, \cdots, \lambda_m ) \in \mathbb C^m , \ \sum_{\ell = 1}^m \lambda_\ell = 0 
 \right\} \]
We now introduce
\[  W ( k ) :=  \left\{ e^{ i \langle z, k \rangle}  G_k ( z , \lambda )  u_0 (z ) 
:  \lambda = ( \lambda_1, \cdots, \lambda_m ) , \ \sum_{\ell = 1}^m \lambda_\ell = 0 
 \right\}  , \ \ \  k \in \mathbb C .\]
As in \eqref{eq:defE} this family of subspaces of $ L^2_0 $ defines a rank $ m-1$ vector bundle, $ F \to \mathbb C/\Lambda^* $. 
 Since $ k \mapsto G_k ( z, \lambda ) $ is periodic, $ F $ is trivial. If  $ E_1 \to \mathbb C/\Lambda^*  $ is the line bundle coming from 
 the family of subspaces of $ L^2_0 $, 
\[ V_1 ( k ) := \mathbb C F_k ( z - z_1 ) u_0 ( z ) , \ \ k \in \mathbb C ,\] 
(again, in the sense of \eqref{eq:defE}) we see as in \cite{led} and \cite[(5.9),(B.8)]{bhz2} that $ c_1 ( E_1 ) = -1 $. Since
$ E = F \oplus E_1 $, we obtain $ c_1 ( E ) = - 1 $.

Finally, we observe that for $ \Omega : L^2_0  ( \mathbb C/\Lambda ; \mathbb C )
\to L^2_0  ( \mathbb C/\Lambda ; \mathbb C )  $, 
$ \Omega u ( z ) := u ( \omega z)  $,  $ \ker_{H^1_0} ( D ( \alpha ) + \bar \omega k ) = \Omega 
 \ker_{H^1_0} ( D ( \alpha ) +  k )$ (see \cite[\S 2.1]{bhz2}). Hence, in the notation of
 \eqref{eq:Berry}. 
 $ \Omega \Pi ( k ) \Omega^* = \Pi ( \bar  \omega k ) $. If $ R k := 
 \bar \omega k $, this means that $ R ^* \Pi = \Omega \Pi \Omega^* $. 
 Also the pull back of $ \Theta $ by $ R $ is well defined 
 and, using \eqref{eq:Pi2Th} we see that 
 \[    R^* \Theta = R^* (  \Pi d \Pi \wedge d \Pi ) = 
  R^*  \Pi d (  R^* \Pi )  \wedge d (  R^* ) \Pi 
 =  \Omega ( \Pi d \Pi \wedge d \Pi ) \Omega^* = \Omega \Theta \Omega^* . \]
 In particular, in the notation of  \eqref{eq:defH}, we have
\[   \tr R^* \Theta = \tr \Theta  \ \Longrightarrow \ 
H (\bar  \omega k )  = H ( k ) . \]
Strictly speaking we should, just as we did at the end of \eqref{eq:defE}, justify passing to the quotient.
That is again easy by noting that $ \Omega \tau ( p ) \Omega^* = \tau ( \bar \omega p ) $. This completes the proof of Theorem \ref{t:Chern}.

\begin{table}[h!] 
\centering
 \begin{tabular}{||c | c | c ||} 
 \hline
$X=L^2_{0,2}$ & $X=L^2_{0,0}$  & $X=L^2_{0,1}$ \\ [0.5ex] 
 \hline\hline
 {\color{blue}1.2400 -- 0.0000i}   &  1.6002 + 0.0000i  & 1.6002 + 0.0000i \\ 
 {\color{blue}1.2400 -- 0.0000i} &  1.2583 -- 1.1836i   & 1.2583 -- 1.1836i \\ 
1.3424 + 1.6788i  &  1.2583 + 1.1836i  & 1.2583 + 1.1836i \\ 
1.3424 -- 1.6788i   &  1.4019 -- 2.2763i   & 1.4019 -- 2.2763i \\ 
 2.9543 + 0.0000i &  1.4019 + 2.2763i  & 1.4019 + 2.2763i \\ 
 1.4575 + 2.7610i &  1.5001 + 3.3130i  & 1.5001 + 3.3130i \\
 1.4575 -- 2.7610i  &  1.5001 -- 3.3130i   & 1.5001 -- 3.3130i \\
  3.5878 + 1.9298i&  3.4078 + 1.3122i  &  3.4078 + 1.3122i \\
  3.5878 -- 1.9298i &  3.4078 -- 1.3122i   &  3.4078 -- 1.3122i\\ [1ex] 
 \hline
 \end{tabular}
 \medskip
 \caption{Magic angles for $\theta =2.808850897$ and $U_{\pm}:=U_0(\pm \bullet)$ with $U_0(\zeta)= \cos(\theta) U_1(\zeta )+ \sin(\theta) \sum_{i=0}^2 \omega^i e^{-(\zeta \bar \omega^i - \bar \zeta \omega^i)}$ such that $1/\alpha \in \Spec_{X}(T_0)$ (counting algebraic multiplicity). The magic angle with \blue{algebraic multiplicity $2$} and \blue{geometric multiplicity $1$} is highlighted in blue. }
 \label{tab:double}
\end{table}

\section{Numerical observations}
\label{sec:numerics}
Here we present two numerical observations related to our mathematical results. For all our numerics, we used a Fourier discretization of the operators, see \cite[Sec.5.1]{suppl} for an explanation, with $N=101$ Fourier coefficients per spatial dimension.

\subsection{Algebraic multiplicities in the spectral characterization}

Theorem \ref{theo:rigidity} implies that it is impossible to have $$\dim \ker_{L_0^2}(D(\alpha))=\dim \ker_{L_{0,2}^2}(D(\alpha))=2$$ which is equivalent to having eigenvalues of geometric multiplicity $2$ for $T_0$, i.e.
$$\dim \ker_{L^2_0}(T_0-1/\alpha) =  \dim \ker_{L^2_{0,2}}(T_0-1/\alpha)=2,$$ 
we can indeed have that $1/\alpha$ is an eigenvalue of algebraic multiplicity $2$ and geometric multiplicity $1$ on $L^2_0$ and $L^2_{0,2}.$ This is illustrated in Table \ref{tab:double} and Figure \ref{fig:double}. In particular, it implies that $T_k$ in general is not diagonalizable. Since the algebraic multiplicity of $T_k$ is independent of $k$, it follows by Theorem \ref{p:zeros} that the geometric multiplicity is independent of $k$. Examples of this are exhibited in Table  \ref{tab:double} and Figures \ref{fig:double} and \ref{fig:double2}.

\begin{figure}
\includegraphics[width=10.5cm,height=6.5cm]{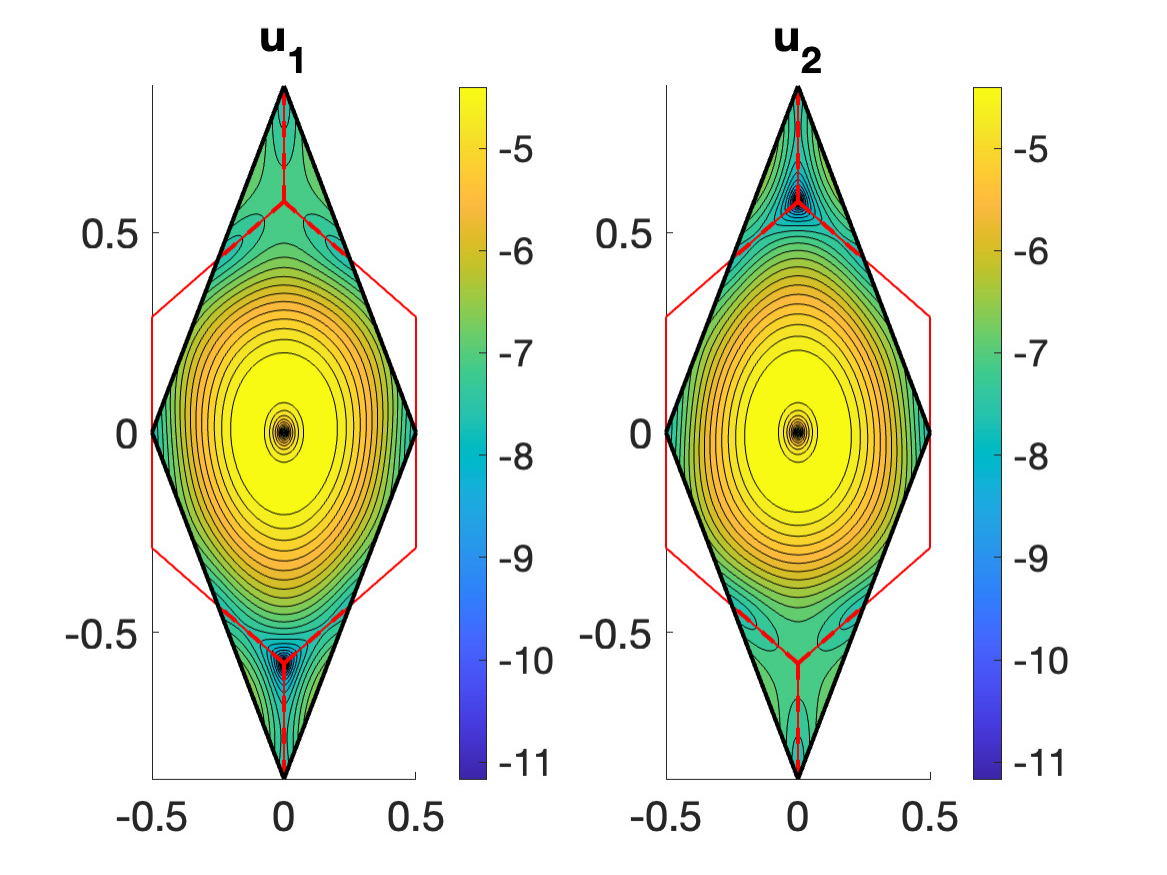}\\
\includegraphics[width=10.5cm,height=6.5cm]{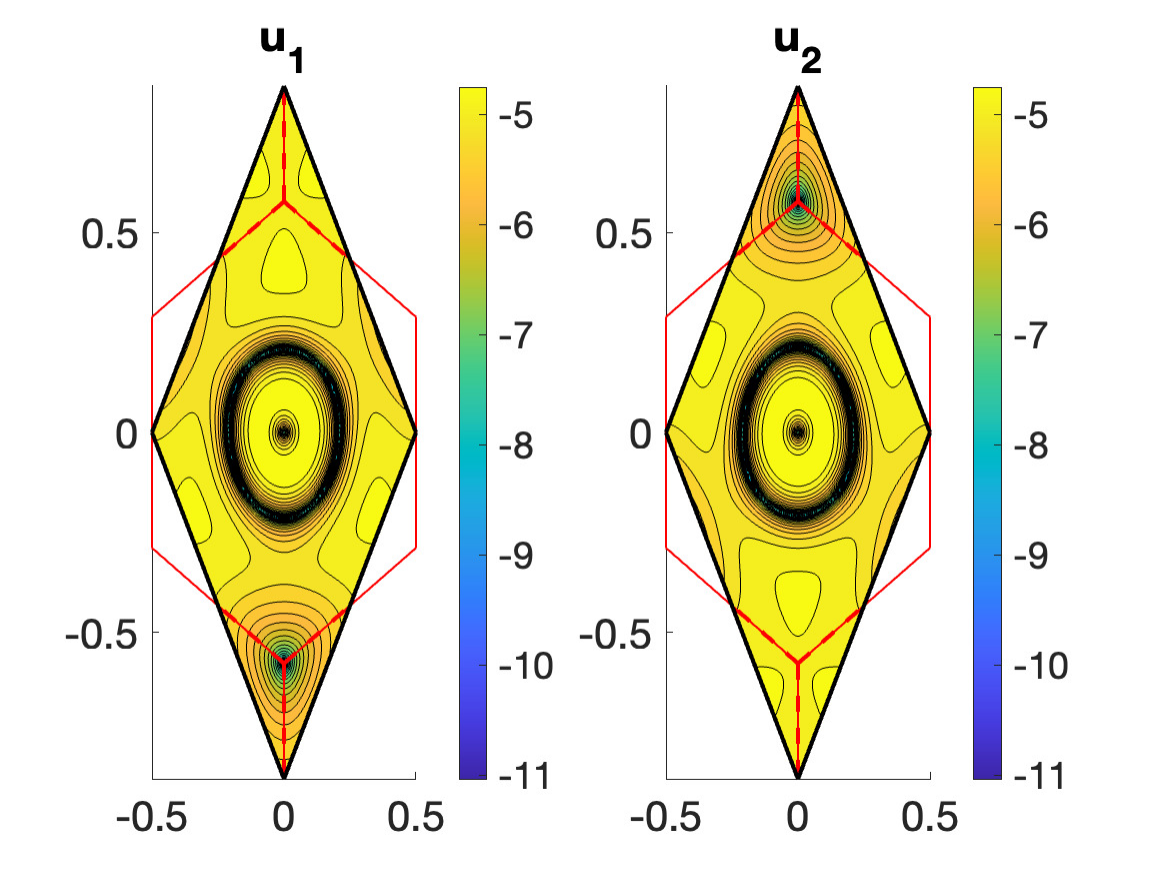}
\caption{The first two singular values of $D(\alpha)$ are $2.804e-15$ and $3.990$ suggesting the existence of only one flat band at $\alpha=1.2400$ for $\theta =2.808850$ and $U_{\pm}:=U_0(\pm \bullet)$ with $U_0(\zeta)= \cos(\theta) U_1(\zeta )+ \sin(\theta) \sum_{i=0}^2 \omega^i e^{-(\zeta \bar \omega^i - \bar \zeta \omega^i)}$. The eigenvector of $T_0$ with eigenvalue $1/\alpha$ is shown on top and the generalized one at the bottom.}
\label{fig:double}
\end{figure}

\begin{figure}
\includegraphics[width=10.5cm,height=6.5cm]{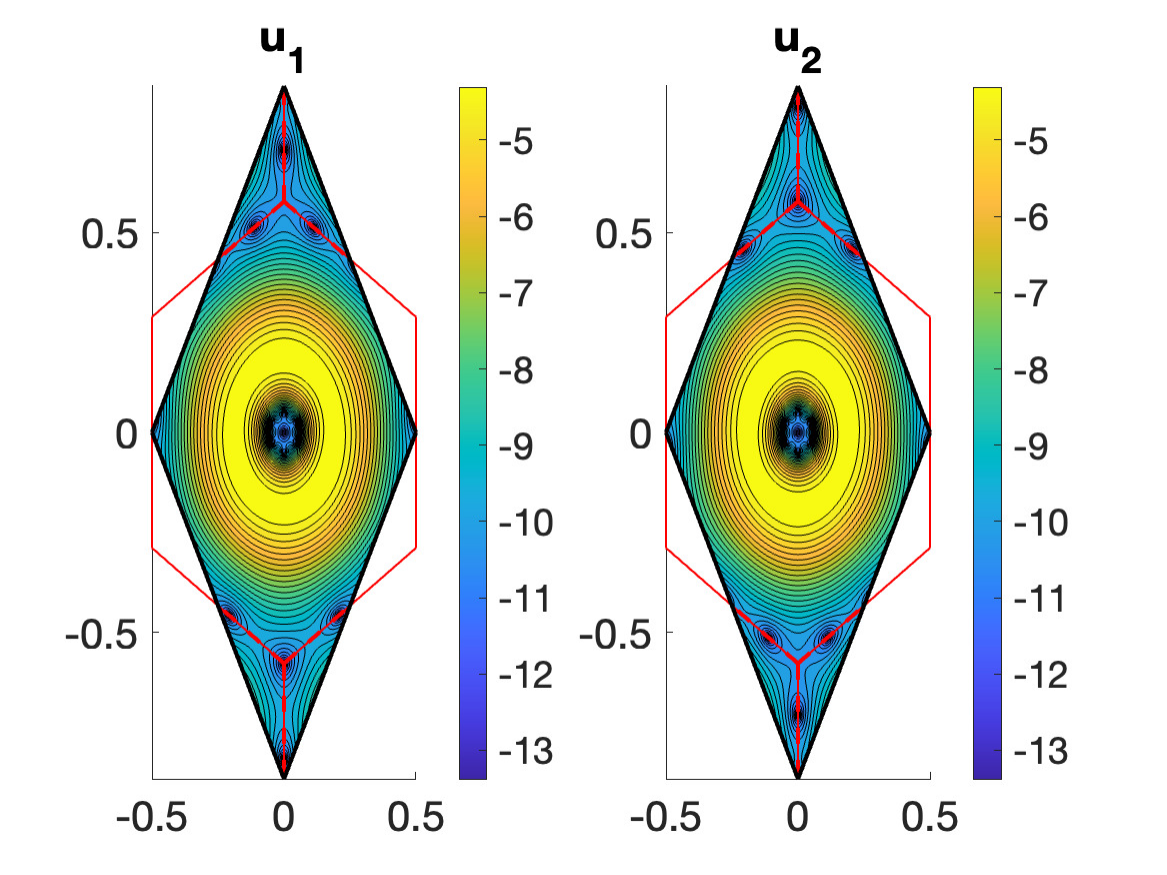}\\
\includegraphics[width=10.5cm,height=6.5cm]{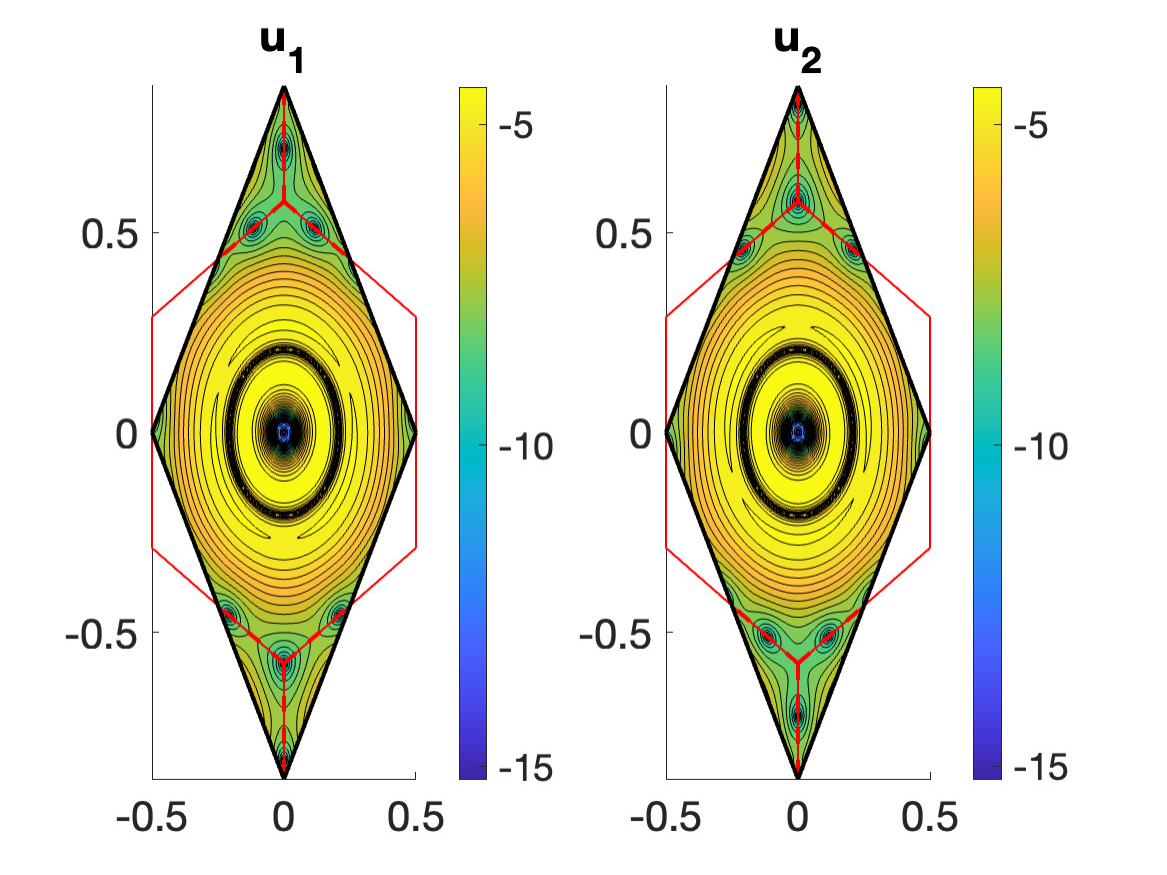}
\caption{One flat band for $\alpha=5.3811$ for $\theta =2.7672151$ and $U_{\pm}:=U_0(\pm \bullet)$ with $U_0(\zeta)= \cos(\theta) U_1(\zeta )+ \sin(\theta) \sum_{i=0}^2 \omega^i e^{-(\zeta \bar \omega^i - \bar \zeta \omega^i)}$. The eigenvector of $T_0$ with eigenvalue $1/\alpha$ is shown on top and the generalized one at the bottom.}
\label{fig:double2}
\end{figure}

\begin{center}
\begin{figure}
\includegraphics[width=10cm]{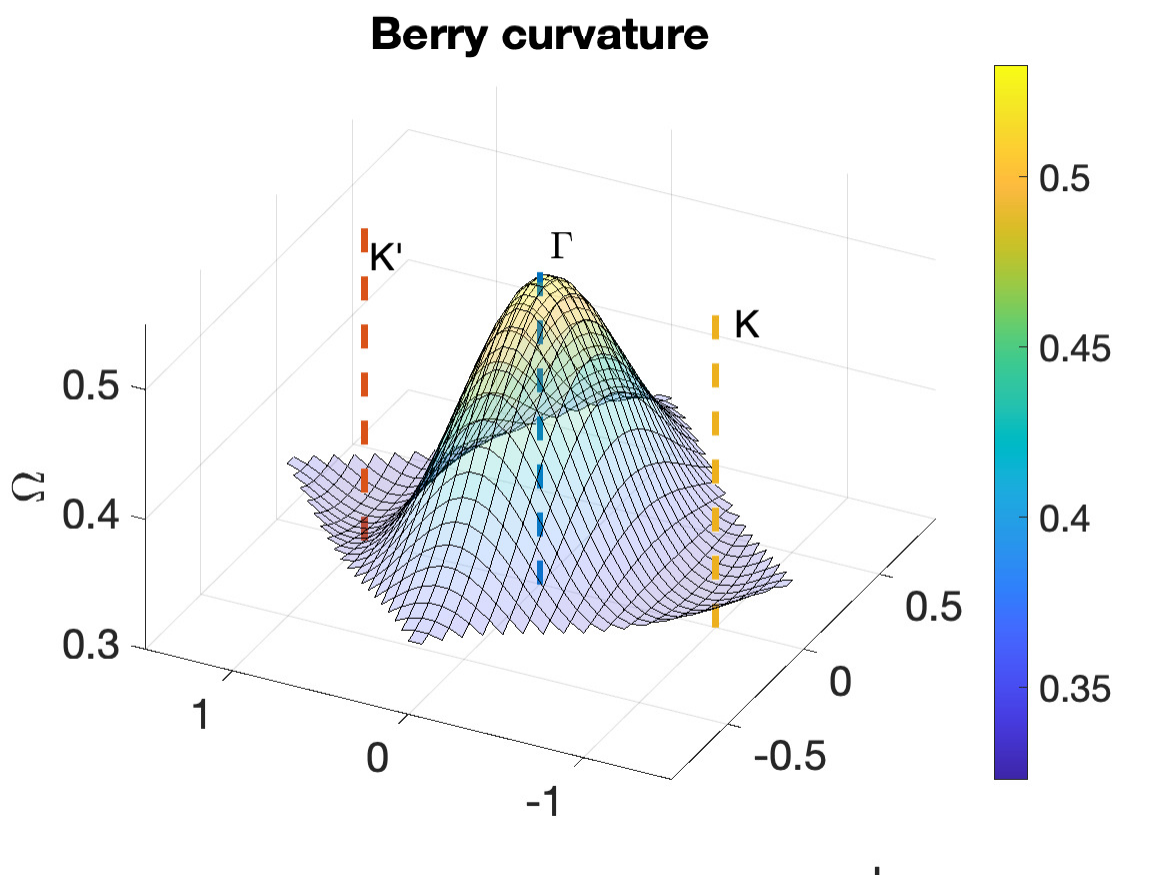}
\caption{\label{f:curv}  The plot of the curvature of the holomorphic line bundle 
corresponding to the first two-fold generate magic angle, defined in \eqref{eq:defH} with potential $U_{\pm}:=U_2(\pm \bullet)$, as in \eqref{eq:potential}. The extrema at $K,\Gamma,K'$ follow from Theorem \ref{t:Chern}.}
\end{figure}
\end{center}
\begin{figure}
\includegraphics[width=10cm]{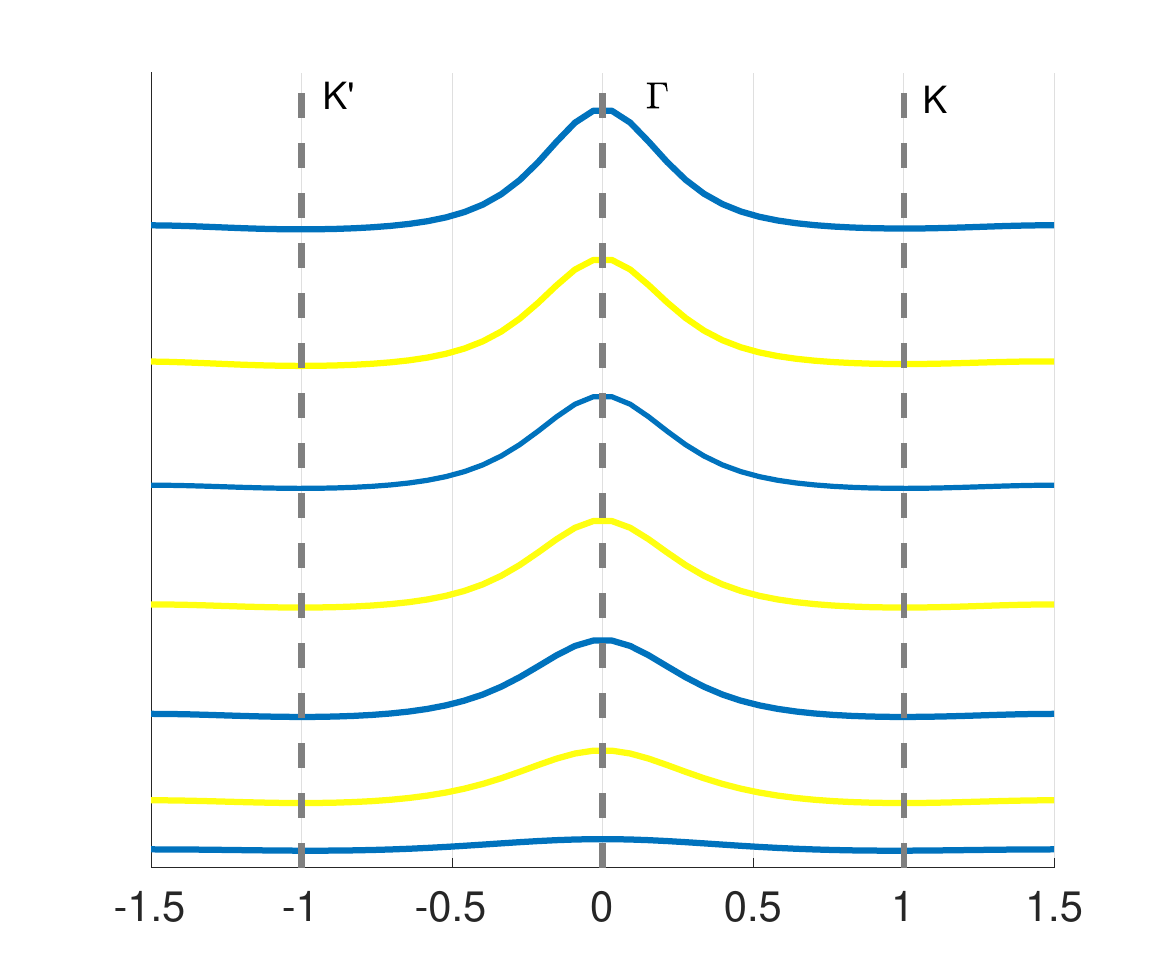}
\caption{\label{f:Ukraine} Cross-section of curvature for $k_x=0$ for the first seven magic angles with potential $U_{\pm}:=U_2(\pm \bullet)$, as in \eqref{eq:potential}, in increasing order. The extrema at $K,\Gamma,K'$ follow from Theorem \ref{t:Chern}.}
\end{figure}

\subsection{Behaviour of the curvature}

Since we established in Theorem \ref{t:Chern} that $ H ( \omega z ) = H ( z ) $ where $H$ is the scalar curvature. We conclude that $ 0 $ and $ \pm z_S $ are critical points of $ H $. In addition, the symmetry $\mathscr E$ defined in \eqref{eq:E} and the formula \eqref{eq:flat_band} imply that the Gramian matrix satisfies for simple or two-fold degenerate magic angles 
\[  G ( k ) = G ( - k ). \]
This implies the symmetries in Figure \ref{f:curv}.

However, while it seems that the maximum is attained at $\Gamma$ and the minima at $K,K'$, we do not have an analytical argument for this at the moment.

Figure \ref{f:Ukraine} shows that the standard deviation of the Berry curvature, for the potential $U_2$ with only two-fold degenerate real magic angles, increases monotonically for the real magic angles. This is in contrast to the case of simple magic angles in \cite[Figure $7$]{bhz2}.


\end{document}

%% file: alpha_table_1.tex

\begin{center}
\begin{tabular}{rclcc}
\multicolumn{1}{c}{$k$} & &
\multicolumn{1}{c}{$\alpha_k$} &
 & $\alpha_{k}-\alpha_{k-1}$ \\[2pt] \hline
1  &\ &   \phantom{0}0.585663 &\ &               \\
2  &&   \phantom{0}2.221182  && 1.6355        \\
3  &&   \phantom{0}3.751406  && 1.5302        \\
4  &&   \phantom{0}5.276498   && 1.5251        \\
5  &&   \phantom{0}6.794785    && 1.5183        \\
6  &&   \phantom{0}8.312999     && 1.5182        \\
7  &&   \phantom{0}9.829067     && 1.5161        \\
8  &&    11.345340       && 1.5163        \\
9  &&    12.860608        && 1.5153        \\
10 &&    14.376072         && 1.5155        \\
11 &&    15.890964          && 1.5149        \\
\end{tabular}
\end{center}

%% file: alpha_table_2.tex

\begin{center}
\begin{tabular}{rclcc}
\multicolumn{1}{c}{$k$} & &
\multicolumn{1}{c}{$\alpha_k$} &
 & $\alpha_{k}-\alpha_{k-1}$ \\[2pt] \hline
1  &\ &   \phantom{0}0.853799 &\ &               \\
2  &&     \phantom{0}2.691433  && 1.8376       \\
3  &&     \phantom{0}4.507960  && 1.8165        \\
4  &&     \phantom{0}6.332311  && 1.8244        \\
5  &&     \phantom{0}8.157130    && 1.8248        \\
6  &&     \phantom{0}9.983510     && 1.8264       \\
7  &&     11.809376      && 1.8259        \\
8  &&     13.635446      && 1.8261        \\
9  &&     15.460894        && 1.8255        \\
10 &&    17.286231        && 1.8253        \\
11 &&    19.111041 && 1.8248
\end{tabular}
\end{center}